\documentclass[a4paper,fleqn]{cas-sc}

\usepackage[authoryear,longnamesfirst]{natbib}

\def\tsc#1{\csdef{#1}{\textsc{\lowercase{#1}}\xspace}}
\tsc{WGM}
\tsc{QE}

\usepackage{amssymb}
\usepackage{amsmath}

\usepackage[T1]{fontenc}
\usepackage[utf8]{inputenc}
\usepackage{CJKutf8}
\usepackage{etoolbox}
\usepackage{xcolor, xspace}
\AtBeginEnvironment{thebibliography}{\begin{CJK}{UTF8}{gbsn}}
\AtEndEnvironment{thebibliography}{\end{CJK}}

\usepackage{amsthm}
\theoremstyle{plain}
\newtheorem{thm}{Theorem}[section]

\newtheorem{lem}[thm]{Lemma}
\newtheorem{cor}[thm]{Corollary}

\theoremstyle{definition}

\theoremstyle{remark}
\newtheorem{rem}[thm]{Remark}

\usepackage{algorithm}
\usepackage{algpseudocode}

\usepackage{booktabs}
\usepackage{threeparttable,multirow,array,tabularx,makecell}

\newcolumntype{L}[1]{>{\raggedright\arraybackslash}p{#1}}
\newcolumntype{Y}{>{\raggedright\arraybackslash}X}

\usepackage{subcaption}

\usepackage{physics}

\usepackage{tikz}
\usetikzlibrary{positioning,arrows.meta,calc,fit,matrix,decorations.pathreplacing}

\tikzset{
  frame/.style={draw=gray!35, fill=gray!2, rounded corners=2mm, thick},
  nodeg/.style={draw=gray!35, fill=white, rounded corners=2mm, thick, inner sep=5pt, align=center},
  nodeb/.style={draw=blue!65!black, fill=blue!8, rounded corners=2mm, very thick, inner sep=5pt, align=center},
  circ/.style={circle, draw=gray!35, fill=white, thick, inner sep=1.3pt},
  arrb/.style={-Latex, thick, draw=blue!65!black},
  arrg/.style={-Latex, thick, draw=gray!60},
  arrgr/.style={-Latex, thick, draw=gray!60, rounded corners=2mm},
}

\newif\ifshowcomments
\showcommentstrue

\usepackage{dsfont} 
\DeclareMathOperator{\indicator}{\mathds{1}}
\DeclareMathOperator{\expect}{\mathbb{E}}
\DeclareMathOperator{\prob}{\mathbb{P}}
\newcommand{\pconv}{\xrightarrow{p}}

\DeclareMathOperator{\Var}{Var}

\DeclareMathOperator{\Cov}{Cov}
 
\DeclareMathOperator{\diag}{diag}

\newcommand{\eqd}{\stackrel{d}{=}}

\newcommand{\R}{\mathbb{R}}
\newcommand{\N}{\mathbb{N}}

\newcommand{\bigO}{\mathcal O}

\newcommand*{\defeq}{\mathrel{\vcenter{\baselineskip0.5ex \lineskiplimit0pt
  \hbox{\scriptsize.}\hbox{\scriptsize.}}}=}

\begin{document}
\let\WriteBookmarks\relax
\def\floatpagepagefraction{1}
\def\textpagefraction{.001}

\shorttitle{RESAPLE}    

\shortauthors{Khan \& Franklin}  

\title [mode = title]{RESAPLE: An Approximate One-Step Restricted Likelihood Estimator of Spatial Dependence for Exploratory Spatial Analysis}  



%

\author[1,2]{Aditya Khan}[orcid=0009-0007-9798-5211]

\cormark[1]


\ead{adityakhan@cs.toronto.edu}


\credit{Conceptualisation, Methodology, Software, Validation, Formal analysis, Investigation, Writing - Original Draft, Writing - Review \& Editing, Visualisation}

\affiliation[1]{organization={Department of Computer Science, University of Toronto}, addressline={40 St George St}, postcode={M5S 2E4}, city={Toronto}, state={ON}, country={Canada}}

\author[2]{Meredith Franklin}[orcid=0000-0003-3802-8829]


\ead{meredith.franklin@utoronto.ca}


\credit{Writing - Review \& Editing, Supervision, Project Administration}

\affiliation[2]{organization={Department of Statistical Sciences, University of Toronto}, addressline={700 University Ave 9th Floor}, postcode={M5G 1Z5}, city={Toronto}, state={ON}, country={Canada}}

\cortext[1]{Corresponding author}



\begin{abstract}
Spatial autocorrelation diagnostics such as Moran's index and the approximate profile likelihood-based estimators (APLE) are widely used to assess spatial dependence in areal data. Yet, although Moran's index is frequently applied to regression residuals and APLE is typically formulated for raw outcomes, neither is explicitly derived as an estimator of residual spatial dependence after adjusting for large-scale trends and covariates. We propose RESAPLE, a one-step approximate restricted maximum likelihood (REML) estimator of the spatial error model's spatial dependence parameter $\rho$ based on REML residuals. Because RESAPLE is a Rayleigh coefficient, it retains the interpretability and diagnostic convenience of exploratory indices, while also providing a computationally inexpensive and accurate estimator of $\rho$ under moderate spatial dependence. We show theoretically that for small to medium sample sizes and well-specified trend models, RESAPLE is a better estimator of, and competitive test statistic for, residual spatial dependence relative to Moran's index and the APLE across a broad range of spatial configurations. The theory we develop also yields a diagnostic for spatial weight selection, providing guidance towards resolving a longstanding challenge in spatial data analysis. We illustrate the method through simulations on both regular and highly irregular lattices with a case study using American Community Survey tract-level data. 
\end{abstract}



\begin{keywords}
Spatial Autocorrelation \sep Exploratory Spatial Data Analysis \sep Spatial Error Model
\end{keywords}

\maketitle

\section{Introduction}
Spatial autoregressive (SAR) models, which include lag (SAR$_{lag}$) and error (SEM) specifications, are widely used to capture spatial dependence in areal (lattice) data across spatial statistics and econometrics \citep{zhang2018-sar-econ}, with growing applications in ecology \citep{verhoef2018-sar-ecol} and public health \citep{gao2021-sar-health}. Tests and estimators of spatial dependence derived from SAR and SEM models serve two complementary roles. First, they provide diagnostics for assessing model assumptions, especially violations of spatial independence, where tests of residual spatial autocorrelation are standard \citep{anselin1996}. Second, they underpin exploratory spatial data analysis (ESDA), by offering simple, interpretable, and computationally light summaries of spatial dependence or clustering in the data \citep{anselin1995}.

The classical statistic for summarising global spatial dependence is Moran's $\mathcal I_M$ \citep{moranNOTESCONTINUOUSSTOCHASTIC1950}. Its appeal lies in its closed form expression as a ratio of quadratic forms. However, extensive work has documented its shortcomings as an estimator of the spatial dependence parameter, $\rho$. $\mathcal I_M$ can be severely biased except when $\rho \approx 0$, in which case it does constitute a locally best invariant test for the parameter under appropriate conditions \citep{liMoransTestingSpatial2007}. Hence, Moran's $\mathcal I_M$ is excellent for detection but not for estimation. This limitation motivates alternatives that preserve the interpretability of Moran's $\mathcal I_M$ while more closely targeting the spatial dependence parameter, $\rho$.

A leading example is the approximate profile likelihood estimator (APLE) of $\rho$ \citep{liMoransTestingSpatial2007}, derived as the ``one-step'' (first-order) solution to the profile-likelihood score equation for $\rho$ in a SAR model. Like Moran's $\mathcal I_M$, APLE is a closed-form ratio of quadratic forms and is straightforward to compute, making it feasible for ESDA and large datasets. Implementations are available in standard \textbf{R} packages such as \texttt{spdep} \footnote{\url{https://www.rdocumentation.org/packages/spdep/versions/0.7-8/topics/aple}}. Under a Gaussian model, one can instead compute the maximum likelihood estimator for $\rho$, but that can be computationally burdensome (discussed in \cite{liMoransTestingSpatial2007}). Subsequent work has elaborated on APLE's distributional properties and extensions (see, for instance, \cite{jinApproximatedLikelihoodRoot2012, liOnestepEstimationSpatial2012, rederMoreAPLEStatistic2009}).

In this article we argue that despite substantial progress on APLE-type estimators, current methodology is misaligned with how practitioners actually analyse spatial regression models with covariates. In applied regression analyses, Moran's $\mathcal I_M$ and APLE are rarely applied directly to raw outcomes. Instead, they are typically computed from residuals after removing the mean structure or detrending, either 1) implicitly, by centring variables first (as in \cite{liMoransTestingSpatial2007}), or 2) explicitly, by allowing for regressors in the APLE formula (as in \cite{liOnestepEstimationSpatial2012}).

In other words, although APLE is derived from the full profile likelihood involving the both the regression coefficients $\beta$ and the spatial dependence parameter $\rho$, in practice $\beta$ is first treated as a nuisance parameter, typically estimated by least squares, and then spatial dependence is summarised using the resulting residuals. Moreover, APLE-type procedures with covariates can depend in non-trivial ways on how fixed effects are specified, including the particular choice of regressors and the resulting residualisation operator, even though ESDA summaries are typically intended to be stable across such choices.

Our key insight is that this tension between joint ML of all parameters ($\beta$, $\rho$, $\sigma^2$), and inference on the covariance parameter ($\rho$), is exactly the problem that REML was designed to solve in Gaussian linear mixed models \cite{harvilleMaximumLikelihoodApproaches1977, pattersonRecoveryInterblockInformation1971}. REML bases estimation of covariance parameters on linear combinations of the data that remove the fixed effects, yielding reduced bias and insensitivity to the parametrisation of the design matrix $X$. In the SAR/SEM setting, this perspective also yields a principled route to residual-based ESDA. Specifically, it allows spatial dependence to be summarized in the residual space after accounting for $X$, while preserving the computational and interpretive advantages of quadratic-form statistics.

Motivated by this, we propose RESAPLE, a REML-aligned, closed-form ESDA statistic for SEM models with covariates. RESAPLE is derived as an approximate one-step estimator from the restricted profile likelihood for $\rho$, and therefore aligns naturally with residual-based practice by construction. In addition, RESAPLE is invariant to reparameterizations of $X$, and its curvature at the null yields an analytic information quantity that can be used to compare candidate spatial weight matrices $W$ in terms of their local informativeness and power for detecting spatial dependence.

The remainder of this paper is organised as follows. In Section \ref{sec:model-likelihood}, we provide some context towards the setting we work with, as well as the model definition. We then define the RESAPLE statistic in Section \ref{sec:RESAPLE-estimator} and discuss some properties of the statistic. Next we discuss the primary use case of RESAPLE, as an ESDA statistic in Section \ref{sec:RESAPLE-esda}. This is followed by a discussion of how the choice of spatial weight matrix, $W$, impacts our analysis in Section \ref{sec:impact-W}. We then provide a simulation study and real data case study in Section \ref{sec:simulations} and \ref{sec:real-data} respectively to empirically validate our claims. We conclude in Section \ref{sec:conclusion}.

\section{Model Definitions and Background}\label{sec:model-likelihood}
We work under the Gaussian SEM. The purpose of this section is to set our notation and to introduce the residualised contrasts that underpin the definition of the RESAPLE estimator in Section \ref{sec:RESAPLE-estimator}. Let $n \in \N$ denote the number of spatial units, and let $p \in \{1,\dots,n-1\}$ denote the number of covariates.

\subsection{Notation for the SEM}\label{subsec:gaussian-sem}
The Gaussian SEM is given by
\begin{equation}\label{eq:sem}
Z = X\beta + U,\qquad U = \rho WU + \epsilon,\qquad \epsilon \sim N(0, \sigma^2 I_n).
\end{equation}
Here, $Z \in \R^n$ is the observed response, $X \in \R^{n \times p}$ is the design matrix representing fixed effects, and $\beta \in \R^p$ are regression coefficients. On the error side, $W \in \R^{n\times n}$ is a spatial weight matrix with zero diagonal that encodes the spatial topology of the data, where $w_{ij}$ encodes the (possibly asymmetric) proximity between locations $i$ and $j$.

The scalar spatial autoregressive parameter is $\rho \in \Theta \subseteq \R$, where $\Theta$ is taken to be a compact subset of $$\{\rho \in \R : \det R(\rho) \equiv \det(I_n - \rho W) \neq 0\}.$$ Writing $R(\rho)=I_n-\rho W$, the non-singularity condition ensures that the induced covariance structure is well defined. Further details, including the implied marginal distribution of $Z$ and the corresponding log-likelihood, are given in Appendix \ref{app:model-likelihood}.

\subsection{The Residual Space}\label{subsec:residual-space}
Our estimator is constructed in the residual space obtained by removing the contribution of $X$. Let $$P = X(X^\top X)^{-1}X^\top,\qquad M = I_n - P.$$ Since $X$ has full column rank $p$, $P$ is a rank $p$ orthogonal projector onto $\Im(X)$, and $M$ is the rank $n-p$ orthogonal projector onto $\Im(X)^\perp$. We refer to $\Im(M)$ as the residual space of the data, of dimension $r=n-p$.

Let $H\in\R^{n\times r}$ have orthonormal columns chosen to satisfy $$H^\top H = I_r,\qquad HH^\top = M,$$ so that the columns of $H$ form an orthonormal basis for $\Im(M)=\Im(X)^\perp$. We define the residualised contrasts $$e = H^\top Z \in \R^r.$$ Intuitively, $e$ contains all information in $Z$ that is orthogonal to the column space of $X$. The following lemma characterises its distribution under the SEM.

\begin{lem}\label{lem:error-dist}
Under the SEM model \eqref{eq:sem},
\begin{enumerate}
    \item $H^\top X = 0$ and hence $e = H^\top Z = H^\top U$.
    \item $e = H^\top R(\rho)^{-1}\epsilon$, and hence $e \sim N\big(0,\sigma^2 \Sigma_r(\rho)\big)$ with $$\Sigma_r(\rho) = H^\top R(\rho)^{-1}R(\rho)^{-\top}H.$$
\end{enumerate}
\begin{proof}
See Appendix \ref{app:model-likelihood}.
\end{proof}
\end{lem}

The distribution of $e$ is the starting point for the derivation of the RESAPLE. The key point is that this distribution is independent of $\beta$ and depends only on the covariance structure of $U$.

\subsection{Recap of the APLE estimator}\label{subsec:aple-recap}
Before deriving the RESAPLE estimator, we briefly recall the APLE estimator, starting with the case where there are no covariates. The estimator is given in \cite{liMoransTestingSpatial2007}: 
$$\hat \rho_{APLE} = \frac{Z^\top KZ}{Z^\top(W^\top W + \nu_n I_n)Z},\qquad \nu_n = \Tr(W^2)/n,\qquad K=(W+W^\top)/2.$$
This is a ratio of quadratic forms, and it is the one-step estimator $$\hat \rho_{APLE} = \rho_0 - \frac{f(\rho_0)}{f'(\rho_0)}\qquad (\rho_0 = 0),$$ where $f(\rho) = \pdv{\rho} \ell_M^p(\rho)$ and $\ell_M^p$ denotes the likelihood for the data $Z$ after profiling out $\sigma^2$. Because it is a one-step estimator, it provides a first-order approximation to the maximum likelihood estimator while remaining in closed form. Moreover, because it is a ratio of quadratic forms, the distributional and asymptotic theory for this statistic is well developed. In particular, under a Gaussian null hypothesis $H_0: \rho = 0$, the exact distribution of this statistic can be expressed in terms of a quadratic-form mixture of $\chi_1^2$ variables and is numerically evaluated using Imhof's method \citep{imhof1961}. In the general case where the null is $H_0: \rho = \rho_0$, \cite{liMoransTestingSpatial2007} show that a test based on APLE is first-order equivalent to a Score test and a Lagrange multiplier test, which are locally most powerful.

Later, \cite{liOnestepEstimationSpatial2012} generalise this statistic to account for covariates, defining the updated statistic (which we refer to as MAPLE for convenience) as $$\hat \rho_{MAPLE}= \frac{Z^\top MKMZ}{Z^\top\big(MW^\top WM - M(W^\top + W)P(W^\top W)M + \nu_n M\big)Z},$$ where the additional terms arise from incorporating covariates into the likelihood. However, despite these appealing properties, the MAPLE estimator remains an ML-based one-step approximation to the profile likelihood, and therefore it does not align with residual-based ESDA in the presence of fixed effects. In particular, it inherits several drawbacks that are precisely those REML was designed to address. The one-step profile likelihood underlying MAPLE treats $\beta$ as a nuisance parameter to be estimated by ML within the same objective. This differs from how APLE-type statistics are typically used in practice, where spatial dependence is summarised on regression residuals after removing fixed effects \citep{liMoransTestingSpatial2007}. This residual-based workflow corresponds more closely to the restricted likelihood perspective than to the full ML profile likelihood.

We address these issues by defining the RESAPLE estimator in the subsequent section, as a REML-aligned closed-form statistic constructed in the residual space.

\section{The RESAPLE Estimator}\label{sec:RESAPLE-estimator}
\subsection{Defining the Estimator}
We now define the RESAPLE estimator. The construction begins from the restricted profile likelihood for the residualised contrasts $e=H^\top Z$, which eliminates the regression coefficients $\beta$ and bases inference for $(\rho,\sigma^2)$ on the residual space. The corresponding restricted likelihood and its profiled form are given in Appendix \ref{app:RESAPLE-estimator}, together with all proofs for this section.

To motivate an estimator at $\rho=0$, we require the restricted profile score at the null.

\begin{lem}\label{lem:score}
Let $\ell_r^p(\rho;e)$ denote the restricted profile log-likelihood for $\rho$ based on $e$, and let $S_r(\rho) = \pdv{\rho}\ell_r^p(\rho;e).$
Then the following all hold.
\begin{enumerate}
    \item $\Sigma_r(0) = I_r$.
    \item $\Sigma_r'(0) = W_r + W_r^\top = 2K_r$, where $W_r = H^\top WH$ and $K_r = (W_r + W_r^\top)/2$.
    \item Hence the restricted profile score under no spatial dependence is
    \begin{equation}\label{eq:restricted-score}
    S_r(0) = r\frac{e^\top K_r e}{e^\top e} - \Tr(K_r)
    = r\Big(\frac{e^\top K_r e}{e^\top e} - \mu_r\Big),
    \end{equation}
    where $\mu_r = \Tr(K_r)/r$.
\end{enumerate}
\end{lem}

We now define the RESAPLE estimator. Writing
$$\mu_r = \frac{\Tr(K_r)}{r},\qquad \omega_r = \frac{\Tr(K_r^2)}{r},\qquad B_r = K_r^2+\omega_r I_r,$$
we define
$$\hat \rho_{RESAPLE} = \frac{e^\top(K_r - \mu_r I_r)e}{e^\top B_r e}.$$
This definition is motivated by a one-step approximation to the restricted score equation at $\rho=0$, with a denominator that is positive definite whenever $K_r\neq 0$. Details are given in Appendix \ref{app:RESAPLE-estimator}.

\begin{rem}
Notice that $S_r$ has Moran's $\mathcal I_M$ embedded into the formula, through $\frac{e^\top K_r e}{e^\top e}$, up to scaling by $r$, and centring by $\mu_r$. 
\end{rem}

\begin{thm}\label{thm:RESAPLE-def}
The following all hold for the RESAPLE estimator.
\begin{enumerate}
    \item The RESAPLE estimator is invariant to the choice of orthonormal basis $H$ satisfying $HH^\top = M$, $H^\top H = I_r$, and $H^\top X = 0$.
    \item Denote the null expectation under $\rho = 0$ as $\expect_0$. Then $\expect_0[e^\top(K_r - \mu_r I_r)e] = 0$. Moreover, if $K_r\neq0$, then $B_r=K_r^2+\omega_r I_r\succ0$ and
    $$\expect_0[e^\top B_r e] = 2\sigma^2\Tr(K_r^2)=\sigma^2\big(\Tr(W_r^\top W_r)+\Tr(W_r^2)\big).$$
    \item The RESAPLE estimator is the one-step solution to the approximate restricted score equation $$S_r(0) + \mathcal I_r^A(0)\rho = 0,$$ where
    $$\mathcal I_r^A(0) = -\frac{r}{e^\top e}e^\top B_r e$$
    is a stabilised approximate restricted Fisher information (curvature) motivated by the restricted Fisher information at $\rho = 0$ with $\sigma^2$ plugged in by $\hat\sigma_r^2(0)=\frac1r e^\top e$.
    \item Consequently, $$\hat \rho_{RESAPLE} = - \frac{S_r(0)}{\mathcal I_r^A(0)}.$$ 
\end{enumerate}
\end{thm}

It is important to note that the standard APLE is recovered as a special case of the RESAPLE estimator when $X$ is absent and the residualised spatial weight matrix is symmetric. When $X$ is absent ($p=0$), we have $M = I_n$. Taking $H = I_n$ yields $r = n$, $e = Z$, $W_r = W$, and $K_r=K$. If, in addition, $W$ is symmetric, then $K=W$, $\mu_r=\Tr(W)/n=0$ because $W$ has zero diagonal, and
$$B_r = W^2+\frac{\Tr(W^2)}{n}I_n = W^\top W+\frac{\Tr(W^2)}{n}I_n.$$
Therefore
$$\hat \rho_{RESAPLE} = \frac{Z^\top K Z}{Z^\top(W^\top W + \nu_n I_n)Z},\qquad \nu_n=\frac{\Tr(W^2)}{n},$$
which is exactly the APLE proposed in \cite{liMoransTestingSpatial2007}. For non-symmetric weights, RESAPLE remains an APLE-type one-step statistic in the residual space, but uses the positive-definite stabilised denominator $B_r=K_r^2+\omega_r I_r$ rather than the APLE-style denominator $W_r^\top W_r+\Tr(W_r^2)I_r/r$. Finally, we note that RESAPLE depends on the data only through the error contrasts $e$ and is invariant to reparameterisations of the design matrix $X$. We state this as a corollary of the first point of Theorem \ref{thm:RESAPLE-def}.

\begin{cor}\label{cor:design-invariance}
Let $\tilde X = XQ$ for any invertible $Q \in \R^{p\times p}$, and construct $\tilde H$, $\tilde e$, $\tilde W_r$, and $\tilde K_r$ accordingly. Then $$\hat \rho_{RESAPLE}(X) = \hat \rho_{RESAPLE}(\tilde X).$$
\end{cor}

\subsection{Properties Inherited From One-Step Estimation}\label{subsec:one-step}
It is worth emphasising that one-step estimators are well-studied, and hence, the RESAPLE inherits a number of desirable properties by virtue of being one (indeed, shared by the APLE and covariate-adjusted MAPLE estimators as well). We discuss some properties of the estimator, following \cite{vaartAsymptoticStatistics1998}.

Recall that under the standard regularity conditions for likelihood-based estimation, a one-step estimator is asymptotically linear, and it has the same first-order asymptotic distribution as the likelihood estimator that solves the corresponding score equation. In our setting,
$$\hat\rho_{RESAPLE}=-\frac{S_r(0)}{\mathcal I_r^A(0)},$$
where $\mathcal I_r^A(0)$ is the stabilised approximate restricted curvature used for the one-step update at $\rho=0$ (Theorem \ref{thm:RESAPLE-def}). Moreover, under mild conditions (Theorem \ref{thm:curvature-consistency}), $$-\frac{\mathcal I_r^A(0)}{\mathcal I_r(0)} \pconv 1,$$
where $S_r(0)$ is defined in equation \eqref{eq:restricted-score} and the restricted information at $\rho = 0$ is $$\mathcal I_r(0) = 2\Tr(K_r^2)=\Tr(W_r^\top W_r) + \Tr(W_r^2)$$ (see Lemma \ref{lem:restricted-info}). Consequently, when $\rho$ lies in a neighbourhood of $0$ that shrinks at the $\sqrt r$ scale ($r \defeq n-p$), the RESAPLE admits a first-order expansion of the form
$$\hat\rho_{RESAPLE} = -\frac{S_r(0)}{\mathcal I_r(0)} + o_p(r^{-1/2}).$$
This immediately implies asymptotic normality if $\rho=0$ (the reader might recognise this as a null assumption of no spatial dependence) and under local alternatives of order $r^{-1/2}$. We note that the APLE literature adopts the same one-step perspective for spatial dependence parameters \citep{liMoransTestingSpatial2007, liOnestepEstimationSpatial2012}.

Second, one-step estimators inherit bias properties from Taylor expansion (see for instance \cite{vaartAsymptoticStatistics1998}). Specifically, the error in the Taylor expansion comes from higher-order derivatives of the score, while the stabilised denominator is first-order equivalent to the restricted Fisher information under the conditions above. Therefore, the local approximation error (and hence local bias) of the RESAPLE around a linearisation point of $\rho=0$ is at most \emph{quadratic} in $|\rho|$ in the $\sqrt{r}$-scaling. A quadratic error is negligible when $\rho$ is small. In other words, one can treat the RESAPLE as locally \emph{approximately} unbiased for $\rho$ around $\rho = 0$.

It is worth bearing in mind that just as we discuss these properties of the RESAPLE, they are equally applicable to likelihood-based one-step estimators that use consistent curvature normalisations. Indeed, the above arguments imply that these estimators are asymptotically equivalent to full REML and ML estimation procedures for $\rho$ in a first-order neighbourhood around $\rho = 0$ under a $\sqrt{r}$-scaling. The main difference comes from the \emph{finite-sample} case. It is in this case that the RESAPLE has the benefits that we describe earlier in this paper.

\begin{figure}[pos=htp]
\centering

\resizebox{\linewidth}{!}{%
\begin{tikzpicture}[
  font=\small,
  >={Latex[length=2.2mm]},
  arr/.style={->, thick, draw=blue!65!black},
  carr/.style={->, thick, draw=gray!60},
  box/.style={rounded corners=2mm, draw=gray!35, thick, align=center, inner sep=5pt, fill=white, text width=30mm},
  op/.style={rounded corners=2mm, draw=blue!65!black, very thick, align=center, inner sep=5pt, fill=blue!8, text width=46mm},
  wide/.style={rounded corners=2mm, draw=gray!35, thick, align=center, inner sep=5pt, fill=white, text width=56mm},
  gbox/.style={draw=gray!35, thick, rounded corners=2mm, inner sep=6pt, fill=gray!2},
  subspace/.style={draw=gray!35, thick, rounded corners=1.5mm, align=center, inner sep=4pt, fill=white}
]

\node[box] (Z) {$Z\in\R^n$\\ observed response};
\node[box, below=10mm of Z] (X) {$X\in\R^{n\times p}$\\ fixed effects};
\node[box, below=10mm of X] (W) {$W\in\R^{n\times n}$\\ spatial weights\\ (diag $0$)};

\node[gbox, right=18mm of X, minimum width=64mm, minimum height=54mm] (geom) {};
\node[anchor=north, font=\bfseries, text=gray!65] at ($(geom.north)+(0,-2mm)$) {decomposition};

\node[
  draw=gray!35, rounded corners=1.5mm, thick,
  minimum width=60mm, minimum height=40mm, fill=gray!1
] (Rn) at ($(geom.center)+(0,-5mm)$) {};
\node[anchor=north west, text=gray!65] at ($(Rn.north west)+(2mm,-1.5mm)$) {$\R^n$};

\node[subspace, fill=blue!6, minimum width=56mm, minimum height=12mm]
  (imX) at ($(Rn.center)+(0,7mm)$)
  {$\Im(X)$ (fixed-effect space)};

\node[subspace, fill=gray!6, minimum width=56mm, minimum height=18mm]
  (imM) at ($(Rn.center)+(0,-10mm)$)
  {$\Im(M)=\Im(X)^\perp$\\ residual space, $r=n-p$};

\node[op, right=18mm of geom, yshift=12mm] (proj)
{$P=X(X^\top X)^{-1}X^\top$\\ $M=I-P$};

\node[op, below=10mm of proj] (Hdef)
{$H\in\R^{n\times r}$\\ $H^\top H=I_r,\ HH^\top=M$\\ $e=H^\top Z$};

\node[op, below=10mm of Hdef] (Wr)
{$W_r=H^\top W H$\\ $K_r=(W_r+W_r^\top)/2$};

\node[wide, below=10mm of Wr] (info)
{$\mathcal I_r(0)=2\Tr(K_r^2)$\\
$=\Tr(W_r^\top W_r)+\Tr(W_r^2)$};

\node[wide, below=10mm of info] (RESAPLE)
{\vspace{0.5mm}
$\displaystyle \hat\rho_{RESAPLE}
=\frac{e^\top (K_r-\mu_r I_r)e}{e^\top B_r e}$\\[1mm]
$B_r=K_r^2+\omega_r I_r$\\
$\mu_r=\Tr(K_r)/r,\ \omega_r=\Tr(K_r^2)/r$
\vspace{0.5mm}};

\node[anchor=west, font=\footnotesize\bfseries, text=gray!65] at ($(proj.east)+(3mm,0)$) {project out covariates};
\node[anchor=west, font=\footnotesize\bfseries, text=gray!65] at ($(Hdef.east)+(3mm,0)$) {residualise data};
\node[anchor=west, font=\footnotesize\bfseries, text=gray!65] at ($(Wr.east)+(3mm,0)$) {residualise weights};
\node[anchor=west, font=\footnotesize\bfseries, text=gray!65] at ($(info.east)+(3mm,0)$) {compute restricted null information};
\node[anchor=west, font=\footnotesize\bfseries, text=gray!65] at ($(RESAPLE.east)+(3mm,0)$) {compute RESAPLE};

\coordinate (branch) at ($(RESAPLE.south)+(0,-8mm)$);

\node[box, text width=34mm] (scat)  at ($(branch)+(0,-18mm)$) {RESAPLE\\ scatterplot};
\node[box, text width=34mm] (test)  at ($(scat.west)+(-46mm,0)$) {Testing\\ (exact / perm / $z$)};
\node[box, text width=34mm] (local) at ($(scat.east)+(46mm,0)$) {Local\\ contributions};

\draw[arr] (Z.east) -- ($(geom.west)+(0,12mm)$);
\draw[arr] (X.east) -- (geom.west);
\draw[arr] (W.east) -- ($(geom.west)+(0,-12mm)$);

\draw[carr] (geom.east) -- (proj.west);

\draw[carr] (proj.south) -- (Hdef.north);
\draw[carr] (Hdef.south) -- (Wr.north);
\draw[arr]  (Wr.south) -- (info.north);
\draw[arr]  (info.south) -- (RESAPLE.north);

\draw[carr] (W.east) -- ++(16mm,0) |- (Wr.west);

\draw[arr] (RESAPLE.south) -- (branch);
\draw[arr] (branch) -- (scat.north);
\draw[arr] (branch) -| (test.north);
\draw[arr] (branch) -| (local.north);

\end{tikzpicture}%
}

\caption{Overall workflow for using RESAPLE as a residual-space diagnostic for spatial dependence, in ESDA.}
\label{fig:RESAPLE-overview}
\end{figure}

\section{RESAPLE in ESDA}\label{sec:RESAPLE-esda}
While in the previous sections, we discuss some of the baseline properties of the RESAPLE, those results primarily serve as sanity checks that the estimator we propose is well behaved. The main use case of this statistic is in ESDA, as summarized in Figure \ref{fig:RESAPLE-overview}. We discuss three ESDA use cases of RESAPLE: (i) testing for autocorrelation (as a alternative to Moran's $\mathcal I_M$), (ii) constructing a RESAPLE scatterplot analogous to the Moran's $\mathcal I_M$ scatterplot, and (iii) defining a local variant of RESAPLE to diagnose local autocorrelation. A summary of the workflow is provided in Figure \ref{fig:RESAPLE-overview}.

\subsection{Testing with the RESAPLE}\label{subsec:RESAPLE-testing}
In many applications, the primary inferential task is to test for the presence of residual spatial autocorrelation after adjusting for covariates. In the Gaussian SEM, this leads to a hypothesis test for the spatial dependence parameter
$$H_0:\rho = 0 \qquad\text{versus}\qquad H_1:\rho > 0,$$
where the one-sided alternative reflects the usual interest in testing for positive spatial autocorrelation between neighbouring areal units. In general, practitioners default to Moran's $\mathcal I_M$ computed on unadjusted data or on regression residuals. For Gaussian SAR or SEM models without covariates, and under suitable regularity conditions, it is rather well-known that $\mathcal I_M$ yields a locally most powerful invariant test for $H_0:\rho=0$ against local alternatives in which $\rho$ is small and positive (discussed for instance in \cite{liMoransTestingSpatial2007}). However, it is well documented that Moran's $\mathcal I_M$ is severely biased as an estimator of $\rho$ away from the null and that its finite-sample distribution can differ from its asymptotic approximation, especially when the residual dimension is small or when adjacency induces non-regular spectral structure \cite{rederMoreAPLEStatistic2009}. In particular, when covariates are present, the usual procedure of computing Moran's $\mathcal I_M$ on residuals obtained from a fixed effects model is not aligned with the likelihood for $\rho$ and does not account for the contribution of $X$ to the covariance structure.

The RESAPLE framework we propose provides a natural alternative test statistic that is constructed in the residual space and is directly linked to the restricted likelihood. In what follows (Figure \ref{fig:RESAPLE-overview}), we regard the RESAPLE estimator as the basic test statistic for $H_0:\rho=0$. There are three potential ways to test with $\hat \rho_{RESAPLE}$:
\begin{enumerate}
    \item exact testing,
    \item Freedman-Lane permutation testing, and
    \item (approximate) $z$-testing.
\end{enumerate}
We discuss each method and their benefits and drawbacks in turn.

\paragraph{Exact testing.}
Assuming the Gaussian SEM as we have throughout, there is an exact reference law for the RESAPLE statistic. Indeed, suppose the null of $H_0:\rho=0$. Lemma \ref{lem:error-dist} implies $e=H^\top Z\sim N(0,\sigma^2 I_r)$ with $r=n-p$. Under this null, $\hat\rho_{RESAPLE}$ is a ratio of quadratic forms in a Gaussian vector, with positive denominator whenever $K_r\neq0$. Thus its null distribution can be evaluated using classical methods for quadratic forms in normal variables. In particular, for any observed value $t\in\R$, define $$D_t \defeq K_r-\mu_r I_r - tB_r.$$ Then
\begin{equation}\label{eq:RESAPLE-qf-event}
\prob_0\big(\hat\rho_{RESAPLE}\geq t\big) = \prob_0\big(e^\top D_t e \geq 0\big).
\end{equation}
The matrix $D_t$ is deterministic and symmetric. After diagonalising $D_t$, the quadratic form $e^\top D_t e$ has the same distribution as a finite weighted sum of independent $\chi^2_1$ variables. Numerical computation of the tails is standard in the literature, for instance by Imhof's method \citep{imhof1961}. This gives an exact test under the Gaussian null.

A key robustness observation is that the exact same reference law also holds for a spherically symmetric extension to the standard Gaussian SEM. Start by replacing the Gaussian innovation in \eqref{eq:sem} by $\sigma\xi$, where $\xi$ is spherically symmetric. Thus
$$Z=X\beta+U,\qquad U=\rho WU+\sigma\xi.$$
Under $H_0:\rho=0$, this gives $Z=X\beta+\sigma\xi$ and hence $e=\sigma H^\top\xi$. The statistic $\hat\rho_{RESAPLE}$ is invariant to rescaling of $e$, so it depends on $e$ only through its direction. Since the direction of $H^\top\xi$ is uniform on the unit sphere in $\R^r$, the null distribution is the same as in the Gaussian case. This is the residual-space version of the invariant spherical testing argument in \citet{kingRobustTestsSpherical1980}, and it is consistent with the invariant formulation of residual spatial autocorrelation tests in linear regression \citep{martellosio2010power}. The formal result is given in Appendix \ref{app:RESAPLE-testing}. The exact test therefore does not require Gaussian tails under the null as the Gaussian SEM would suppose, but it does require spherical innovations in the SEM. This encompasses a number of heavier-tailed distributions than the multivariate Normal, including the multivariate-$t$, or in general, scale-mixtures of Gaussians (which may also yield lighter tails than a Gaussian). Using this test requires an eigendecomposition of a matrix of dimension $r$, and its numerical cost can therefore be non-trivial when $r$ is large.

\paragraph{Freedman-Lane permutation testing.}
If we wish to avoid over-reliance on this distributional assumption, we can instead use a residual-based randomisation approach, per \citet{freedman-lane-1983}. Under $H_0:\rho=0$, the mean structure is non-spatial and the fixed effects can be removed by ordinary least squares. Let $$\hat Z=PZ,\qquad \hat u=MZ.$$ For each randomly sampled permutation matrix $\Pi_b\in\R^{n\times n}$, the Freedman-Lane procedure constructs the pseudo-response
$$Z^{(b)}=\hat Z+\Pi_b\hat u.$$
The RESAPLE statistic is then recomputed on $Z^{(b)}$ using the same design matrix $X$ and the same spatial weight matrix $W$. Equivalently, with $e^{(b)}=H^\top Z^{(b)}$, we compute
$$\hat\rho^{(b)}_{RESAPLE}=\frac{(e^{(b)})^\top(K_r-\mu_r I_r)e^{(b)}}{(e^{(b)})^\top B_r e^{(b)}}.$$
For the one-sided alternative $H_1:\rho>0$, the permutation p-value is
$$\hat p_{FL}=\frac{1+\sum_{b=1}^L \indicator\{\hat\rho^{(b)}_{RESAPLE}\geq \hat\rho_{RESAPLE}\}}{L+1}.$$
For a two-sided test, one replaces the event above by $|\hat\rho^{(b)}_{RESAPLE}|\geq |\hat\rho_{RESAPLE}|$. This procedure preserves the fitted large-scale mean while randomising the residual spatial arrangement. It is therefore well aligned with residual-based ESDA, and it avoids defining the randomisation operation through an arbitrary coordinate system for the residual space. The method is straightforward to implement and avoids heavy eigendecompositions. It can, however, require many permutations for stable tail probability estimation.

\paragraph{Approximate $z$-testing.}
Finally, one may appeal to the asymptotic normal approximation implied by the one-step construction, discussed in Section \ref{subsec:one-step}. This motivates us to use the test statistic $$Z_{RESAPLE} = \sqrt{\mathcal I_r(0)}\hat\rho_{RESAPLE},\qquad \mathcal I_r(0)=\Tr(W_r^\top W_r)+\Tr(W_r^2),$$ and compare $Z_{RESAPLE}$ to the upper tail of a standard normal distribution under $H_0$. This method is of course, computationally simple. 

Each of these three methods has their benefits and drawbacks. It is worth briefly commenting on the alternatives to RESAPLE. As has been the theme in this paper, tests based on RESAPLE are better aligned with covariate adjusted workflows than APLE or MAPLE. The main concern the reader may have though, is why not a simply use a likelihood ratio statistic derived from the restricted likelihood for the SEM, which compares the restricted likelihood maximised under $\rho=0$ to that maximised over $\rho\in\Theta$ \citep{anselinLagrangeMultiplierTest1988}? It is true that this approach is appealing when one intends to fit the SEM by REML in any case. It can, however, be computationally burdensome in large problems, since it requires repeated evaluation of $\log|R(\rho)|$ and numerical optimisation over $\rho$. In the ESDA setting where one seeks an exploratory check, RESAPLE remains low-cost while just as accurate as REML asymptotically, and accurate for small $\rho$ (which is typically the regime one would wish to test autocorrelation for--when autocorrelation is not exactly obvious).

\subsection{Scatterplot Diagnostics}\label{subsec:RESAPLE-scatter}
In this subsection we develop a RESAPLE analogue of the Moran scatterplot and of the APLE scatterplot, as illustrated schematically in Figure \ref{fig:RESAPLE-scatter-schematic}. The main idea here is that statistics of the form 
$$\frac{e^\top A e}{e^\top B e},$$ 
with $A$ symmetric and $B$ positive definite, define a generalised Rayleigh quotient. Such a ratio admits a canonical regression interpretation after a whitening transformation determined by $B$. This yields a scatterplot whose ordinary least squares slope through the origin equals the ratio itself, and whose residuals provide a diagnostic for how well the implied one-dimensional summary represents the transformed contrasts. When applied to Moran-type statistics, this reproduces the classical Moran scatterplot \citep{anselin1996}, as with the APLE scatterplot when applied to APLE statistics \citep{liMoransTestingSpatial2007}. Of course, when we apply it to the RESAPLE, we continue to get a scatterplot whose slope equals $\hat\rho_{RESAPLE}$, but the key distinction is that the axes of the scatterplot are defined intrinsically in the restricted residual space.

To construct the scatterplot, recall that 
$$\hat\rho_{RESAPLE}=\frac{e^\top A_r e}{e^\top B_r e},\qquad A_r \defeq K_r-\mu_r I_r,\qquad B_r \defeq K_r^2+\omega_r I_r.$$ 
Then $A_r$ is symmetric and, whenever $K_r\neq0$, $B_r$ is symmetric positive definite. Define the whitened contrasts $x,y\in\R^r$ by
\begin{equation}\label{eq:RESAPLE-whitened}
x \defeq B_r^{1/2} e, \qquad y \defeq B_r^{-1/2}A_r e.
\end{equation}
Then $$x^\top x = e^\top B_r e, \qquad x^\top y = e^\top A_r e.$$ 
A straightforward calculation shows then that the ordinary least squares slope through the origin in the regression of $y$ on $x$ equals 
$$\frac{x^\top y}{x^\top x}=\frac{e^\top A_r e}{e^\top B_r e}=\hat\rho_{RESAPLE},$$ which is the standard regression interpretation used in Moran and APLE scatterplots \citep{anselin1996, liMoransTestingSpatial2007}. The key idea here is that the vectors $x$ and $y$ live in the residual space $\R^r$ and \emph{not} in the fixed-effect space. To obtain unit-level coordinates indexed by the original spatial units, define
\begin{equation}\label{eq:RESAPLE-tilde-xy}
\tilde x \defeq Hx = H B_r^{1/2} e,
\qquad
\tilde y \defeq Hy = H B_r^{-1/2}A_r e.
\end{equation}
We refer to the components $(\tilde x_i,\tilde y_i)$ for $i\in\{1,\dots,n\}$ as the plotted coordinates. Since $H^\top H=I_r$, these satisfy $$\sum_{i=1}^n \tilde x_i^2 = \tilde x^\top \tilde x = x^\top x = e^\top B_r e, \qquad \sum_{i=1}^n \tilde x_i\tilde y_i = \tilde x^\top \tilde y = x^\top y = e^\top A_r e. $$ 
Consequently,
\begin{equation}\label{eq:RESAPLE-contrib}
\hat\rho_{RESAPLE}=\frac{\sum_{i=1}^n \tilde x_i\tilde y_i}{\sum_{i=1}^n \tilde x_i^2}.
\end{equation}
In other words, equation \eqref{eq:RESAPLE-contrib} provides a direct decomposition of $\hat\rho_{RESAPLE}$ into unit-level contributions $\tilde x_i\tilde y_i$, with weights determined by $\tilde x_i^2$ (Figure \ref{fig:RESAPLE-scatter-schematic}).

\begin{rem}
Although the formula uses an orthonormal basis $H$ for the residual space, the resulting plotted coordinates are not a function of a choice of basis. Indeed, if $\tilde H=HQ$ for an orthogonal $r\times r$ matrix $Q$, then $\tilde e=Q^\top e$, $\tilde A_r=Q^\top A_rQ$, and $\tilde B_r=Q^\top B_rQ$. Taking $B_r^{1/2}$ and $B_r^{-1/2}$ to denote the principal symmetric matrix square roots gives $\tilde B_r^{1/2}=Q^\top B_r^{1/2}Q$ and $\tilde B_r^{-1/2}=Q^\top B_r^{-1/2}Q$. Hence the transformed residual-space coordinates rotate as $\tilde x=Q^\top x$ and $\tilde y=Q^\top y$, while the mapped unit-level coordinates satisfy $\tilde H\tilde x=Hx$ and $\tilde H\tilde y=Hy$. Thus the scatterplot is basis-invariant. In computation, we take $H$ to be the orthonormal complement of $X$ obtained from a QR decomposition.
\end{rem}

It is worth noting that RESAPLE scatterplot provides information that Moran and APLE scatterplots do not provide when covariates are present. First, the construction in \eqref{eq:RESAPLE-tilde-xy} depends on $(X,W)$ through the restricted operators $(W_r,K_r)$, so the plot reflects the interaction between the spatial topology and the residual space after adjusting for the covariates, $X$. This differs from the Moran scatterplot computed on raw data, and (more importantly) it also differs from the practice of plotting a residual against a spatial lag of that residual, which does not incorporate the restricted operators that come from the restricted likelihood. Second, the denominator $e^\top B_r e$ in \eqref{eq:RESAPLE-contrib} is obtained from the RESAPLE construction and connects to the curvature quantities such as $\mathcal{I}_r(0)$. The terms in the denominator of \eqref{eq:RESAPLE-contrib} (i.e. $\tilde x_i^2$), therefore identify spatial units that contribute strongly to this curvature in the residual space. This yields a diagnostic for where the information about $\rho$ concentrates after adjusting for $X$, which the Moran and APLE scatterplots do not directly display.

\begin{figure}[pos=htp]
\centering
\resizebox{\linewidth}{!}{%
\begin{tikzpicture}[
  x=1mm, y=-1mm,
  font=\small,
  >={Latex[length=2.2mm]},
  arr/.style={->, thick, draw=blue!65!black},
  carr/.style={->, thick, draw=purple!70!black},
  frame/.style={draw=gray!35, fill=gray!2, rounded corners=2mm, thick},
  box/.style={rounded corners=2mm, draw, thick, align=center, inner sep=5pt, fill=gray!4},
  op/.style={rounded corners=2mm, draw, thick, align=center, inner sep=5pt, fill=blue!6},
  note/.style={rounded corners=2mm, draw, thick, align=left, inner sep=4.5pt, fill=white, font=\scriptsize},
  grid/.style={draw=gray!18},
  axis/.style={draw=gray!65, thick},
  quad/.style={draw=gray!45, thick, dashed},
  pt/.style={circle, draw=blue!65!black, fill=blue!6, thick, inner sep=0.9pt},
  lfit/.style={draw=blue!65!black, very thick}
]

\path[frame] (0,0) rectangle (180,74);
\node[font=\scriptsize, text=gray!55] at (90,6) {RESAPLE scatterplot diagnostic};

\node[op, anchor=north west, text width=54mm] (def) at (6,12) {%
\textbf{Definition}\\[0.5mm]
$\displaystyle \hat\rho_{RESAPLE}=\frac{e^\top A_r e}{e^\top B_r e}$
};

\node[box, anchor=north west, text width=54mm] (coords) at (6,30) {%
\textbf{Plotted points}\\[0.5mm]
$\tilde x=HB_r^{1/2}e,\quad \tilde y=HB_r^{-1/2}A_r e$
};

\node[box, anchor=north west, text width=54mm] (contrib) at (6,48) {%
\textbf{Contribution}\\[0.5mm]
$C_i=\tilde x_i\tilde y_i$\qquad (weight $\tilde x_i^2$)
};

\begin{scope}[shift={(68,12)}]
  \path[frame] (0,0) rectangle (72,54);
  \node[font=\bfseries, text=gray!65] at (36,10) {RESAPLE scatterplot};

  \def\xL{8}   \def\xR{66}
  \def\yT{14}  \def\yB{48}
  \def\x0{37}
  \def\y0{31}

  \foreach \x in {12,20,...,64} \draw[grid] (\x,\yT) -- (\x,\yB);
  \foreach \y in {18,24,...,46} \draw[grid] (\xL,\y) -- (\xR,\y);

  \draw[axis, -Latex] (\xL,\y0) -- (\xR+3,\y0);
  \draw[axis, -Latex] (\x0,\yB) -- (\x0,\yT-3);

  \draw[quad] (\x0,\yT) -- (\x0,\yB);
  \draw[quad] (\xL,\y0) -- (\xR,\y0);

  \node[font=\scriptsize, text=gray!65, anchor=north] at (\x0,\yB+8) {$\tilde x_i$};
  \node[font=\scriptsize, text=gray!65, rotate=90, anchor=south] at (\xL-8,\y0) {$\tilde y_i$};

  \draw[lfit]
    (14,45) -- (60,19)
    node[
      pos=0.60, sloped, below,
      font=\scriptsize,
      text=blue!65!black,
      fill=white, fill opacity=0.75, text opacity=1,
      inner sep=1pt
    ] {$\text{slope}=\hat\rho_{\mathrm{RESAPLE}}$};

  \node[pt] (pL) at (62,36) {};

  \foreach \p in {(18,22),(24,28),(30,24),(34,32),(38,28),
                 (46,26),(50,22),(54,30),(58,30),
                 (28,42),(48,40),(60,42)}
    \node[pt] at \p {};

  \begin{scope}[
    qnote/.style={
      note,
      align=center,
      inner sep=2.0pt,
      fill=white, fill opacity=0.85, text opacity=1,
      font=\scriptsize
    }
  ]
    \node[qnote, anchor=north west] at (\xL+0.2,\yT+0.2) {\textbf{QII}\;$C_i<0$};
    \node[qnote, anchor=north east] at (\xR-0.2,\yT+0.2) {\textbf{QI}\;$C_i>0$};
    \node[qnote, anchor=south west] at (\xL+0.2,\yB-0.2) {\textbf{QIII}\;$C_i>0$};
    \node[qnote, anchor=south east] at (\xR-0.2,\yB-0.2) {\textbf{QIV}\;$C_i<0$};
  \end{scope}
\end{scope}

\node[note, anchor=north west, text width=34mm] (lev) at (142,54) {%
\textbf{Leverage}\\large $|\tilde x_i|$
};
\draw[carr] (pL.east) -- (lev.west);

\end{tikzpicture}%
}
\caption{Schematic RESAPLE scatterplot. The slope equals $\hat\rho_{RESAPLE}$. Unit contributions are $C_i=\tilde x_i\tilde y_i$ with leverage proportional to $\tilde x_i^2$.}
\label{fig:RESAPLE-scatter-schematic}
\end{figure}

\subsection{Local Testing}\label{subsec:RESAPLE-local-testing}
So far we have been working with global ESDA summaries, which are primarily designed to detect or quantify dependence aggregated over the entire spatial domain. In practice, dependence in the data is usually spatially heterogeneous, with clustering that often is not well captured by global statistics. This motivates the use of local indicators of spatial association, where each spatial unit is assigned a local statistic whose extremes identify candidate hot spots in the data. The classical example is the local Moran statistic in \citet{anselin1995}.

Just as one can construct a local Moran statistic, one can do the same for the RESAPLE. However the key distinction from classical local Moran diagnostics (as with local APLE) is conceptual. Local Moran statistics provide a local assessment of spatial association in the original data or residual scale \citep{anselin1995}. In contrast, the RESAPLE is a ratio in the residual space.

Recall that the RESAPLE can be written as a weighted sum of local contributions, as in \eqref{eq:RESAPLE-contrib}. This motivates a local contribution statistic in the form
\begin{equation}\label{eq:RESAPLE-local-ci}
C_i \defeq \tilde x_i\tilde y_i, \qquad i\in\{1,\dots,n\}.
\end{equation}
The sign of $C_i$ indicates whether unit $i$ contributes positively or negatively to the global numerator $e^\top A_r e$, and its magnitude determines the strength of the contribution. Since $B_r\succ0$ whenever $K_r\neq0$, the denominator in \eqref{eq:RESAPLE-contrib} is positive. Units with large positive $C_i$ push $\hat\rho_{RESAPLE}$ upward, while units with large negative $C_i$ push $\hat\rho_{RESAPLE}$ downward.

Further, if we normalise the local contributions by the global denominator and define
\begin{equation}\label{eq:RESAPLE-local-si}
S_i \defeq \frac{C_i}{\sum_{j=1}^n \tilde x_j^2} = \frac{\tilde x_i\tilde y_i}{e^\top B_r e},
\end{equation}
this normalisation yields $\sum_{i=1}^n S_i=\hat\rho_{RESAPLE}$, so $S_i$ provides a direct unit-level decomposition of the RESAPLE itself. 

Clearly, the local statistics have a close connection to the RESAPLE scatterplot. This can be seen by partitioning the scatterplot units into four quadrants according to the signs of $(\tilde x_i,\tilde y_i)$, essentially like the Moran scatterplot \citep{anselin1996}. Units in the first and third quadrants satisfy $C_i>0$ and contribute positively to RESAPLE, just as units in the second and fourth quadrants satisfy $C_i<0$ and contribute negatively. 

Beyond the interpretation of the local RESAPLE statistics in terms of the scatterplot, we can also define local tests (just as there are local tests defined based on the local Moran statistics). The procedure of testing by permutation is essentially identical to the procedure for the global RESAPLE statistic. The key distinction is that to control the family-wise error rate or the false discovery rate (generally the latter in ESDA), one must conduct multiple testing corrections.  

\section{Impact of the Choice of $W$}\label{sec:impact-W}
Throughout this paper, we have been ignoring the important role that the choice of $W$ plays. Indeed, a central practical difficulty in spatial regression is that the spatial weight matrix $W$ is rarely determined uniquely by the research problem. In many applied analyses, one can envision the data analyst beginning with a family of plausible neighbourhood structures, and then selecting one weight matrix as a compromise between interpretability and power \citep{anselinLagrangeMultiplierTest1988}. We can formalise this by defining a candidate class of plausible weight matrices that faithfully describe the spatial topology, $$\mathcal W \subseteq \R^{n\times n},$$ where $\mathcal W$ may include, for example, contiguity-based weights such as rook or queen adjacency, $k$ nearest neighbour weights, or distance-band weights.

Very often, there are two considerations the hypothetical data analyst has to consider. First, there is the question of neighbourhood size. Increasing the number of neighbours reduces variance in the spatial lag, but it also changes the scale of $W$ under standard normalisations, and it can affect boundary behaviour. In particular, for row-standardised weights, boundary units with fewer neighbours receive larger per-neighbour weights, which may amplify edge effects. Second, there is the question of regularity. Weight matrices arising from regular lattices often yield more stable spectral behaviour, while non-regular adjacency structures can produce more uneven residualised spectra. This issue is emphasised in \citet{rederMoreAPLEStatistic2009}, who examine maximally connected planar structures from the B-series with fixed size $n=8$ and fixed overall connectivity, and report that for certain small irregular structures, the APLE loses its empirical advantage over Moran-type diagnostics at moderate dependence levels, exhibiting both appreciable bias and increased variance when $\rho$ is not close to $0$. We use B07 in the same spirit: it is an adverse finite-sample benchmark inherited from the APLE literature, with the understanding that it is not a representative model for all possible irregular spatial graphs.

The statistic that one needs to understand the local behaviour of the RESAPLE is the restricted information at the null, which, as shown in Lemma \ref{lem:restricted-info}, satisfies $$\mathcal I_r(0) = \Tr(W_r^\top W_r)+\Tr(W_r^2)=2\Tr(K_r^2).$$ It follows from local asymptotic theory that the local asymptotic power of likelihood based tests tracks the Fisher information \citep{vaartAsymptoticStatistics1998}. Within a candidate class $\mathcal W$, weight matrices with larger $\mathcal I_r(0)$ yield greater local power. In other words, $\mathcal I_r(0)$ gives an analytic measure of detectability for spatial dependence using RESAPLE.

It is also true that neighbourhood size and edge effects can be related analytically to $\mathcal I_r(0)$ for common choices of weight matrix. One widely used choice is row-standardised contiguity weights derived from a binary adjacency matrix. The following result expresses $\mathcal I_r(0)$ in terms of the out-degrees of the areal units. 

\begin{thm}\label{thm:W-local-info}
Let $G\in\{0,1\}^{n\times n}$ be a directed adjacency matrix with zero diagonal and out-degrees $d_i\defeq\sum_{j=1}^n g_{ij}>0$. Let $D=\diag(d_1,\dots,d_n)$ and $W=D^{-1}G$. Let $M=I_n-X(X^\top X)^{-1}X^\top$ and let $K=(W+W^\top)/2$. Then the restricted information at $\rho=0$ satisfies
\begin{equation}\label{eq:Ir-basis-free}
\mathcal I_r(0)=2\Tr(K_r^2)=2\Tr(MKMK),
\end{equation}
and, in particular,
\begin{equation}\label{eq:Ir-upper}
\mathcal I_r(0)\leq 2\Tr(K^2)=\Tr(W^\top W)+\Tr(W^2).
\end{equation}
Moreover, for row-standardised adjacency weights,
\begin{equation}\label{eq:degree-identities}
\Tr(W^\top W)=\sum_{i=1}^n \frac{1}{d_i},
\qquad
\Tr(W^2)=\sum_{i=1}^n\sum_{j=1}^n \frac{g_{ij}g_{ji}}{d_i d_j}.
\end{equation}
If $G$ is symmetric, then $\Tr(W^2)=2\sum_{\{i,j\}\in E} (d_i d_j)^{-1}$, with $E$ the set of edges. If, further, the graph is $d$-regular, then $\Tr(W^\top W)+\Tr(W^2)=2n/d$.
\end{thm}

The bound \eqref{eq:Ir-upper} and the identities \eqref{eq:degree-identities} give us two immediate qualitative consequences for row-standardised weights. First, the term $\sum_i d_i^{-1}$ increases when the graph contains many low-degree nodes, which occurs naturally at boundaries for contiguity graphs. If the inequality \eqref{eq:Ir-upper} is relatively tight, then this shows that the boundary areal units disproportionately impact the unrestricted information and hence RESAPLE statistic, compared to interior nodes. This is an analytic proof of what is known as \emph{edge effects} in spatial statistics. In general, the identities \eqref{eq:degree-identities} also show that uniformly increasing degrees $d_i$ decrease the Fisher information, which generally decreases detectability of spatial dependence, making the test less powerful. 

Beyond just changing the information, the topology encoded by $W$ can also change the shape of the sampling distribution of the RESAPLE. To study this phenomenon, suppose for a moment that the residual contrasts $e$ are Gaussian in $\R^r$. In this case, noting that the RESAPLE is a Rayleigh coefficient, one can show that its distribution is determined by the spectrum of the corresponding residual-space operators. When $r$ is small, or when this spectrum is highly uneven, the resulting distribution can deviate strongly from a Gaussian approximation. This spectral mechanism provides a theoretical explanation for the type of phenomena reported in \citet{rederMoreAPLEStatistic2009} for the small B-series graphs, and it motivates the use of either the exact Gaussian null reference distribution or Freedman-Lane permutation testing when the effective residual dimension is small or the residualised spectrum is concentrated. A heuristic rationale for this is given by Theorem \ref{thm:moran-spectrum-law} in Appendix \ref{app:impact-W}.

In summary, the choice of $W$ affects RESAPLE in two ways. It affects the restricted information $\mathcal I_r(0)$, which controls local power and local precision near $\rho=0$. It also affects its distributional shape through the spectrum of residualised operators such as $K_r$, which becomes particularly relevant in small samples and irregular neighbourhood structures. The key idea here is that when $\rho$ is expected to be close to $0$ and $r$ is moderate to large, Theorem \ref{thm:W-local-info} suggests that $\mathcal I_r(0)$ provides a principled way to compare candidate weights for local detectability. In settings where $r$ is small with high eigenvalue multiplicities, exact testing is a better option than approximate $z$-testing with RESAPLE. 

\section{Simulation Study}\label{sec:simulations}

We conduct a simulation study to evaluate RESAPLE under controlled spatial error data generating mechanisms. The simulations address four questions. First, they examine how the restricted information $\mathcal I_r(0)$ guides the choice of the lattice weight matrix. Second, they compare the bias, standard deviation, and RMSE of RESAPLE with residual Moran's $\mathcal I_M$, residual APLE, MAPLE, and REML. Third, they assess empirical size and power for testing $H_0:\rho=0$ against $H_1:\rho>0$. Fourth, they examine robustness under non-Gaussian error distributions.

Across simulations, the response is generated from the spatial error model
$$Z=X\beta+U,\qquad U=\rho WU+\epsilon.$$
Unless stated otherwise, $\epsilon\sim N(0,I_n)$, $\rho\in\{0,0.05,\ldots,0.95\}$, and all spatial weights are row-standardised. The nominal level is fixed at $\alpha=0.05$. Permutation tests use $199$ Freedman-Lane residual permutations. Each design point uses $1000$ Monte Carlo replicates.

The design matrix $X$ is fixed within each design. It always contains an intercept. The first non-intercept covariate is constructed from the first spatial coordinate with a small independent Gaussian perturbation. The second is constructed analogously from the second spatial coordinate. Further covariates are independent Gaussian variables. All non-intercept covariates are standardised column-wise. The regression coefficients are fixed as $\beta_1=1$ and $\beta_j=0.6/\sqrt{j-1}$ for $j\geq 2$. This keeps the systematic component comparable as the number of covariates changes.

We consider three graph classes. The first is a regular $m\times m$ lattice, with $n=m^2$ and $n\in\{25,100,400\}$. For the lattice experiments, the candidate set of weights is
$$\mathcal W =\{\mathrm{rook},\mathrm{queen},\mathrm{knn4},\mathrm{knn6},\mathrm{knn8}\}.$$
The second graph is B07, a fixed eight-node binary adjacency structure that is known to be difficult for some spatial autocorrelation diagnostics. The B07 matrix is row-standardised, and its coordinates are assigned on the unit circle only for covariate construction. The third graph is a fixed large irregular planar graph with $n=128$. It is generated by repeated triangular face subdivision. At each step, a triangular face is sampled with probability proportional to the square root of its area. A new point is placed inside that face using Dirichlet barycentric weights with common shape parameter $2.5$, and the new vertex is connected to the three vertices of the selected face. This gives a connected irregular graph with average degree $5.91$, minimum degree $3$, and maximum degree $42$.

Figure~\ref{fig:sim-graphs} shows representative graphs used in the simulations. The lattice panels show the selected rook graph and the denser knn8 candidate on the $10\times 10$ lattice. The remaining panels show B07 and the large irregular graph. Table~\ref{tab:sim-weight-info} reports the lattice candidate weights and their restricted information. The complete graph diagnostics for all $n$, $p$, and $W$ are given in Appendix~\ref{app:simulations}.

\begin{figure*}[pos=!htpb]
\centering
\begin{subfigure}[t]{0.245\textwidth}
\centering
\includegraphics[width=\linewidth]{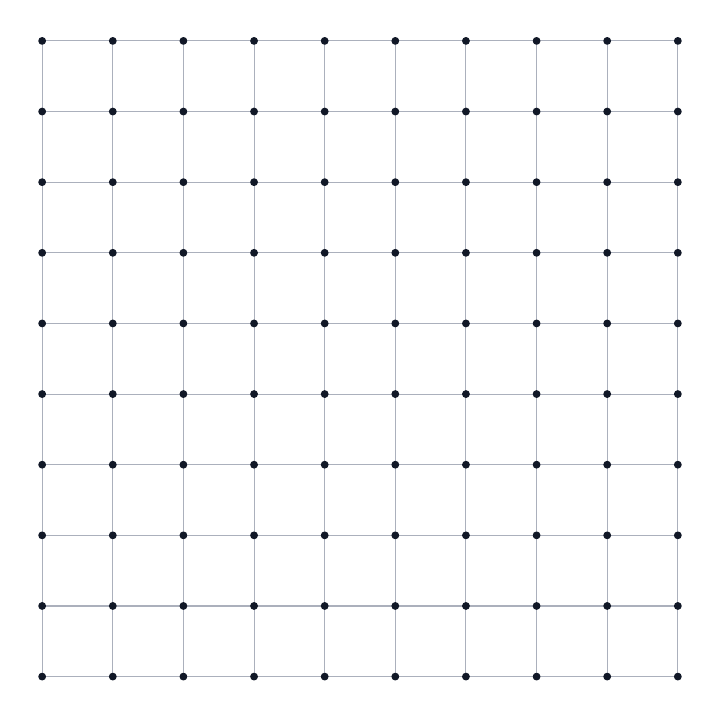}
\caption{\footnotesize Lattice, rook.}
\label{subfig:sim-graph-rook}
\end{subfigure}\hfill
\begin{subfigure}[t]{0.245\textwidth}
\centering
\includegraphics[width=\linewidth]{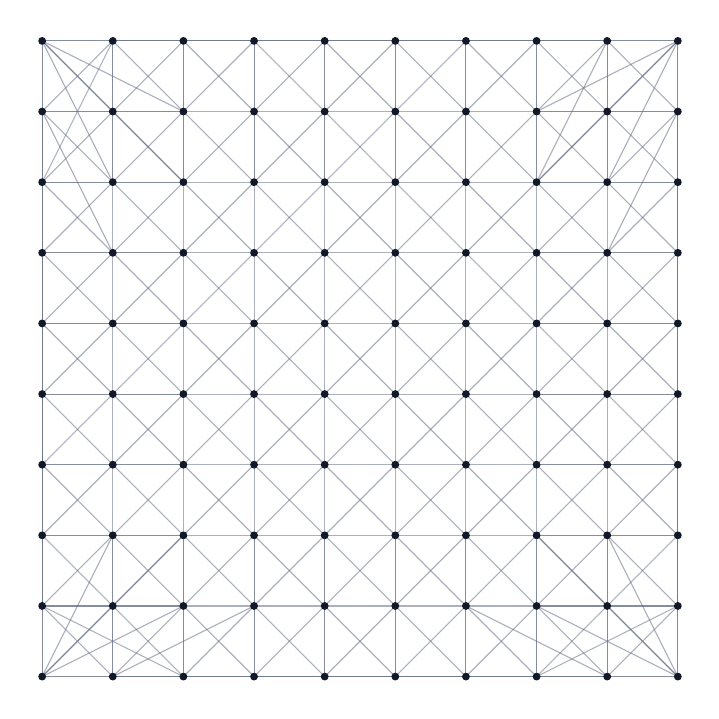}
\caption{\footnotesize Lattice, knn8.}
\label{subfig:sim-graph-knn8}
\end{subfigure}\hfill
\begin{subfigure}[t]{0.245\textwidth}
\centering
\includegraphics[width=\linewidth]{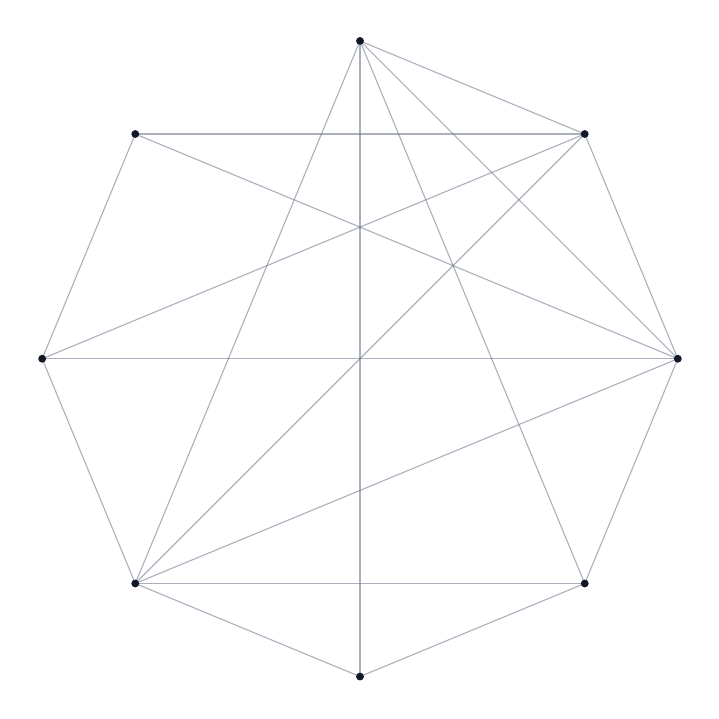}
\caption{\footnotesize B07.}
\label{subfig:sim-graph-b07}
\end{subfigure}\hfill
\begin{subfigure}[t]{0.245\textwidth}
\centering
\includegraphics[width=\linewidth]{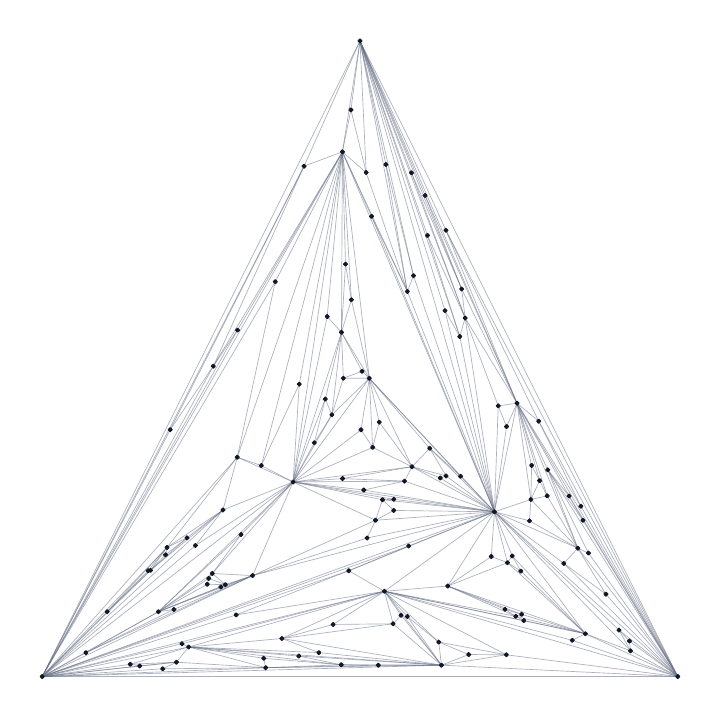}
\caption{\footnotesize Large irregular graph.}
\label{subfig:sim-graph-irregular128}
\end{subfigure}
\caption{Representative spatial graphs used in the simulation study. The lattice candidates are compared through $\mathcal I_r(0)$ before estimation and testing are assessed. B07 and the large irregular graph use their fixed adjacency matrices.}
\label{fig:sim-graphs}
\end{figure*}

\begin{table*}[pos=!htpb]
\centering
\caption{Lattice candidate weights and restricted information. Degree summaries and density are shown for the $10\times 10$ lattice. The fifth column reports $\mathcal I_r(0)$ for $n=100$ and $p=5$. The mean is taken over $n\in\{25,100,400\}$ and $p\in\{1,5,20\}$.}
\footnotesize
\renewcommand{\arraystretch}{1.15}
\setlength{\tabcolsep}{6pt}
\begin{tabular*}{\textwidth}{@{\extracolsep{\fill}} lrrrrrrc @{}}
\toprule
$W$ & Avg. degree & Min. degree & Max. degree & Density & $\mathcal I_r(0)$, $n=100,p=5$ & Mean $\mathcal I_r(0)$ & Selected \\
\midrule
rook  & 3.60 & 2 & 4  & 0.0364 & 49.03 & 85.54 & Yes \\
knn4  & 4.38 & 4 & 7  & 0.0442 & 38.56 & 74.19 & No \\
queen & 6.84 & 3 & 8  & 0.0691 & 23.85 & 43.27 & No \\
knn6  & 7.64 & 6 & 10 & 0.0772 & 19.91 & 39.05 & No \\
knn8  & 8.52 & 8 & 12 & 0.0861 & 17.20 & 36.17 & No \\
\bottomrule
\end{tabular*}
\label{tab:sim-weight-info}
\end{table*}

\subsection{Simulation 1: Selecting the Lattice Weight Matrix}\label{subsec:sim-W-selection}

We first examine whether $\mathcal I_r(0)$ gives a useful ordering of the candidate lattice weights. This experiment uses the Gaussian SEM described above. For each lattice design, we compute $\mathcal I_r(0)$ for $W\in\mathcal W_{\mathrm{lat}}$ and then evaluate the RESAPLE exact test at $\rho=0.3$. The power calculation is not used to choose $W$. It is used only to check whether the information ordering corresponds to empirical detectability.

For each $(n,p)$, the selected weight is the one that maximises $\mathcal I_r(0)$. Table~\ref{tab:sim-W-selection} shows that rook contiguity is selected for every lattice design. This is also the best weight under the average-information summary in Table~\ref{tab:sim-weight-info}. We therefore use rook contiguity as the lattice weight in the main estimation, testing, and robustness experiments. The remaining candidate weights are reported in the appendix.

\begin{table}[pos=!htpb]
\centering
\caption{Selected lattice weights by restricted information. Each entry gives the selected $W$ and the corresponding $\mathcal I_r(0)$.}
\footnotesize
\renewcommand{\arraystretch}{1.12}
\setlength{\tabcolsep}{5pt}
\begin{tabular*}{\columnwidth}{@{\extracolsep{\fill}} rrrr @{}}
\toprule
$n$ & $p=1$ & $p=5$ & $p=20$ \\
\midrule
25  & rook, $14.17$  & rook, $9.19$   & rook, $0.74$ \\
100 & rook, $54.88$  & rook, $49.03$  & rook, $35.05$ \\
400 & rook, $211.27$ & rook, $205.34$ & rook, $190.23$ \\
\bottomrule
\end{tabular*}

\label{tab:sim-W-selection}
\end{table}

Figure~\ref{fig:sim-W-selection} compares $\mathcal I_r(0)$ with empirical power at $\rho=0.3$. The ordering is consistent across the lattice designs. The selected rook weight gives the largest $\mathcal I_r(0)$ and the highest power in each $(n,p)$ panel. The separation is clearest for $n=100$. For $n=25$, the power is limited by the residual dimension when $p$ is large. For $n=400$, power is close to one for several weights, so the ordering is partly compressed.

\begin{figure*}[pos=!htpb]
\centering
\begin{subfigure}[t]{0.7\textwidth}
\centering
\includegraphics[width=0.88\linewidth]{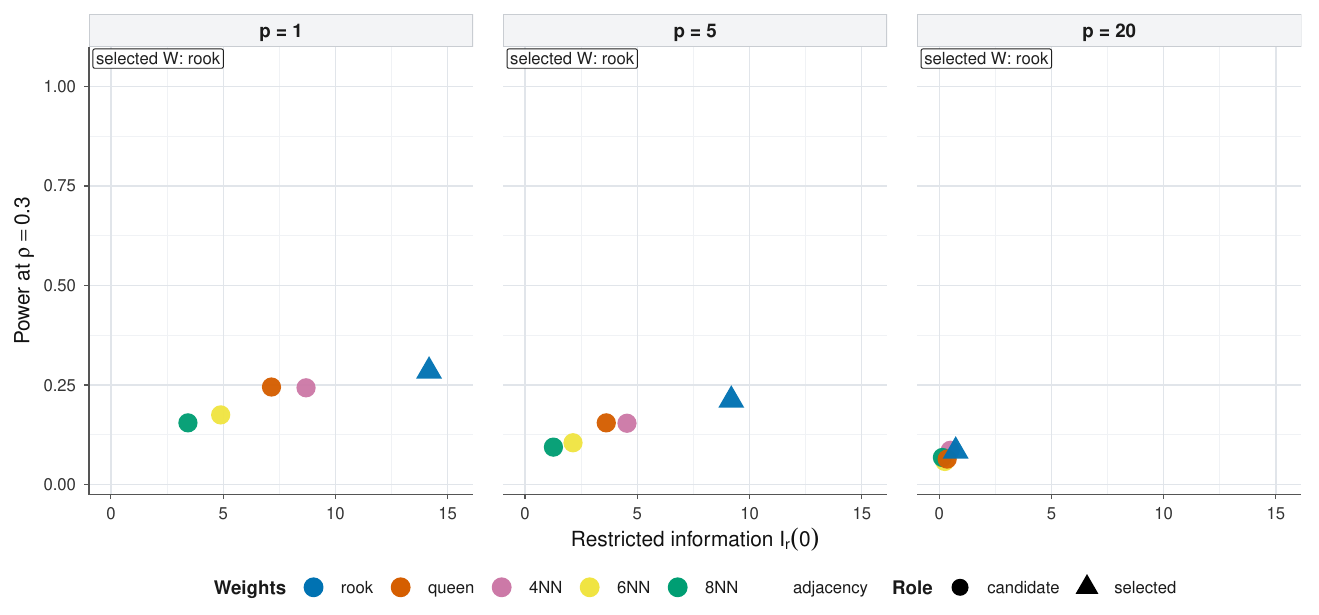}
\caption{\footnotesize $n=25$.}
\label{subfig:sim-W-selection-n25}
\end{subfigure}

\vspace{0.8em}

\begin{subfigure}[t]{0.7\textwidth}
\centering
\includegraphics[width=0.88\linewidth]{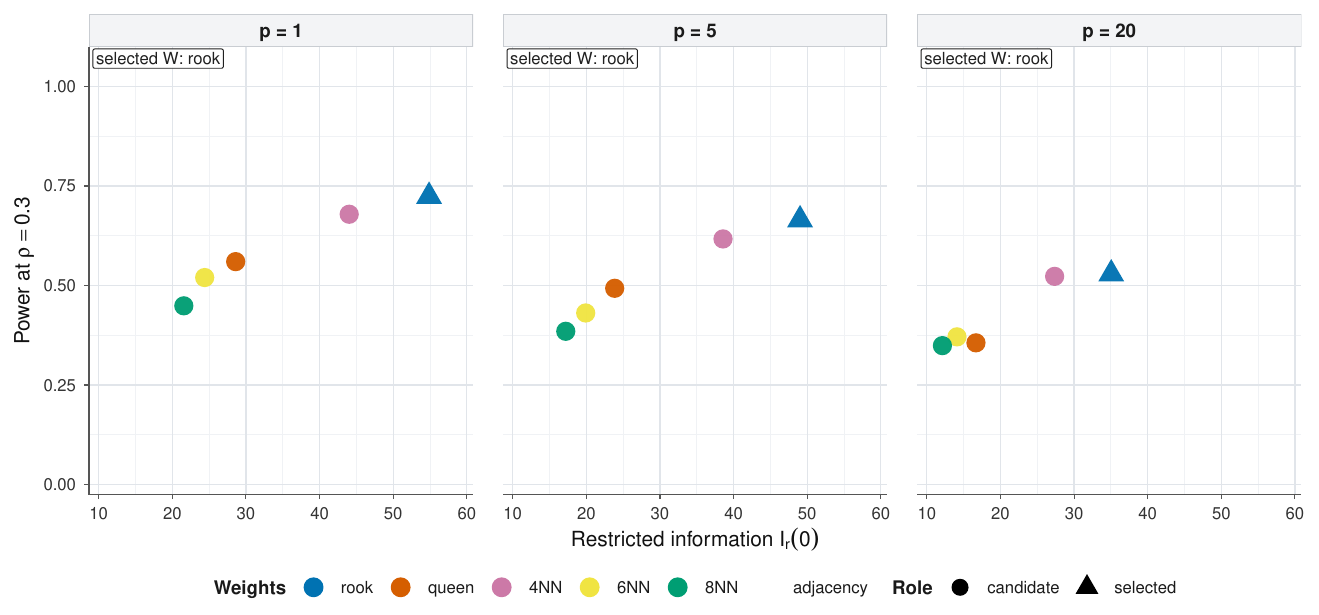}
\caption{\footnotesize $n=100$.}
\label{subfig:sim-W-selection-n100}
\end{subfigure}

\vspace{0.8em}

\begin{subfigure}[t]{0.7\textwidth}
\centering
\includegraphics[width=0.88\linewidth]{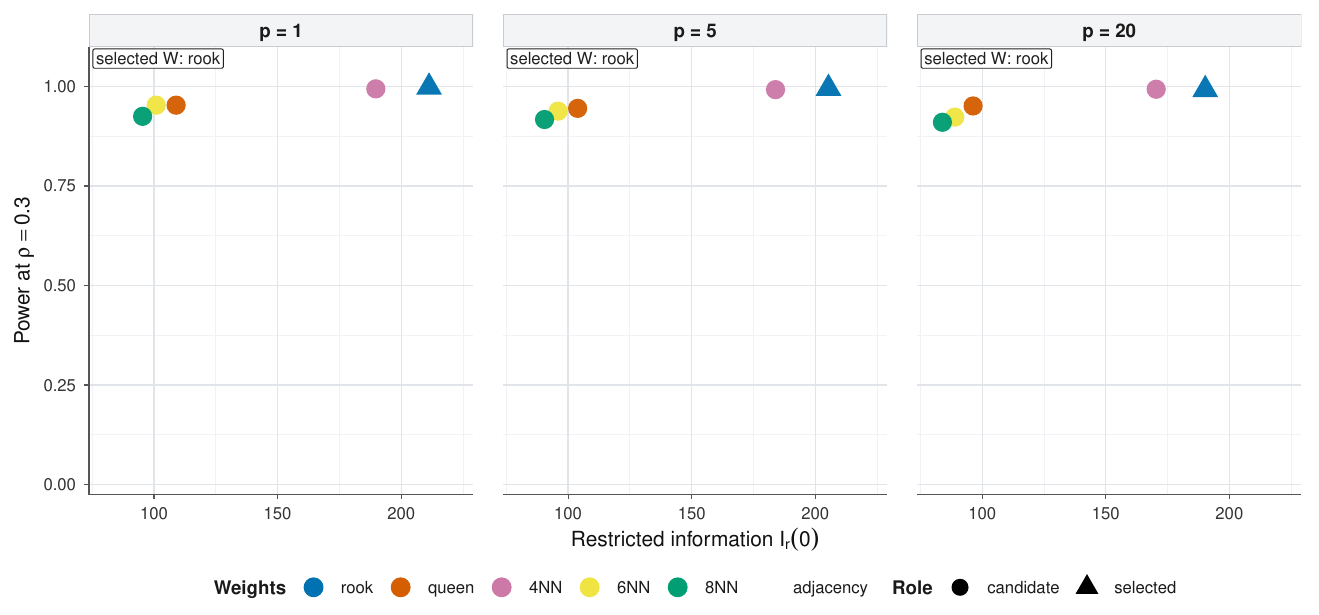}
\caption{\footnotesize $n=400$.}
\label{subfig:sim-W-selection-n400}
\end{subfigure}
\caption{Empirical power at $\rho=0.3$ against restricted information $\mathcal I_r(0)$ for the candidate lattice weights. The selected weight maximises $\mathcal I_r(0)$ within each $(n,p)$ design.}
\label{fig:sim-W-selection}
\end{figure*}

The size results at $\rho=0$ are summarised in Table~\ref{tab:sim-W-size}. The exact test is close to the nominal level for the selected rook designs. The largest departure among the selected designs is $0.013$, occurring at $n=400$ and $p=1$. The full size table, including all candidate weights and the permutation and $z$ reference procedures, is given in Appendix~\ref{app:simulations}.

\begin{table}[pos=!htpb]
\centering
\caption{Empirical size for the RESAPLE exact test using the selected lattice weight. The nominal level is $\alpha=0.05$.}
\footnotesize
\renewcommand{\arraystretch}{1.12}
\setlength{\tabcolsep}{5pt}
\begin{tabular*}{\columnwidth}{@{\extracolsep{\fill}} rrrrrr @{}}
\toprule
$n$ & $p$ & $r$ & $W$ & Size & $|\widehat{\alpha}-0.05|$ \\
\midrule
25  & 1  & 24  & rook & 0.051 & 0.001 \\
25  & 5  & 20  & rook & 0.067 & 0.017 \\
25  & 20 & 5   & rook & 0.058 & 0.008 \\
100 & 1  & 99  & rook & 0.047 & 0.003 \\
100 & 5  & 95  & rook & 0.046 & 0.004 \\
100 & 20 & 80  & rook & 0.043 & 0.007 \\
400 & 1  & 399 & rook & 0.063 & 0.013 \\
400 & 5  & 395 & rook & 0.055 & 0.005 \\
400 & 20 & 380 & rook & 0.058 & 0.008 \\
\bottomrule
\end{tabular*}

\label{tab:sim-W-size}
\end{table}

These results support using $\mathcal I_r(0)$ as the lattice weight-selection criterion in the following simulations. The criterion is fixed before estimation and testing are compared. The subsequent lattice results therefore use rook contiguity in the main text, while the appendix reports the remaining candidate weights.

\subsection{Simulation 2: Gaussian Estimation Performance}\label{subsec:sim-gaussian-estimation}

We next study RESAPLE as a point estimator of $\rho$ under the Gaussian SEM. The lattice experiments use the selected rook weight from Simulation 1. B07 and the large irregular graph use their fixed adjacency matrices. We compare RESAPLE with residual Moran's $\mathcal I_M$, residual APLE, MAPLE, and REML over $\rho\in\{0,0.05,\ldots,0.95\}$. For each estimator $\hat\rho$, we report bias, standard deviation, and RMSE, where $$\operatorname{RMSE}(\hat\rho)=\{\mathbb E(\hat\rho-\rho)^2\}^{1/2}.$$ Further numerical summaries are given in Appendix~\ref{app:simulations}.

Figure~\ref{fig:sim-estimation-rmse} reports RMSE for the selected lattice design, B07, and the large irregular graph. On the lattice, RESAPLE has the lowest RMSE for most values of $\rho$ across all three covariate dimensions. It wins against the ESDA comparators for 16 of 20 values of $\rho$ when $p=1$, 16 of 20 when $p=5$, and 15 of 20 when $p=20$. The remaining wins are attained by residual Moran's $\mathcal I_M$ at small values of $\rho$.

\begin{figure*}[pos=!htpb]
\centering
\begin{subfigure}[t]{0.7\textwidth}
\centering
\includegraphics[width=0.86\linewidth]{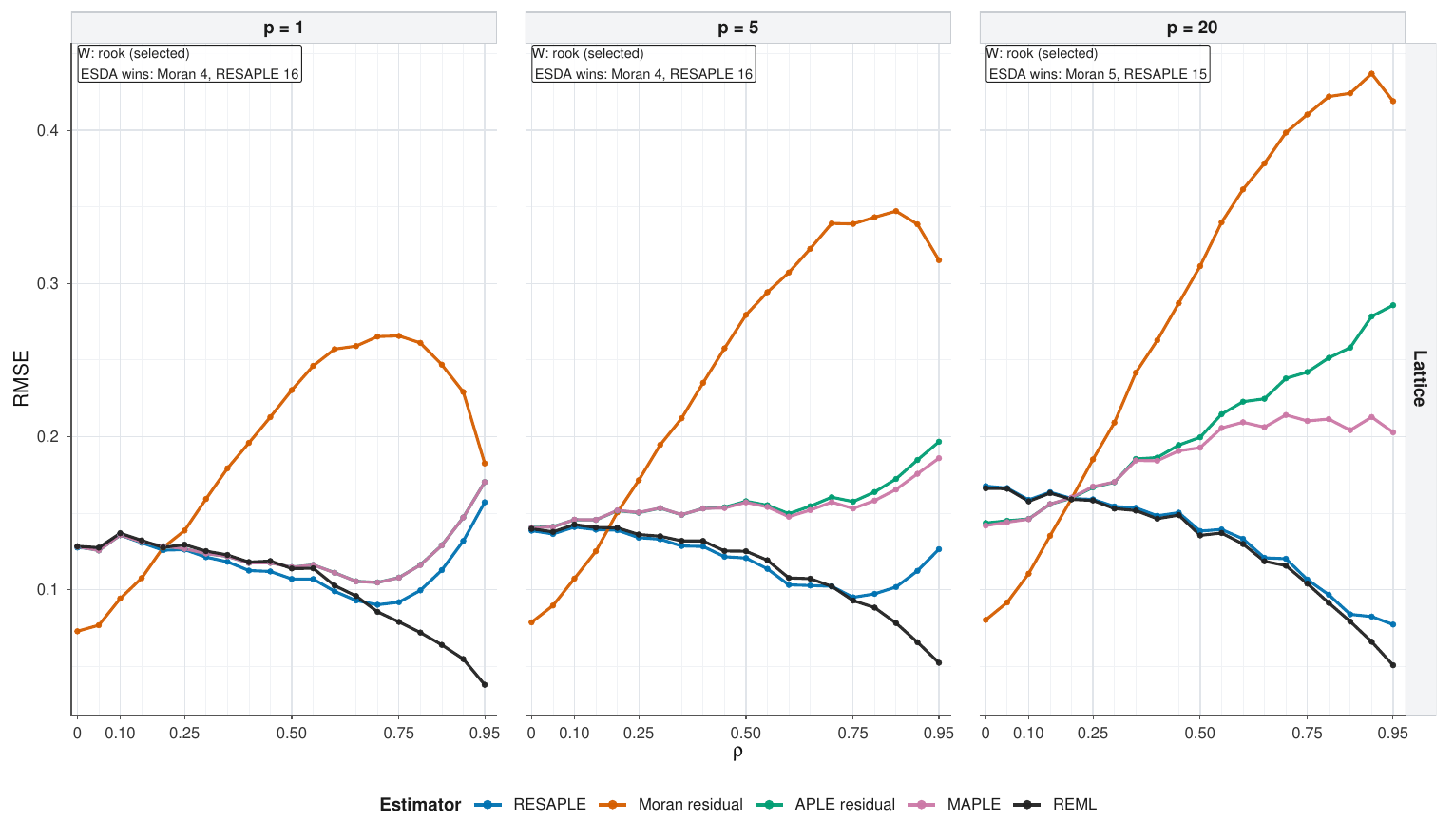}
\caption{\footnotesize Lattice with selected rook weights, $n=100$.}
\label{subfig:sim-estimation-rmse-lattice}
\end{subfigure}

\vspace{0.8em}

\begin{subfigure}[t]{0.7\textwidth}
\centering
\includegraphics[width=0.86\linewidth]{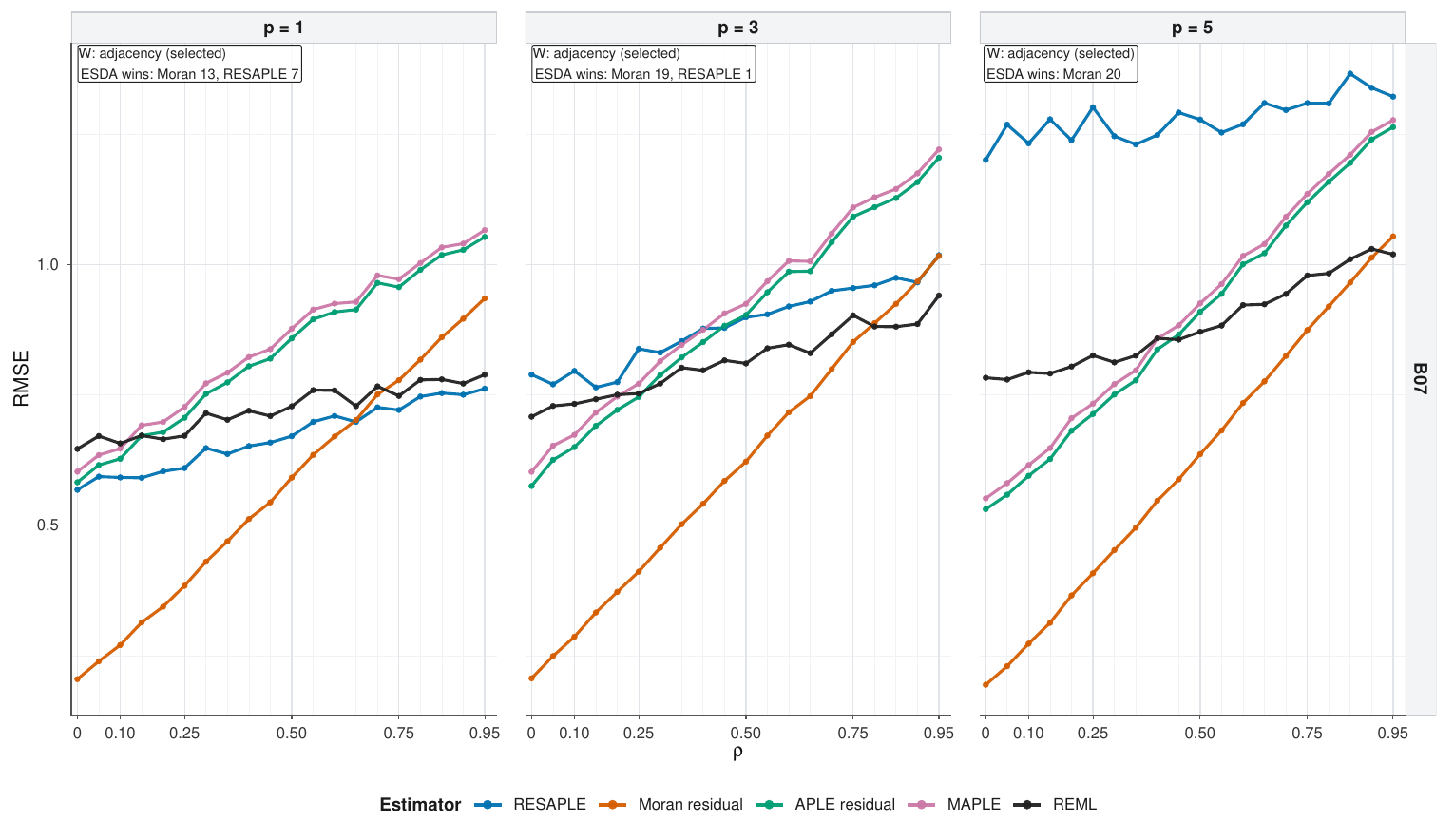}
\caption{\footnotesize B07 adjacency.}
\label{subfig:sim-estimation-rmse-b07}
\end{subfigure}

\vspace{0.8em}

\begin{subfigure}[t]{0.7\textwidth}
\centering
\includegraphics[width=0.86\linewidth]{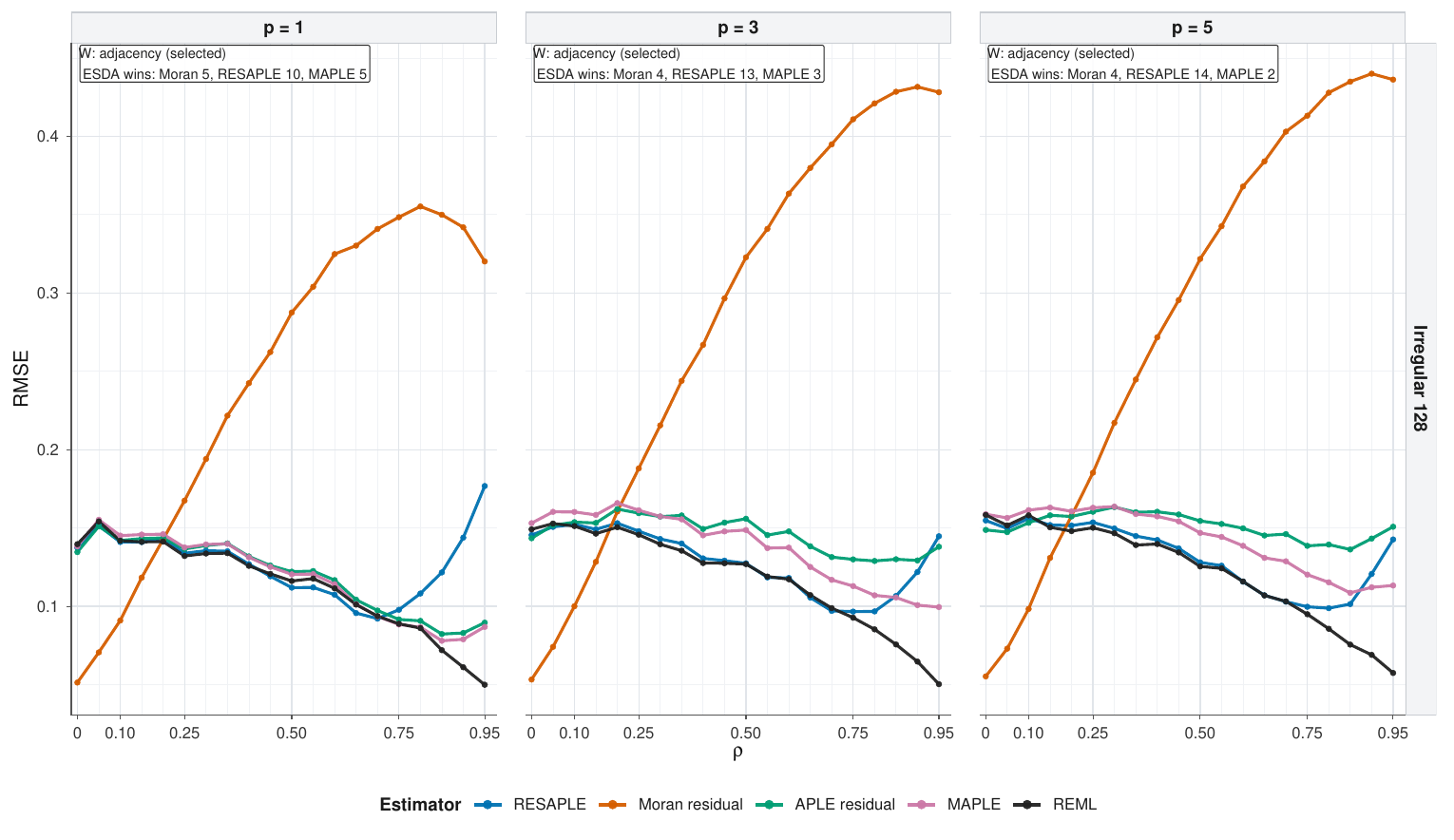}
\caption{\footnotesize Large irregular adjacency graph, $n=128$.}
\label{subfig:sim-estimation-rmse-irregular128}
\end{subfigure}
\caption{RMSE for estimating $\rho$ under the Gaussian SEM. The annotations count the number of $\rho$ values, out of 20, at which each ESDA estimator has the lowest RMSE. REML is shown as a likelihood benchmark, but is not included in the ESDA win count.}
\label{fig:sim-estimation-rmse}
\end{figure*}

The B07 results are different. Residual Moran's $\mathcal I_M$ has the lowest RMSE for most values of $\rho$, especially when $p=3$ or $p=5$. This is a finite-sample effect. B07 has only eight nodes, and the residual dimension falls to $r=3$ when $p=5$. In this setting, the RMSE is dominated by sampling variation. Residual Moran's $\mathcal I_M$ has low variance, but it is strongly biased downward. RESAPLE has much smaller bias, but its variance is larger in this very small residual space.

The large irregular graph gives a more representative irregular setting. It keeps the irregular topology but avoids the extreme residual dimensions of B07. In this case, RESAPLE is again competitive in RMSE. It has the lowest RMSE for 10 of 20 values of $\rho$ when $p=1$, 13 of 20 when $p=3$, and 14 of 20 when $p=5$. The remaining wins are split between residual Moran's $\mathcal I_M$ and MAPLE.

Figure~\ref{fig:sim-estimation-bias} reports the corresponding bias curves. These curves explain the main source of RESAPLE's estimation advantage. On the lattice, RESAPLE has the smallest absolute bias for all plotted values of $\rho$ and all three values of $p$. The same pattern is largely present on the large irregular graph. On B07, RESAPLE also has the smallest absolute bias in most panels, despite not always having the lowest RMSE. Thus, where RESAPLE loses in RMSE, the loss is mainly due to higher sampling variation, not systematic bias.

\begin{figure*}[pos=!htpb]
\centering
\begin{subfigure}[t]{0.7\textwidth}
\centering
\includegraphics[width=0.86\linewidth]{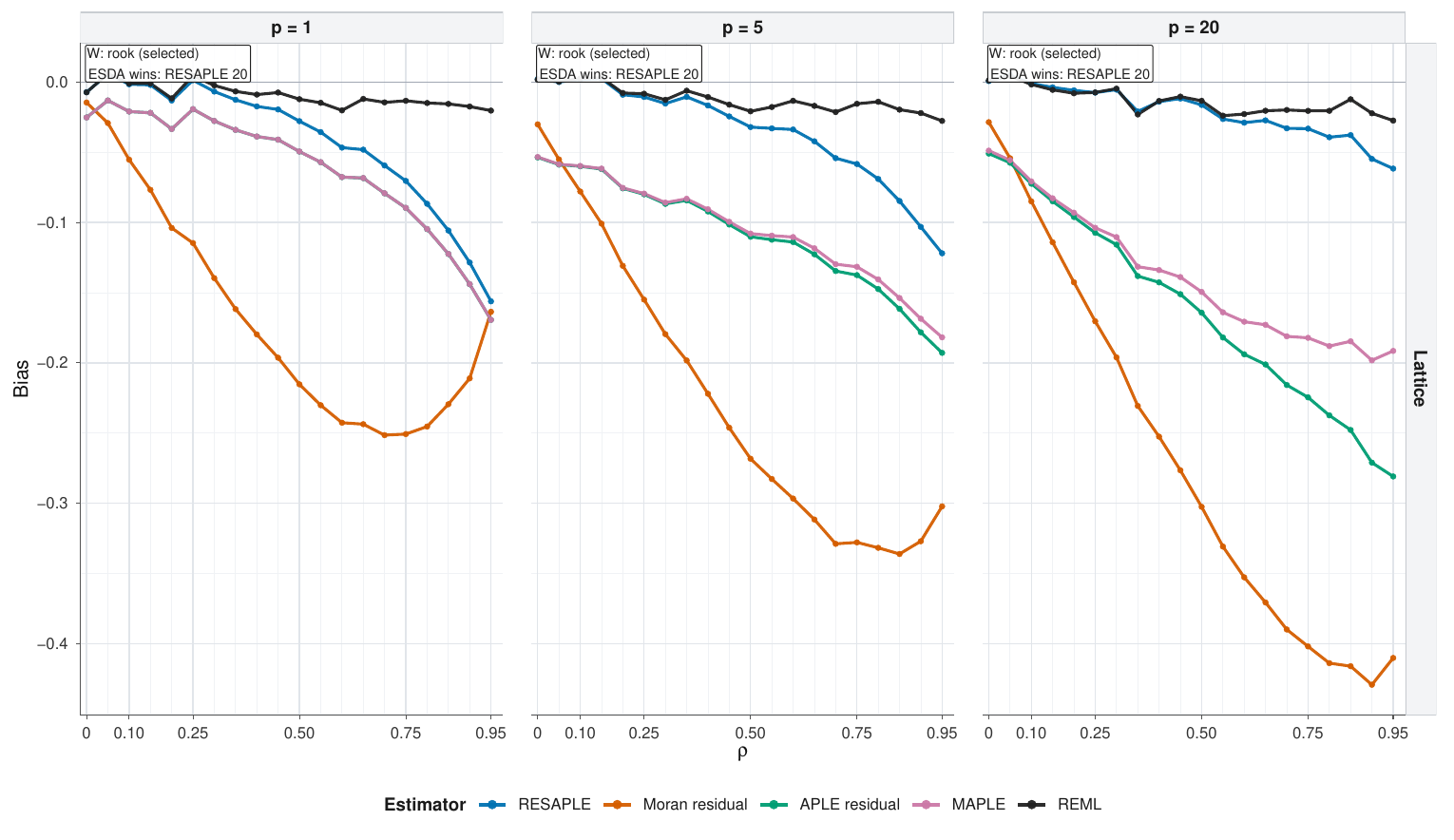}
\caption{\footnotesize Lattice with selected rook weights, $n=100$.}
\label{subfig:sim-estimation-bias-lattice}
\end{subfigure}

\vspace{0.8em}

\begin{subfigure}[t]{0.7\textwidth}
\centering
\includegraphics[width=0.86\linewidth]{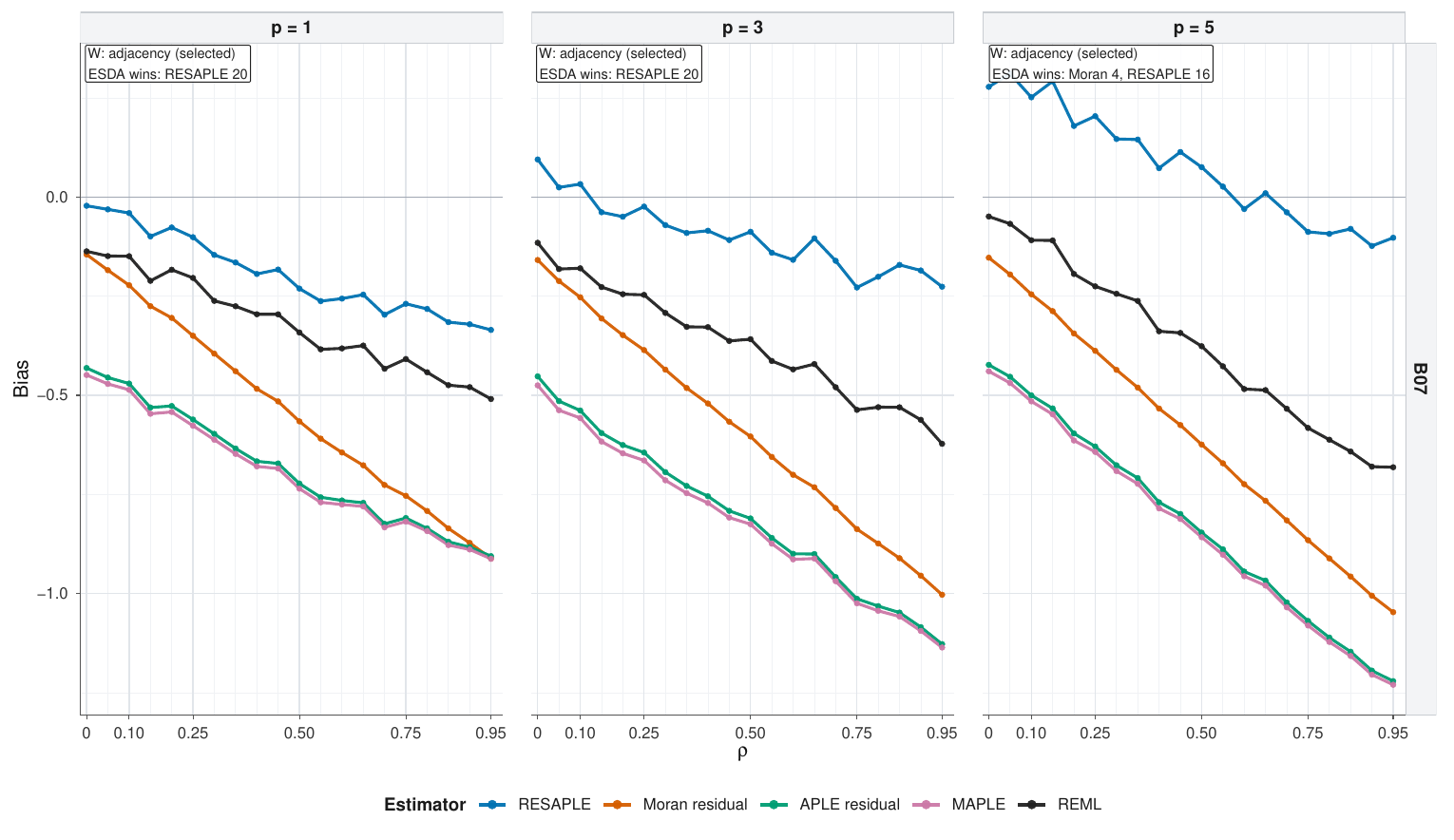}
\caption{\footnotesize B07 adjacency.}
\label{subfig:sim-estimation-bias-b07}
\end{subfigure}

\vspace{0.8em}

\begin{subfigure}[t]{0.7\textwidth}
\centering
\includegraphics[width=0.86\linewidth]{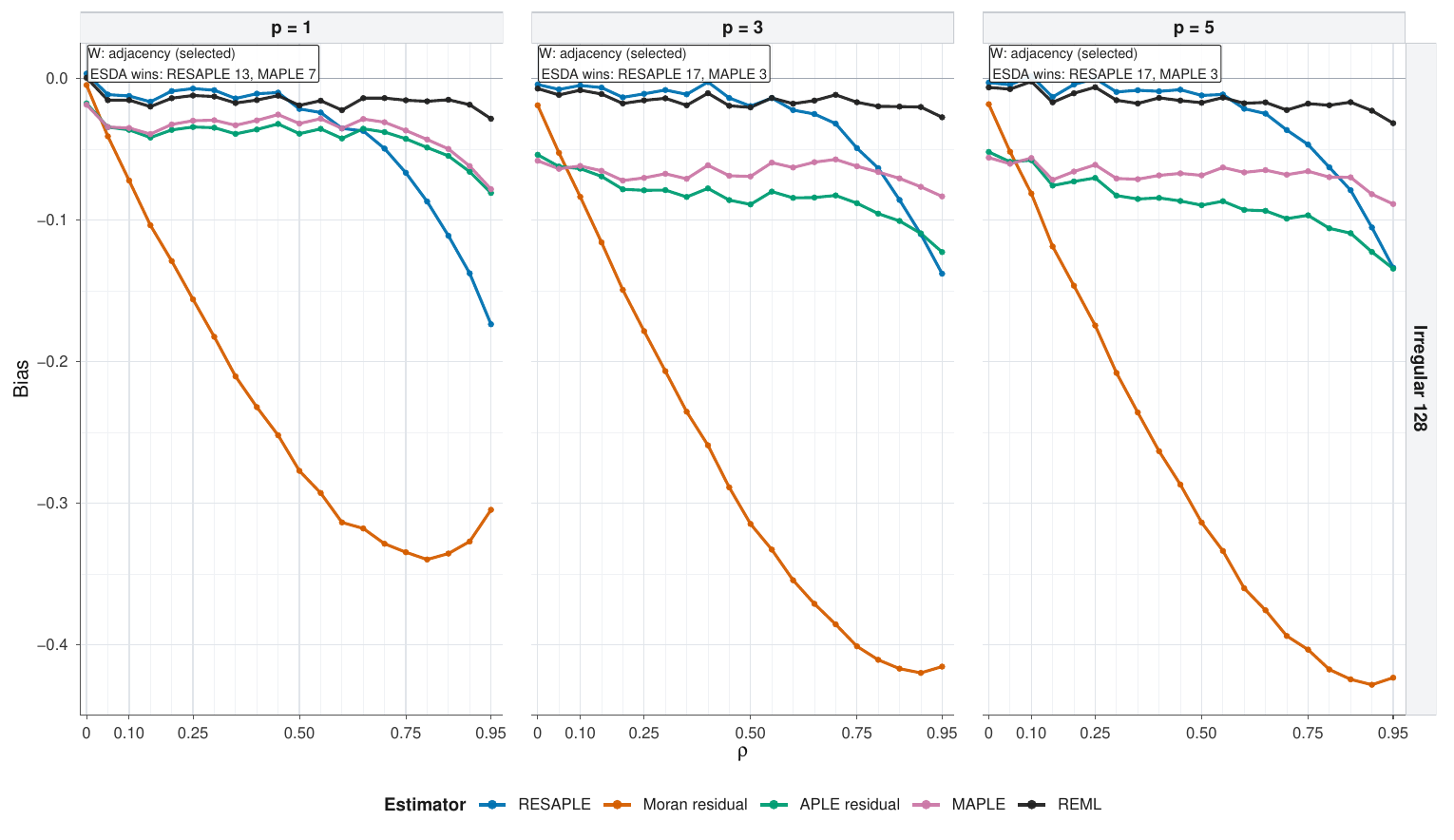}
\caption{\footnotesize Large irregular adjacency graph, $n=128$.}
\label{subfig:sim-estimation-bias-irregular128}
\end{subfigure}
\caption{Bias for estimating $\rho$ under the Gaussian SEM. The annotations count the number of $\rho$ values, out of 20, at which each ESDA estimator has the smallest absolute bias. REML is shown as a likelihood benchmark, but is not included in the ESDA win count.}
\label{fig:sim-estimation-bias}
\end{figure*}

Overall, the estimation results support the intended role of RESAPLE. It is designed to reduce the restricted-likelihood bias of residual spatial autocorrelation diagnostics. This is visible across the lattice and the large irregular graph, and also in the B07 stress test. Its main limitation is increased sampling variation in very small residual spaces. This limitation is most visible for B07, where $n=8$ and $r$ can be as small as $3$. In ordinary lattice and larger irregular settings, this variance cost is less severe, and the reduced bias usually gives RESAPLE the best or near-best RMSE.

\subsection{Simulation 3: Gaussian Testing Performance}\label{subsec:sim-gaussian-testing}

We next study RESAPLE as a test statistic under the Gaussian SEM. The null hypothesis is $H_0:\rho=0$, and the alternative is $H_1:\rho>0$. The lattice experiments use the selected rook weight from Simulation 1. B07 and the large irregular graph use their fixed adjacency matrices. We separate two comparisons. First, we compare RESAPLE with residual Moran's $\mathcal I_M$, residual APLE, and MAPLE using the same Freedman-Lane permutation reference procedure. Second, we compare the exact, permutation, and $z$ reference procedures for RESAPLE itself.

Figure~\ref{fig:sim-testing-permutation} gives the first comparison. On the lattice, the four permutation tests have almost indistinguishable power curves. All methods have correct size at $\rho=0$, and power increases rapidly once $\rho$ moves away from zero. This shows that RESAPLE is competitive as a test statistic, but it does not give the same separation from the residual ESDA alternatives as it does for bias in Simulation 2.

\begin{figure*}[pos=!htpb]
\centering
\begin{subfigure}[t]{0.7\textwidth}
\centering
\includegraphics[width=0.86\linewidth]{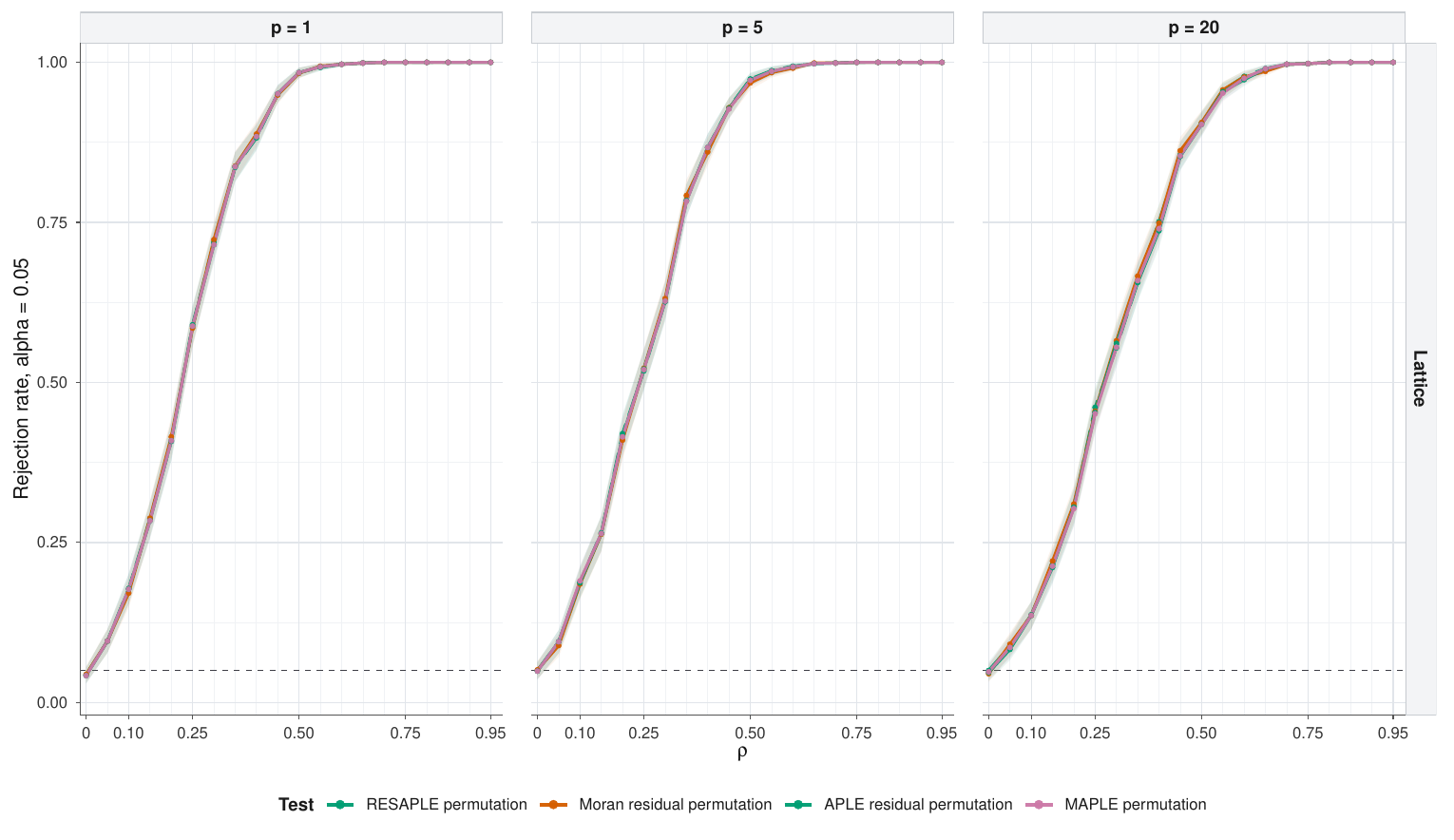}
\caption{\footnotesize Lattice with selected rook weights, $n=100$.}
\label{subfig:sim-testing-permutation-lattice}
\end{subfigure}

\vspace{0.8em}

\begin{subfigure}[t]{0.7\textwidth}
\centering
\includegraphics[width=0.86\linewidth]{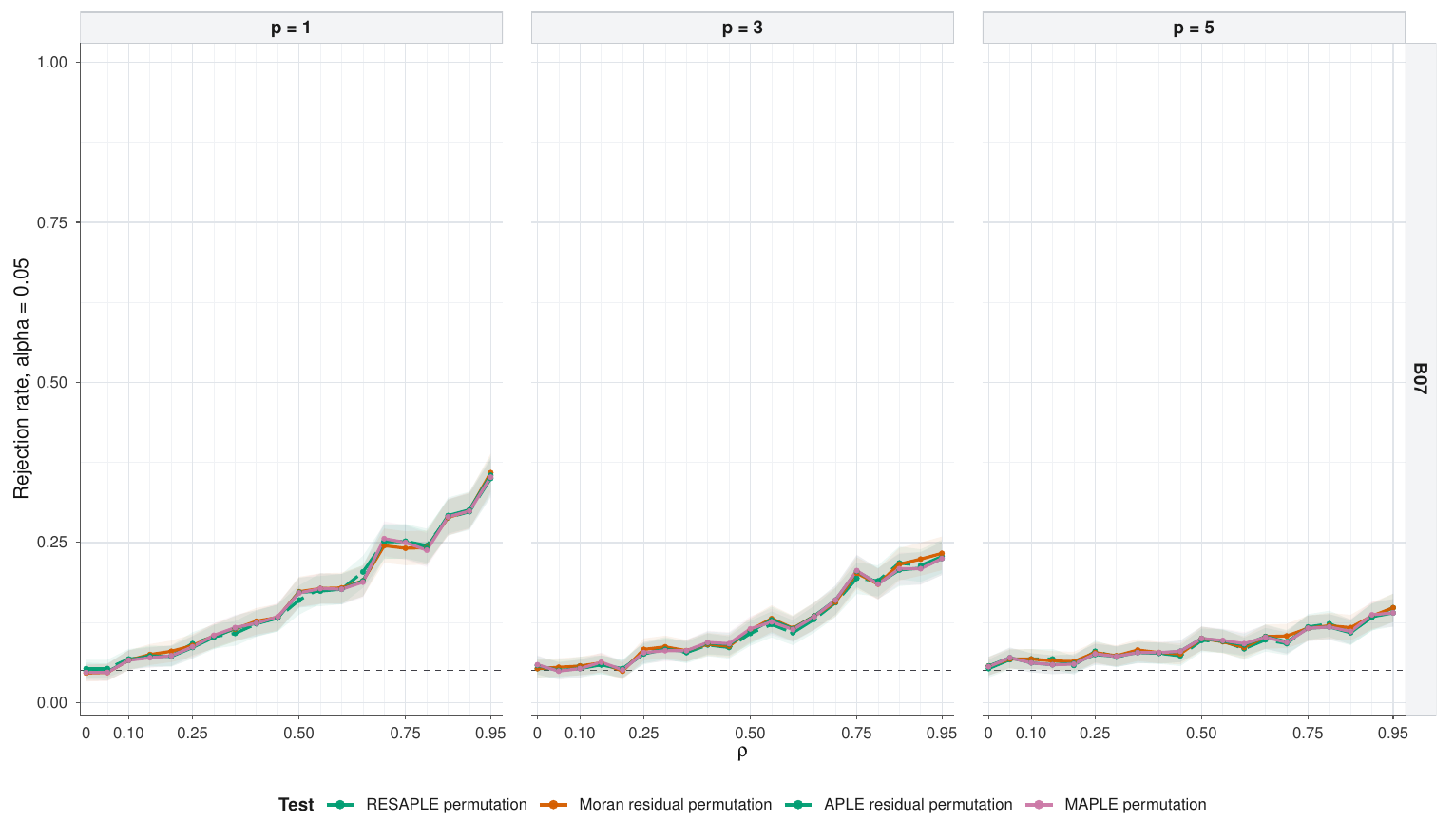}
\caption{\footnotesize B07 adjacency.}
\label{subfig:sim-testing-permutation-b07}
\end{subfigure}

\vspace{0.8em}

\begin{subfigure}[t]{0.7\textwidth}
\centering
\includegraphics[width=0.86\linewidth]{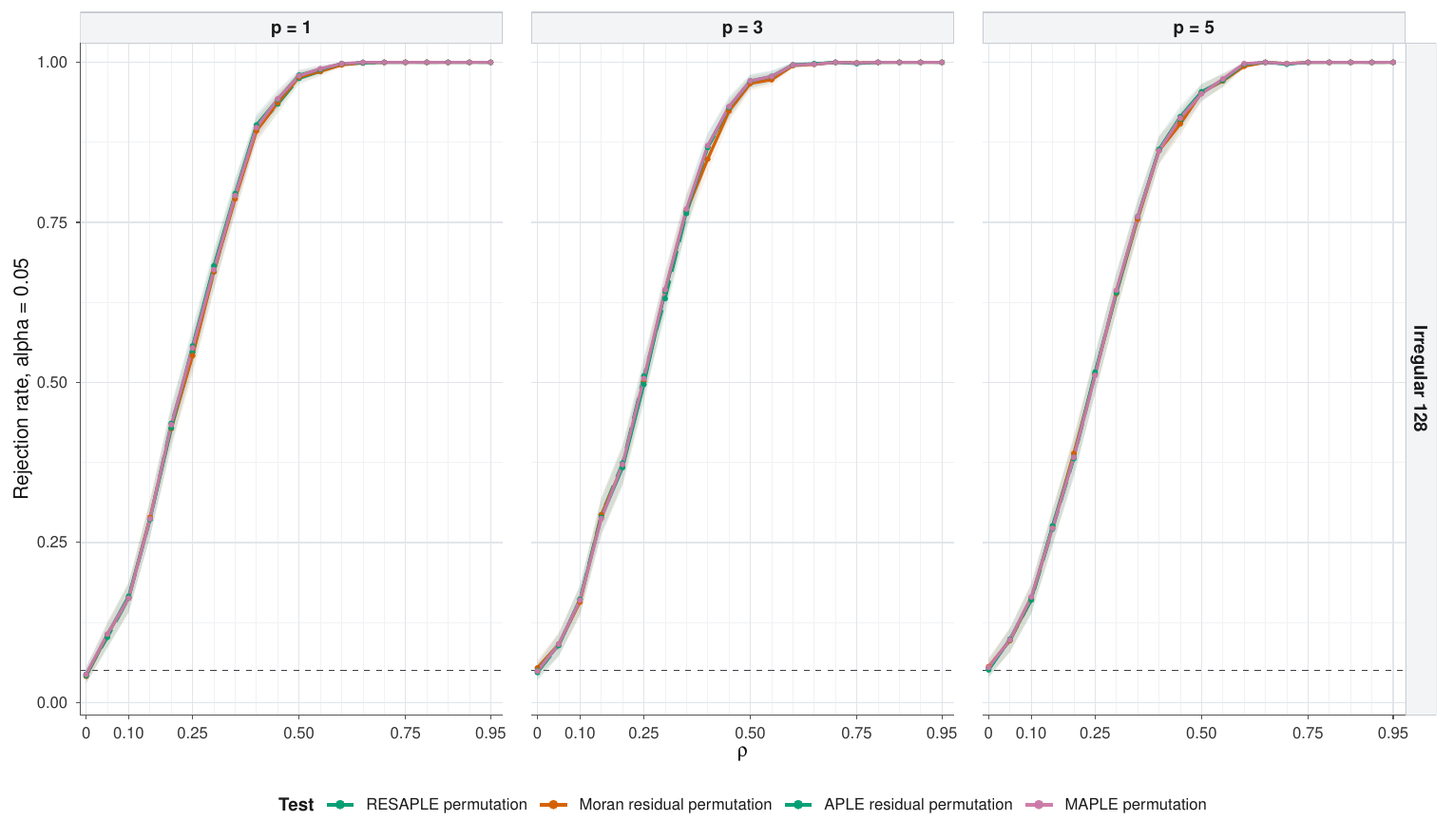}
\caption{\footnotesize Large irregular adjacency graph, $n=128$.}
\label{subfig:sim-testing-permutation-irregular128}
\end{subfigure}
\caption{Permutation rejection rates under the Gaussian SEM. All tests use the same Freedman-Lane residual permutation procedure. The horizontal dashed line marks $\alpha=0.05$.}
\label{fig:sim-testing-permutation}
\end{figure*}

The B07 results are again a stress test. When $p=1$, power increases with $\rho$, but the increase is slow. When $p=3$ or $p=5$, the residual dimension is only $r=5$ or $r=3$, and all four permutation tests have low power. This is not a RESAPLE-specific failure. It reflects the limited residual information available on this graph after covariate adjustment. The large irregular graph gives a different picture. With $n=128$, all four tests have high and very similar power. Thus, the weak B07 performance is mainly a consequence of its small residual space, not irregularity alone.

Figure~\ref{fig:sim-testing-reference} compares the three RESAPLE reference procedures. On the lattice and the large irregular graph, the exact and permutation procedures are very close. The $z$ procedure is also close in these moderate-dimensional settings. On B07, the $z$ procedure can be much more conservative, especially when the residual spectrum is highly discrete. The exact and permutation procedures remain close to each other.

\begin{figure*}[pos=!htpb]
\centering
\begin{subfigure}[t]{0.7\textwidth}
\centering
\includegraphics[width=0.86\linewidth]{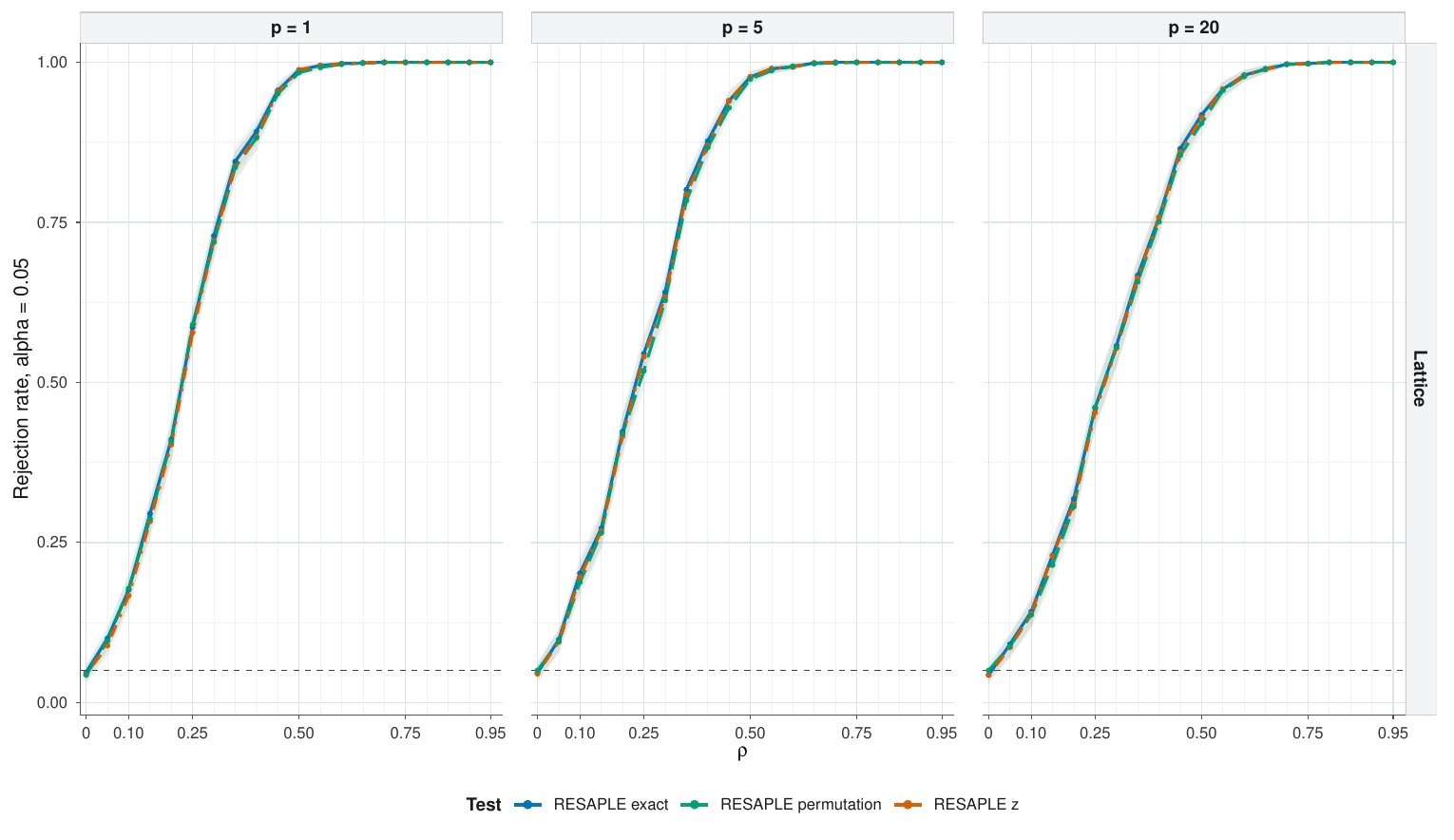}
\caption{\footnotesize Lattice with selected rook weights, $n=100$.}
\label{subfig:sim-testing-reference-lattice}
\end{subfigure}

\vspace{0.8em}

\begin{subfigure}[t]{0.7\textwidth}
\centering
\includegraphics[width=0.86\linewidth]{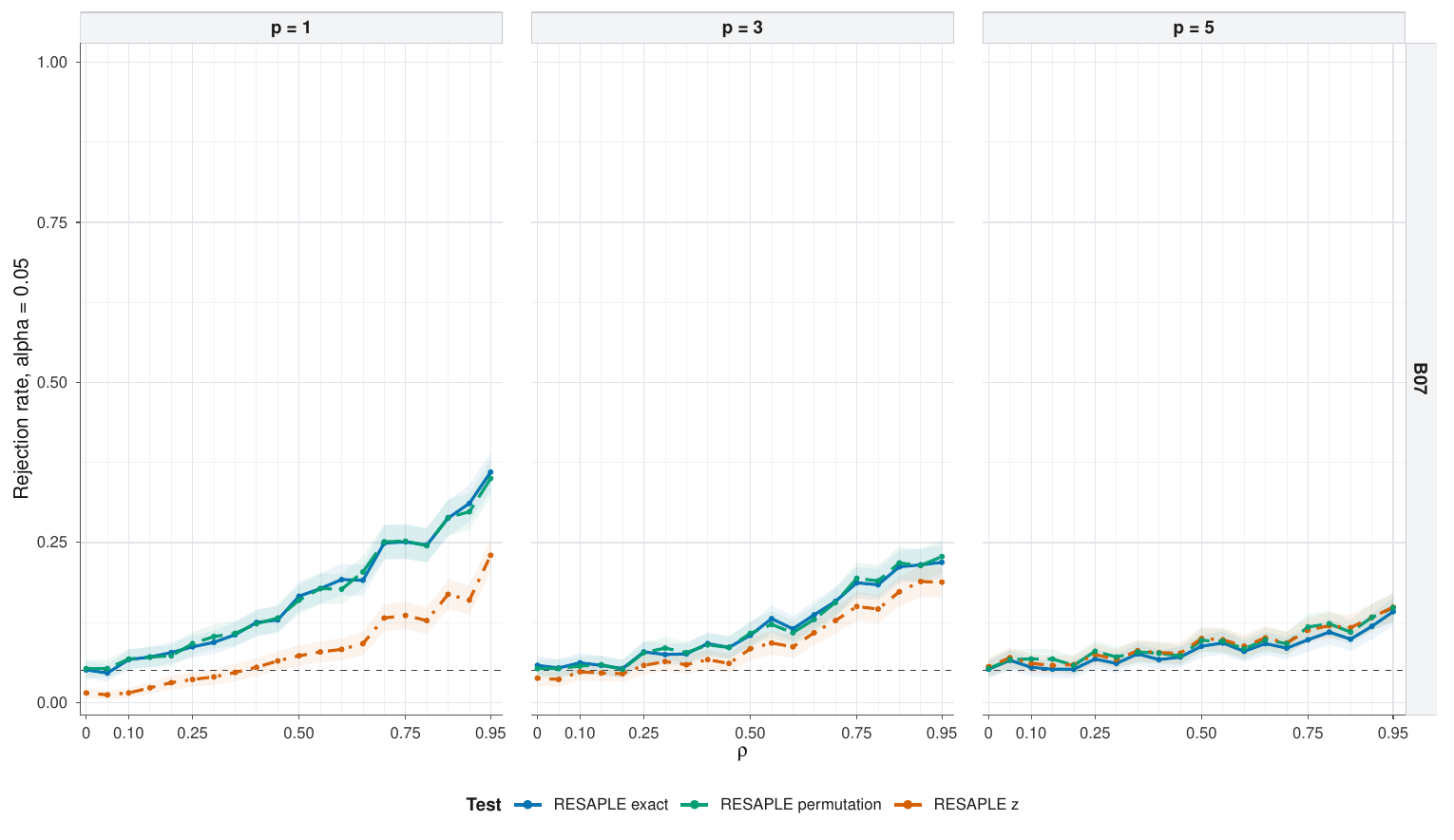}
\caption{\footnotesize B07 adjacency.}
\label{subfig:sim-testing-reference-b07}
\end{subfigure}

\vspace{0.8em}

\begin{subfigure}[t]{0.7\textwidth}
\centering
\includegraphics[width=0.86\linewidth]{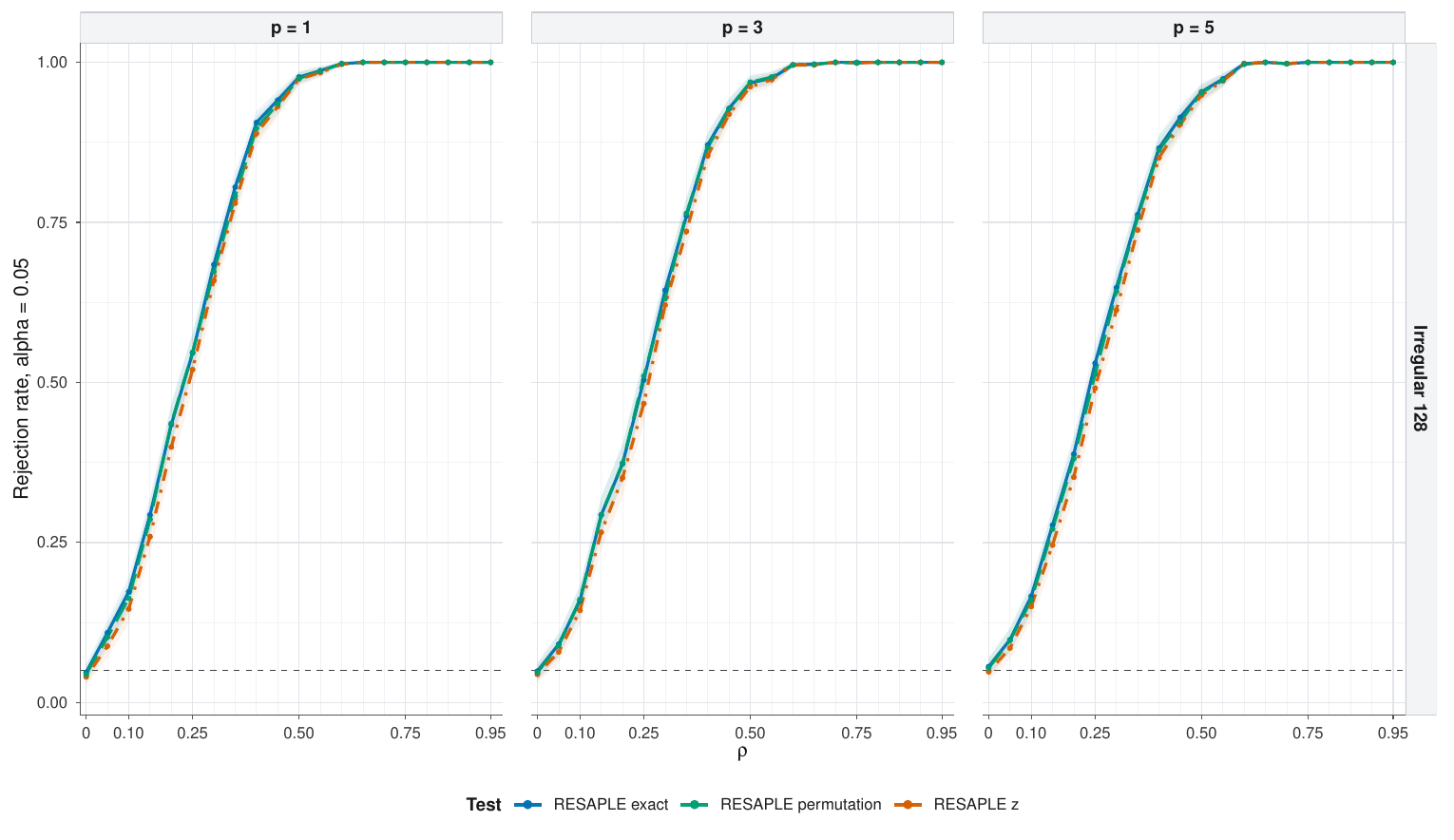}
\caption{\footnotesize Large irregular adjacency graph, $n=128$.}
\label{subfig:sim-testing-reference-irregular128}
\end{subfigure}
\caption{RESAPLE rejection rates under the Gaussian SEM using exact, permutation, and $z$ reference procedures. The horizontal dashed line marks $\alpha=0.05$.}
\label{fig:sim-testing-reference}
\end{figure*}

Table~\ref{tab:sim-testing-size} reports worst absolute size distortion over the plotted covariate dimensions. The permutation tests remain close to the nominal level across all three graph classes. The largest departure in the table is for the RESAPLE $z$ test on B07, where the empirical size is $0.015$ at $p=1$. This supports using the exact or permutation procedure, particularly in small or spectrally irregular residual spaces. 

\begin{table*}[pos=!htpb]
\centering
\caption{Worst absolute size distortion, $\max_p |\widehat{\alpha}-0.05|$, under the Gaussian SEM for the plotted testing designs.}
\footnotesize
\renewcommand{\arraystretch}{1.12}
\setlength{\tabcolsep}{5pt}
\begin{tabular*}{\textwidth}{@{\extracolsep{\fill}} llrrrrrr @{}}
\toprule
Graph & $W$ & RESAPLE exact & RESAPLE perm. & RESAPLE $z$ & Moran perm. & APLE perm. & MAPLE perm. \\
\midrule
Lattice, $n=100$ & rook      & 0.007 & 0.007 & 0.007 & 0.006 & 0.008 & 0.008 \\
B07, $n=8$       & adjacency & 0.008 & 0.003 & 0.035 & 0.007 & 0.008 & 0.009 \\
Irregular, $n=128$ & adjacency & 0.006 & 0.006 & 0.010 & 0.009 & 0.008 & 0.006 \\
\bottomrule
\end{tabular*}

\label{tab:sim-testing-size}
\end{table*}

Overall, the main story here is that RESAPLE gives a clear bias reduction as an estimator, supporting accurate understanding of residual spatial autocorrelation in the data, while simultaneously maintaining comparable power as a test statistic to the residual Moran's $\mathcal I_M$, residual APLE, and MAPLE when all are tested by the same permutation procedure. This latter point is of course expected because the tests are all residual quadratic-form diagnostics for positive spatial dependence. The main practical distinction is therefore the testing procedure used for RESAPLE. The exact and permutation procedures behave reliably across the plotted designs, while the $z$ approximation is less reliable in the smallest residual spaces.

\subsection{Simulation 4: Robustness to Non-Gaussian Errors}\label{subsec:sim-nongaussian}

We examine whether the Gaussian estimation and testing conclusions persist when the error law is changed. The design is the same as in Simulations 2 and 3. We use the selected rook lattice weight with $n=100$ in the main text. The remaining lattice sizes and candidate weights are reported in Appendix~\ref{app:simulations}.

We consider three error laws. The first is the Gaussian baseline, $\epsilon\sim N(0,I_n)$. The second is a spherical multivariate $t$ error with $\nu=5$ degrees of freedom,
$$\epsilon=\sqrt{\frac{\nu-2}{Q}}Z,\qquad Z\sim N(0,I_n),\qquad Q\sim\chi^2_\nu.$$
The same scalar $Q$ is shared across coordinates, so the error vector remains spherical and has marginal variance one. The third is an independent skewed error,
$$\epsilon_i=\frac{Y_i-3}{\sqrt{6}},\qquad Y_i\sim\chi^2_3,$$
which has mean zero and variance one, but is not spherical. In all cases, the spatial error is generated as $U=(I-\rho W)^{-1}\epsilon$.

Figure~\ref{fig:sim-nongaussian-rmse} reports RMSE for the three error laws. The Gaussian and spherical $t$ results are almost identical, consistent with the shared spherical structure of the error law. The skewed case also gives the same qualitative ordering. RESAPLE has the lowest RMSE for most values of $\rho$ in every panel. Residual Moran's $\mathcal I_M$ can have lower RMSE near $\rho=0$ because it has low sampling variation in that local region. As $\rho$ increases, its downward bias dominates. RESAPLE then becomes the more stable estimator.

\begin{figure*}[pos=!htpb]
\centering
\begin{subfigure}[t]{0.7\textwidth}
\centering
\includegraphics[width=0.86\linewidth]{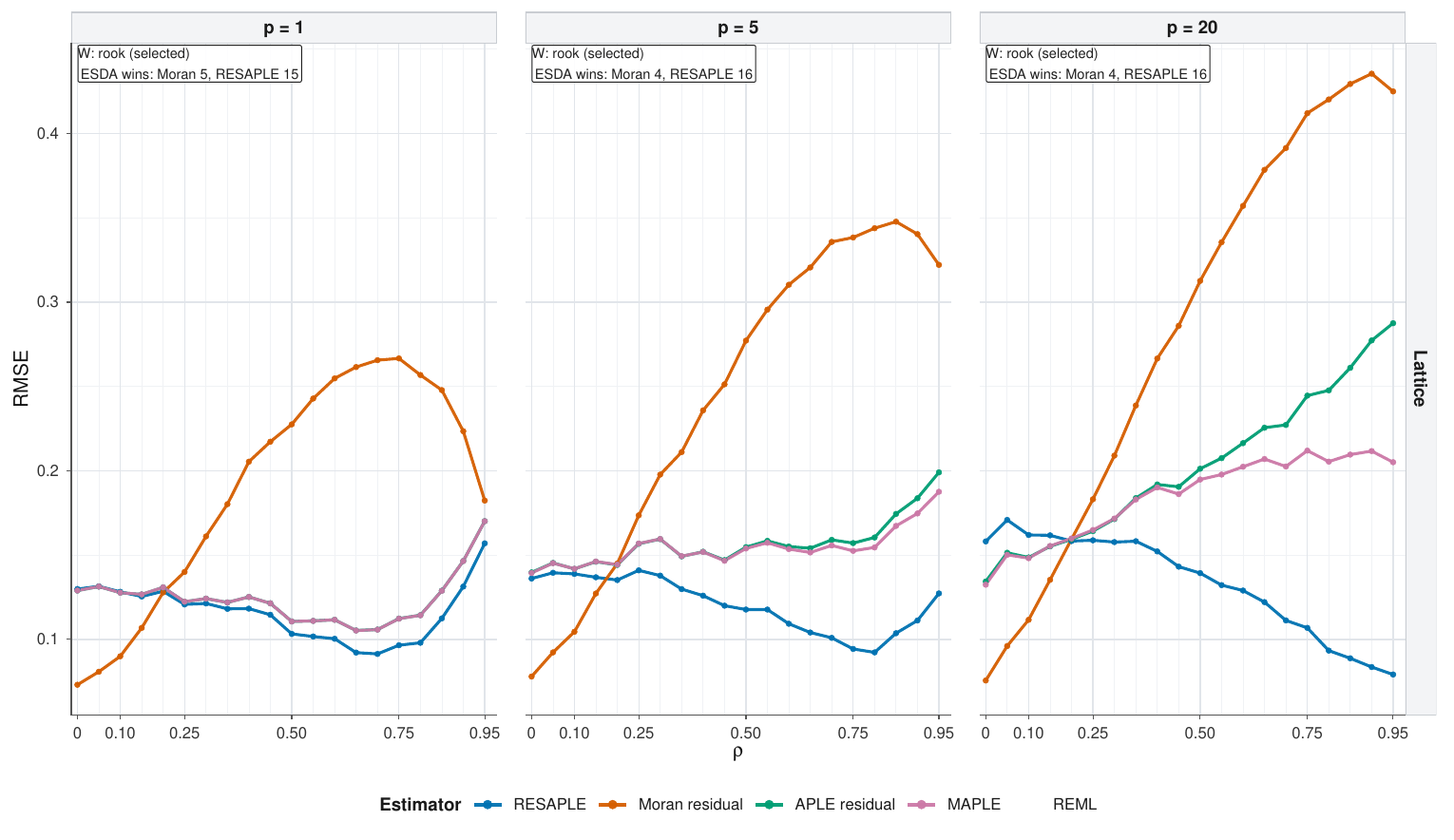}
\caption{\footnotesize Gaussian errors.}
\label{subfig:sim-nongaussian-rmse-gaussian}
\end{subfigure}

\vspace{0.8em}

\begin{subfigure}[t]{0.7\textwidth}
\centering
\includegraphics[width=0.86\linewidth]{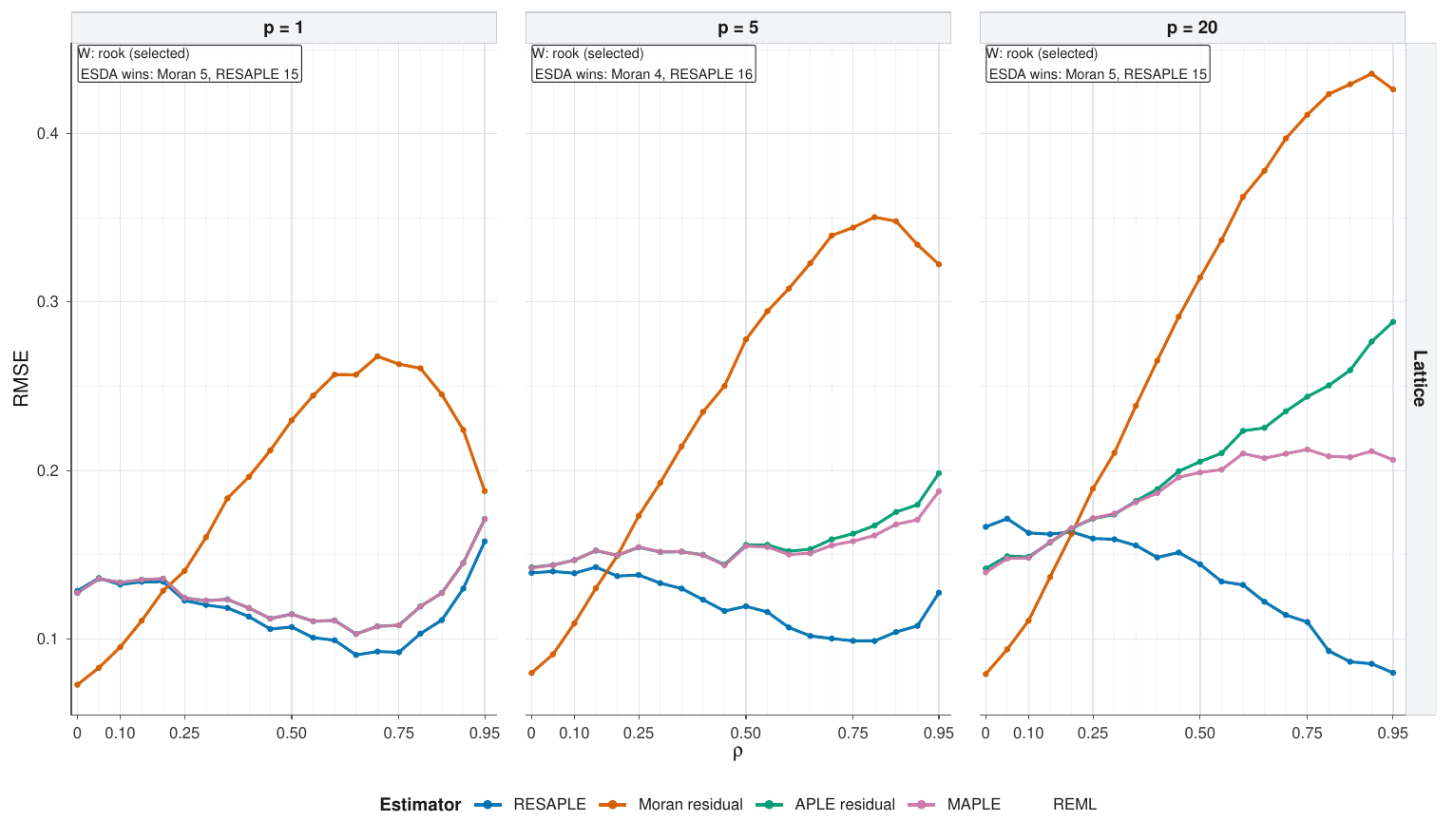}
\caption{\footnotesize Spherical multivariate $t$ errors.}
\label{subfig:sim-nongaussian-rmse-mvt}
\end{subfigure}

\vspace{0.8em}

\begin{subfigure}[t]{0.7\textwidth}
\centering
\includegraphics[width=0.86\linewidth]{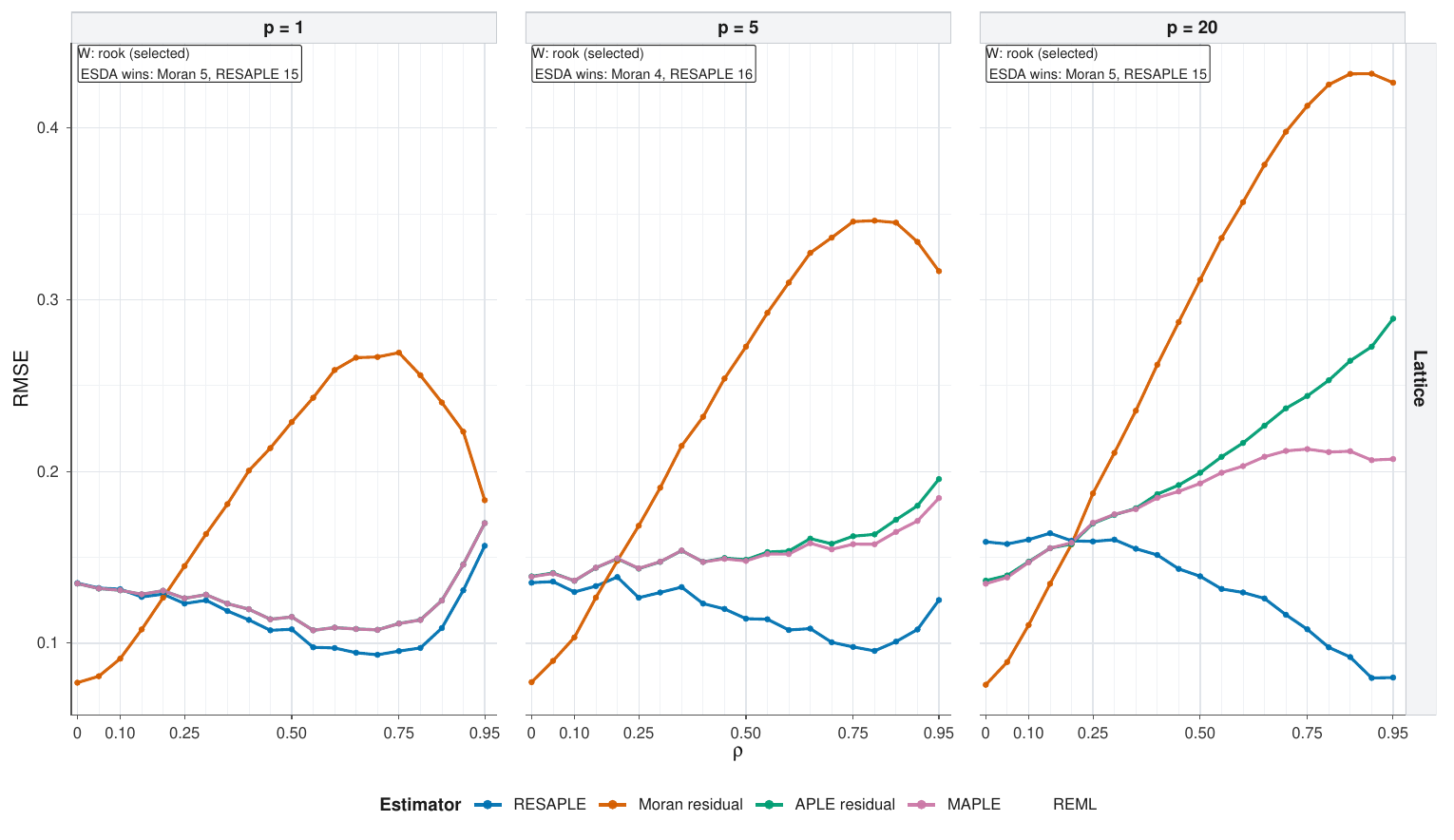}
\caption{\footnotesize Independent skewed errors.}
\label{subfig:sim-nongaussian-rmse-skewed}
\end{subfigure}
\caption{RMSE for estimating $\rho$ under Gaussian, spherical multivariate $t$, and skewed errors. The lattice uses the selected rook weight with $n=100$.}
\label{fig:sim-nongaussian-rmse}
\end{figure*}

The corresponding bias results show the same pattern observed in Simulation 2. RESAPLE has the smallest absolute bias across the plotted $\rho$ values under all three error laws. The remaining RMSE differences near $\rho=0$ are therefore variance effects, not bias effects. This distinction is important for interpreting the skewed case. Skewness changes finite-sample variability, but it does not remove the bias advantage of RESAPLE.

Figure~\ref{fig:sim-nongaussian-power} reports permutation rejection rates. All four ESDA statistics are tested using the same Freedman-Lane residual permutation procedure. The power curves are close under all three error laws. Thus, the non-Gaussian simulations reinforce the conclusions from Simulation 3. RESAPLE is competitive as a test statistic, but its main advantage is estimation bias reduction rather than a large power gain over other residual quadratic-form diagnostics.

\begin{figure*}[pos=!htpb]
\centering
\begin{subfigure}[t]{0.7\textwidth}
\centering
\includegraphics[width=0.86\linewidth]{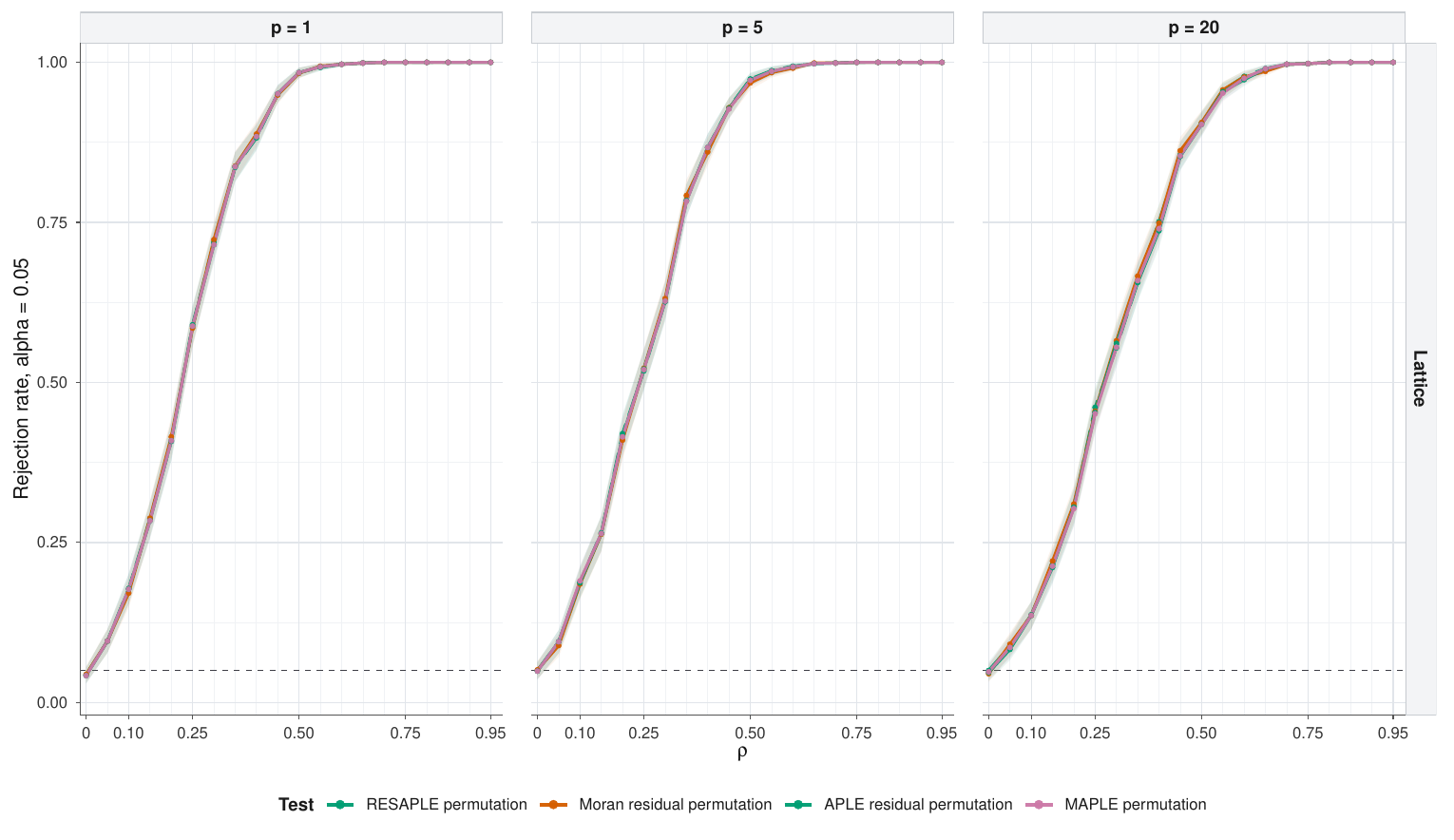}
\caption{\footnotesize Gaussian errors.}
\label{subfig:sim-nongaussian-power-gaussian}
\end{subfigure}

\vspace{0.8em}

\begin{subfigure}[t]{0.7\textwidth}
\centering
\includegraphics[width=0.86\linewidth]{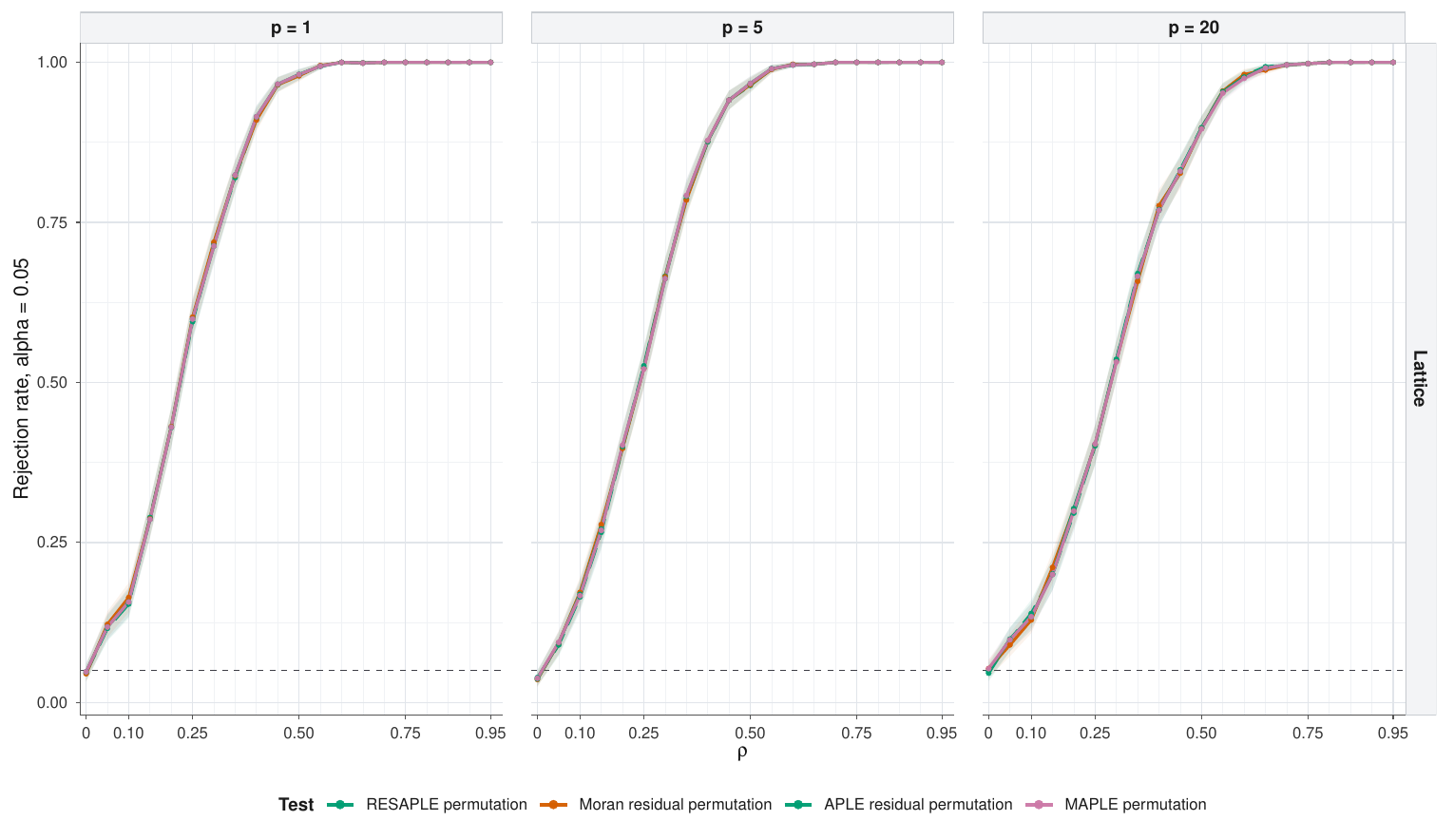}
\caption{\footnotesize Spherical multivariate $t$ errors.}
\label{subfig:sim-nongaussian-power-mvt}
\end{subfigure}

\vspace{0.8em}

\begin{subfigure}[t]{0.7\textwidth}
\centering
\includegraphics[width=0.86\linewidth]{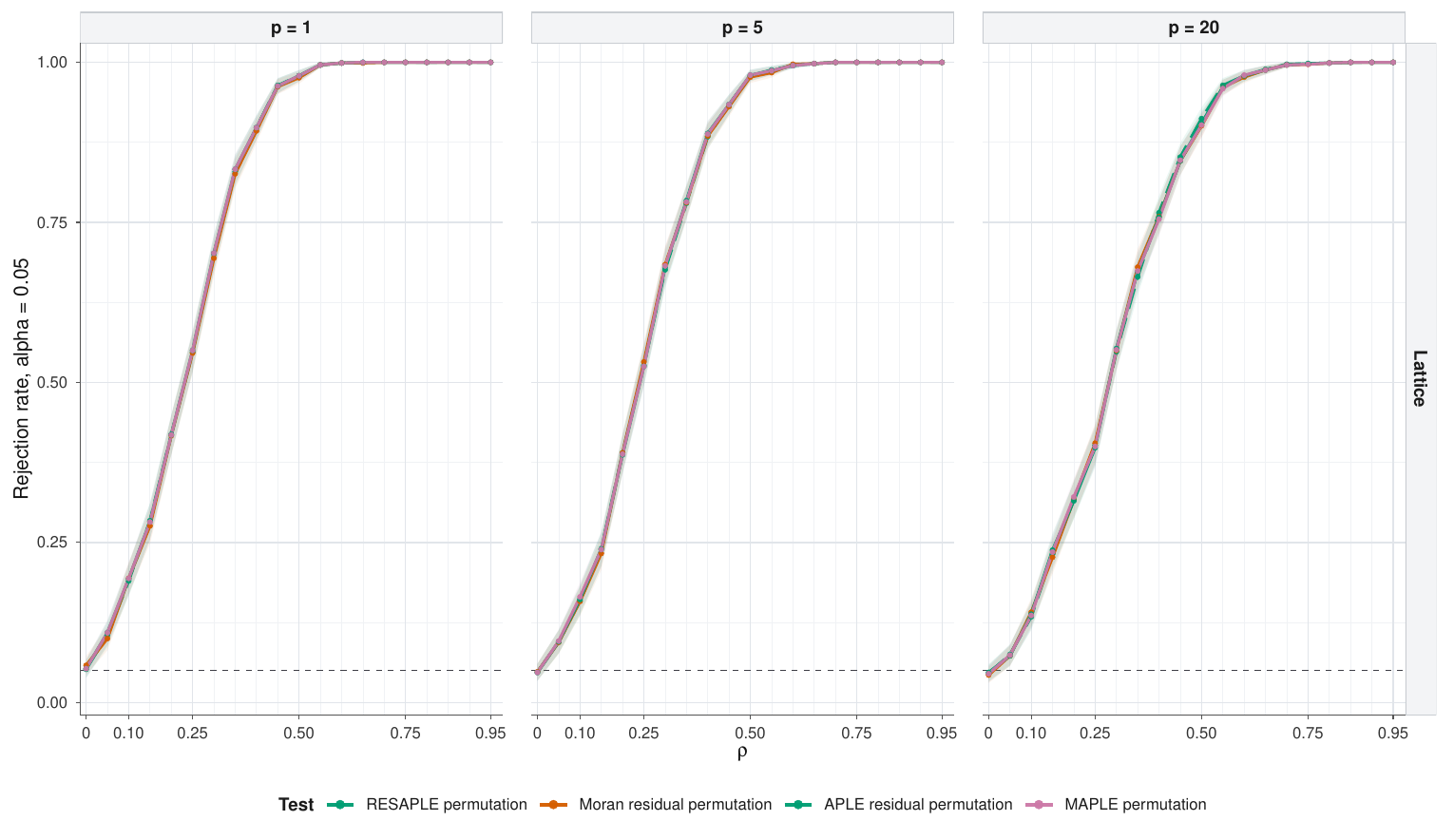}
\caption{\footnotesize Independent skewed errors.}
\label{subfig:sim-nongaussian-power-skewed}
\end{subfigure}
\caption{Permutation rejection rates under Gaussian, spherical multivariate $t$, and skewed errors. All tests use the same Freedman-Lane residual permutation procedure. The horizontal dashed line marks $\alpha=0.05$.}
\label{fig:sim-nongaussian-power}
\end{figure*}

Table~\ref{tab:sim-nongaussian-size} reports worst absolute size distortion over $p\in\{1,5,20\}$ for the permutation tests in Figure~\ref{fig:sim-nongaussian-power}. The empirical sizes remain close to the nominal level. The largest departure is $0.014$, for residual Moran's $\mathcal I_M$ under spherical $t$ errors. For RESAPLE, the largest departure is $0.011$ under spherical $t$ errors and $0.005$ under skewed errors.

\begin{table*}[pos=!htpb]
\centering
\footnotesize
\renewcommand{\arraystretch}{1.12}
\setlength{\tabcolsep}{6pt}
\begin{tabular*}{\textwidth}{@{\extracolsep{\fill}} lrrrr @{}}
\toprule
Error law & RESAPLE perm. & Moran perm. & APLE perm. & MAPLE perm. \\
\midrule
Gaussian & 0.007 & 0.006 & 0.008 & 0.008 \\
Spherical $t$ & 0.011 & 0.014 & 0.013 & 0.012 \\
Skewed & 0.005 & 0.008 & 0.003 & 0.005 \\
\bottomrule
\end{tabular*}
\caption{Worst absolute size distortion, $\max_p |\widehat{\alpha}-0.05|$, for the permutation tests under the non-Gaussian robustness designs. The lattice uses the selected rook weight with $n=100$.}
\label{tab:sim-nongaussian-size}
\end{table*}

Overall, the robustness experiment supports the main simulation conclusions. The spherical $t$ case behaves like the Gaussian case, as expected from the shared spherical structure. The skewed case changes the error law more substantially, but the same broad pattern remains. RESAPLE provides the most consistent bias reduction and generally the lowest RMSE, while permutation power remains close across the residual ESDA statistics.

\section{Real Data Case Study}\label{sec:real-data}
We illustrate the use of RESAPLE on a real areal dataset consisting of U.S. Census tracts in King County, Washington (USA). For each tract, we obtained 2019--2023 American Community Survey (ACS) 5-year estimates for median household income and tract-level demographic and socioeconomic covariates\footnote{\url{https://www.census.gov/data/developers/data-sets/acs-5year.html}}, and we merged these attributes with tract polygons\footnote{\url{https://www.census.gov/geographies/mapping-files/time-series/geo/tiger-line-file.html}}. Both the attributes and geometries were obtained from the \texttt{tidycensus} package in \textsf{R} \citep{walkerTidycensusLoadUS2017}. The ACS 5-year product provides model-based small-area estimates with associated margins of error; in this case study we treat the published estimates as observed inputs for the purpose of demonstrating spatial-residual diagnostics.

We examine natural-log-transformed median household income (MHHI) $Z_i = \log(\text{MHHI}_i)$ for census tract $i=1,\dots,n$ ($n=494$) in King County, constructed from ACS table B19013. A map is provided in Appendix \ref{app:case-study}. Candidate covariates include: 1) the share of adults with at least a bachelor's degree, constructed from educational-attainment counts, 2) unemployment rate, 3) renter share, 4) median age, 5) poverty rate, and 6) household vacancy rate. Variables were joined by tract GEOID.

The objective of our analysis is to assess the extent of residual spatial dependence in log median household income after progressively richer mean model adjustment. Recall that our model assumption is that $$Z = X\beta + U,\quad U = \rho WU + \epsilon.$$ For each choice of $X$, the SEM treats the remaining variation after the mean structure as the spatial error component. We then assess whether residual spatial autocorrelation persists in this component. Because RESAPLE is computed conditional on the chosen design matrix $X$, changing $X$ changes what is attributed to the mean structure versus the residual spatial component. Because there is no known uniquely correct detrending specification in practice, we used a ``ladder'' approach of increasing linear mean model flexibility to assess sensitivity of spatial residual diagnostics to the extent of covariate adjustment. Specifically, we consider the sequence of linear predictors described in Table \ref{tab:model-ladder}, ranging from intercept only (M0), through socioeconomic covariates (M1--M2), and finally adding spatial trend with centroid coordinates (M3--M4). This lets us compare how RESAPLE behaves relative to alternatives as $p$ changes, and hence as $r=n-p$ changes.

\begin{table}[t]
\centering
\caption{Model ladder used in the case study. $\texttt{cx\_s},\texttt{cy\_s}$ denote standardised tract-centroid $(x, y)$ coordinates.}
\small
\setlength{\tabcolsep}{7pt}
\renewcommand{\arraystretch}{1.15}
\begin{tabular}{@{}ll@{}}
\toprule
Model & Mean structure for $Z_i$ (additive terms in $\eta_i = x_i^\top\beta$) \\
\midrule
M0 &
Intercept only: $1$ \\
M1 &
$1 + \texttt{bach\_plus} + \texttt{unemp} + \texttt{renter} + \texttt{age}$ \\
M2 &
M1 $+\ \texttt{poverty} + \texttt{vacancy}$ \\
M3 &
M2 $+\ \texttt{cx\_s} + \texttt{cy\_s}$ \\
M4 &
M3 $+\ \texttt{cx\_s}^2 + \texttt{cy\_s}^2 + (\texttt{cx\_s}\cdot \texttt{cy\_s})$ \\
\bottomrule
\end{tabular}

\label{tab:model-ladder}
\end{table}

For each specification of design matrix $X$, we computed the OLS residuals $\hat u = M_X Z$, where $M_X=I_n-X(X^\top X)^{-1}X^\top$, and applied global and local spatial diagnostics to assess whether residual spatial autocorrelation persisted after adjustment. This required specifying a spatial weights matrix $W$, which is typically not uniquely determined (Section \ref{sec:impact-W}). We considered the candidate set $\mathcal W = \{\texttt{rook}, \texttt{queen}, \texttt{knn4}, \texttt{knn6}, \texttt{knn8}\}$. For each $(X,W)$ pair we computed $\mathcal I_r(0)$ and selected $W$ by maximising the information over $\mathcal W$. Using the selected $W$, we then report Moran's $\mathcal I_M$, APLE, MAPLE, RESAPLE, and the REML SEM estimate of $\rho$. The results are summarised in Table \ref{tab:ladder-results}, and visualised in Figure \ref{fig:rho-indices}.

\begin{table}[t]
\centering
\caption{Global dependence summaries across the model ladder. $p$ is the number of columns in $X$ (including the intercept). $W$ is the chosen candidate weights matrix, here \texttt{knn4}. $I_r(0)$ denotes the restricted null information used for $W$ selection. Moran and APLE are computed on OLS residuals; MAPLE is the covariate-adjusted ML one-step estimator; RESAPLE is the restricted residual-space one-step estimator; REML is the maximiser of the restricted profile likelihood over $\rho$.}
\small
\setlength{\tabcolsep}{6pt}
\renewcommand{\arraystretch}{1.15}
\begin{tabular}{@{}l r l r r r r r r@{}}
\toprule
Model & $p$ & $W$ & $I_{r}(0)$ & Moran & APLE & MAPLE & RESAPLE & REML \\
\midrule
M0 & 1  & knn4 & 215.3634 & 0.5837950 & 0.6823412 & 0.6850960 & 0.6755014 & 0.7031647 \\
M1 & 5  & knn4 & 209.2293 & 0.3800045 & 0.5113108 & 0.5349882 & 0.5254836 & 0.5517835 \\
M2 & 7  & knn4 & 206.5113 & 0.1871191 & 0.3261134 & 0.3596274 & 0.3518962 & 0.4250109 \\
M3 & 9  & knn4 & 202.9122 & 0.1462524 & 0.2724848 & 0.3112651 & 0.3104243 & 0.4039483 \\
M4 & 12 & knn4 & 198.0635 & 0.1260997 & 0.2427351 & 0.2803249 & 0.2925572 & 0.3906871 \\
\bottomrule
\end{tabular}

\label{tab:ladder-results}
\end{table}

\begin{figure}[pos=!htp]
    \centering
    \includegraphics[width=\linewidth]{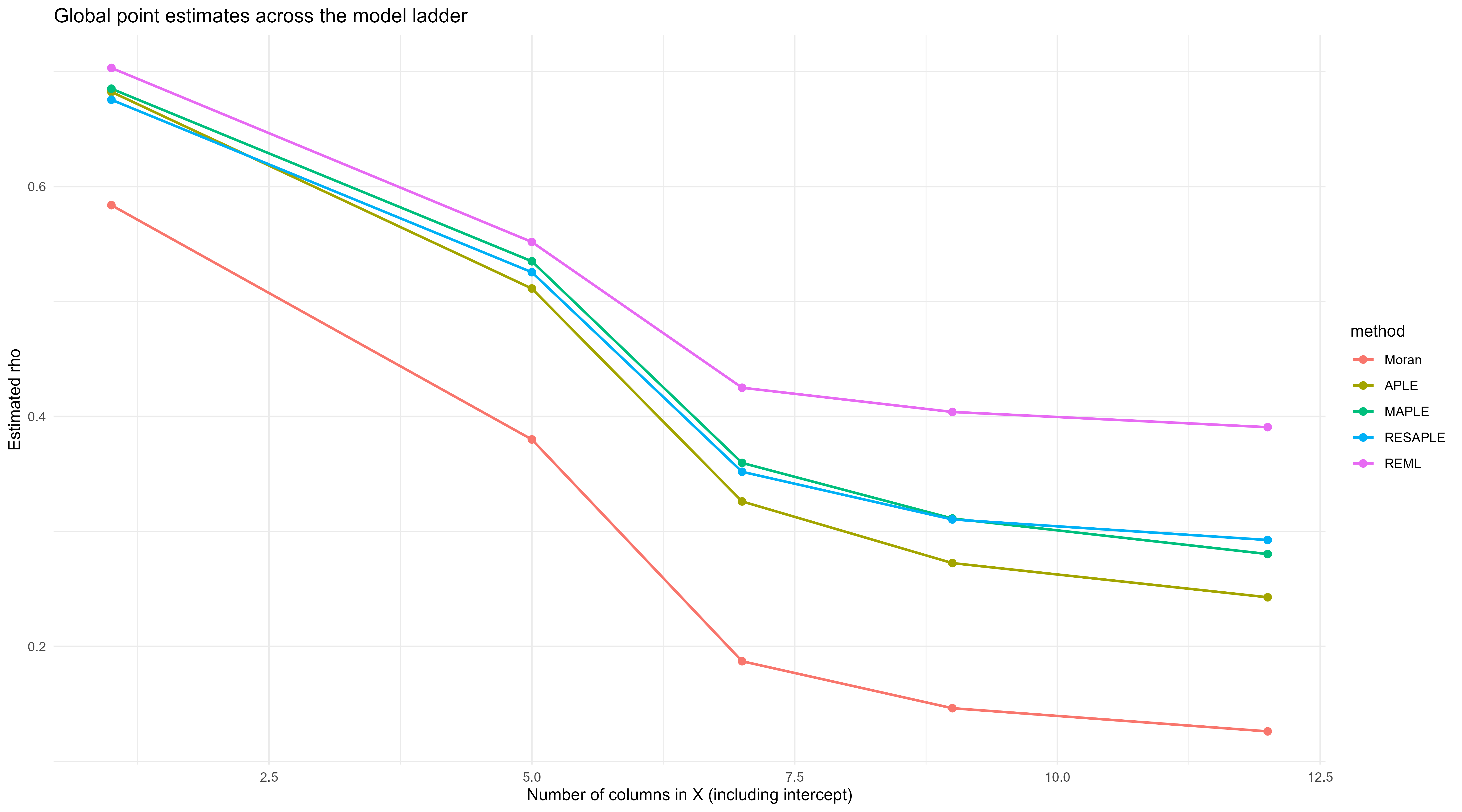}
    \caption{Global estimates of $\rho$ across the model ladder for the selected 4-nearest-neighbour weights.}
    \label{fig:rho-indices}
\end{figure}

Because RESAPLE is derived as a one-step approximation to the REML score equation, it is natural to compare it with the REML estimate as a likelihood-based reference. Table \ref{tab:ladder-results} shows that all methods estimate positive residual spatial dependence under every specification. The largest reduction occurs when socioeconomic covariates are introduced (M0 to M2), after which the estimates decrease more gradually as smooth spatial trend terms are added. Thus, while observed socioeconomic characteristics explain an important component of the spatial pattern in log MHHI, appreciable residual spatial dependence remains even under the richest mean specification (M4).

The one-step estimators all follow the same qualitative pattern as REML. MAPLE is closest to REML for M0--M3, while RESAPLE is closest for M4, where the residual degrees of freedom are smallest because the mean model contains the largest amount of predictors ($p$). The difference between MAPLE and RESAPLE is small relative to the overall decline across the model ladder. Since the true $\rho$ is unknown in this application, REML is used here only as a likelihood-based reference rather than as a gold standard. The case study illustrates that RESAPLE provides a closed-form residual-space estimate that tracks the likelihood-based REML estimator closely while simultaneously supporting the ESDA diagnostics developed in this study.

\begin{figure}[pos=!htbp]
  \centering
  \begin{subfigure}[t]{0.48\textwidth}
    \centering
    \includegraphics[width=\linewidth]{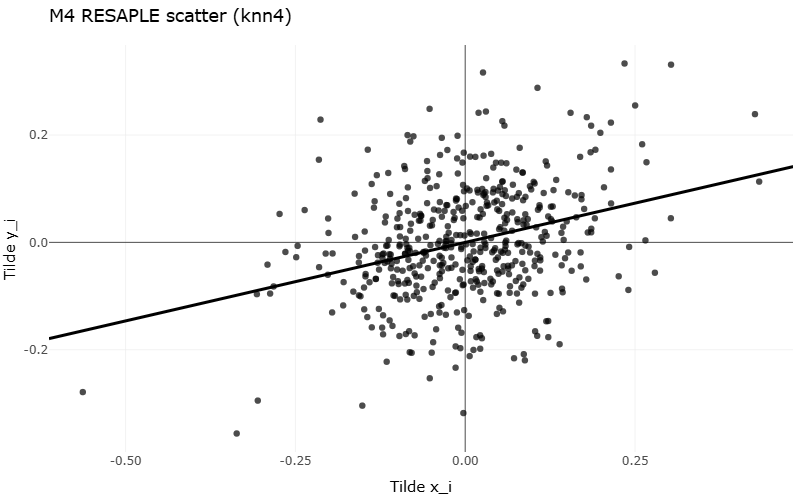}
    \caption{RESAPLE scatterplot.}
    \label{fig:RESAPLE-scatterplot}
  \end{subfigure}\hfill
  \begin{subfigure}[t]{0.48\textwidth}
    \centering
    \includegraphics[width=\linewidth]{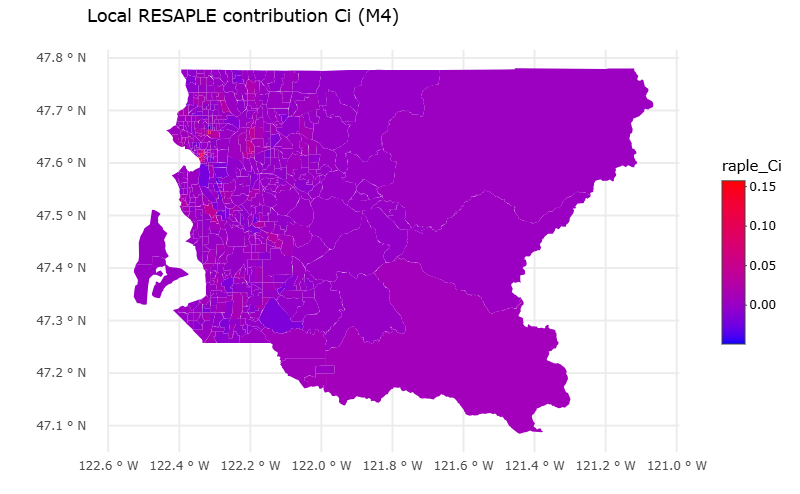}
    \caption{Map of local RESAPLE contributions $C_i$.}
    \label{fig:local-RESAPLE}
  \end{subfigure}
  \caption{Additional ESDA visualisations associated with the RESAPLE statistic for model M4.}
  \label{fig:esda-vis}
\end{figure}

Taking M4 as representative, we plotted the RESAPLE scatterplot and a map of the local RESAPLE contributions $C_i$ in Figure \ref{fig:esda-vis}. The scatterplot and the local map describe the same decomposition of the global RESAPLE statistic. Most tracts have local contributions close to zero, which is reflected in both the scatterplot and on the map by values clustered near zero. The largest positive mapped contributions are concentrated in the western and north-western urban tracts, including the Seattle area. These are contributions to residual spatial dependence rather than maps of income itself. They indicate where the fitted M4 residual structure contributes most strongly to the remaining positive spatial dependence. This shows that RESAPLE provides  not only as a global summary of residual spatial dependence, but also as a coherent framework for residual space scatterplots and local ESDA diagnostics.

\section{Conclusion}\label{sec:conclusion}
In this paper, we introduced RESAPLE, a one-step restricted likelihood estimator for the spatial dependence parameter in Gaussian SEMs with covariates. RESAPLE preserves the computational and interpretive advantages of APLE-type statistics through a closed-form ratio of quadratic forms, while remaining invariant to reparameterisations of the design matrix and to the choice of orthonormal residual basis. Most importantly, it targets the setting encountered in many applied spatial statistical workflows: diagnosing and testing residual spatial dependence \emph{after} detrending data via a mean model. 

Beyond defining the estimator, we developed theory parallel to the ML-based APLE literature in the restricted likelihood setting. This framework yields an analytic expression for the restricted Fisher information at the null, $\mathcal I_r(0)$, which provides a principled measure of local informativeness for comparing candidate spatial weight matrices after covariate adjustment. More broadly, this restricted-information perspective provides guidance for resolving a common practical ambiguity in spatial data analysis: how the choice of the spatial weights affects the detectability of residual spatial dependence. 

RESAPLE is designed as a residual-space estimator that also provides a flexible ESDA framework. In our simulation studies, RESAPLE consistently reduced estimation bias, generally achieved lowest RMSE, and maintained testing performance comparable to existing residual ESDA statistics across a range of sample sizes, covariate dimensions, and spatial topologies.  We further demonstrated how RESAPLE can be used in practice for residual diagnostics, scatterplot visualisation and mapping using a real data case study on median household income in King County, Washington, USA. 

There are several directions for future research that we may consider. First, it would be valuable to extend the restricted one-step construction to other spatial models, including SAR-lag specifications and models with additional variance components. Second, further work is needed to further characterize robustness to broader forms of distributional misspecification,  heteroscedasticity, and uncertainty in small area estimates. Finally, while $\mathcal I_r(0)$ provides an analytic criterion for comparing candidate weights locally, developing a systematic weight selection strategy that balances power, calibration, and interpretability remains an important topic for future work. 

\section*{Acknowledgements}
Dr. Meredith Franklin's contributions to this work were supported by the Natural Sciences and Engineering Research Council of Canada (NSERC) RGPIN-2024-06444.


\appendix

\section{Additional Discussion of Section \ref{sec:model-likelihood}}\label{app:model-likelihood}

\subsection{Covariance Structure Under the Gaussian SEM}\label{app:sem-cov}
Recall the SEM $$Z = X\beta + U,\qquad U = \rho WU + \epsilon,\qquad \epsilon \sim N(0,\sigma^2 I_n),$$ and write $R(\rho)=I_n-\rho W$. Under the standing assumption that $\det R(\rho)\neq 0$, we may solve the error equation to obtain $$U = R(\rho)^{-1}\epsilon.$$ It follows immediately that $$U \sim N\big(0, \sigma^2 R(\rho)^{-1}R(\rho)^{-\top}\big).$$ Consequently, the response satisfies $$ Z \sim N\big(X\beta, \sigma^2 \Sigma(\rho)\big),\qquad \Sigma(\rho)=R(\rho)^{-1}R(\rho)^{-\top}.$$

\subsection{Log-likelihood Under the Gaussian SEM}\label{app:loglik}
Let $\Sigma(\rho)=R(\rho)^{-1}R(\rho)^{-\top}$ so that $Z\sim N(X\beta,\sigma^2\Sigma(\rho))$ and $\Sigma(\rho)^{-1}=R(\rho)^\top R(\rho)$. Moreover, $$|\Sigma(\rho)| = |R(\rho)^{-1}R(\rho)^{-\top}| = |R(\rho)|^{-2}.$$
Therefore, up to the constant $C_n=-\frac{n}{2}\log(2\pi)$, the Gaussian log-likelihood is
\begin{align*}
\ell(\beta,\rho,\sigma^2)
&= C_n -\frac{1}{2}\log|\sigma^2\Sigma(\rho)| - \frac{1}{2}(Z-X\beta)^\top(\sigma^2\Sigma(\rho))^{-1}(Z-X\beta) \\
&= C_n + \log|R(\rho)| - \frac{n}{2}\log\sigma^2 - \frac{1}{2\sigma^2}(Z-X\beta)^\top R(\rho)^\top R(\rho)(Z-X\beta).
\end{align*}

\subsection{The SEM as a Gaussian LMM and the REML Viewpoint}\label{app:reml-view}
The SEM can be viewed as a Gaussian linear model with a covariance parameter. Specifically, the representation $$Z = X\beta + U,\qquad U \sim N\big(0,\sigma^2\Sigma(\rho)\big)$$
shows that $(\rho,\sigma^2)$ act as variance parameters governing the covariance structure of the errors. This motivates the usage of the standard restricted maximum likelihood (REML) framework for Gaussian models with fixed effects, as our ultimate goal is to do inference on the variance components of the model. In particular, REML eliminates $\beta$ by basing inference for $(\rho,\sigma^2)$ on linear combinations of $Z$ that remove the contribution of $X$, which aligns with the standard of detrending data before conducting ESDA.

For completeness, we recall the residual-space construction. Let $$P = X(X^\top X)^{-1}X^\top,\qquad M = I_n-P,$$ and let $H\in\R^{n\times r}$, $r=n-p$, satisfy $H^\top H=I_r$ and $HH^\top=M$. The residualised contrasts are $e=H^\top Z$.

\subsection{Proof of Lemma \ref{lem:error-dist}}\label{app:proof-lem-error}
\begin{proof}
For the first point, note that $MX=0$ by definition of $M$. Since $HH^\top=M$, we have $HH^\top X = MX = 0$. Multiplying on the left by $H^\top$ and using $H^\top H = I_r$ yields $$H^\top X = H^\top H H^\top X = H^\top(HH^\top X)=0.$$ Therefore, $$e = H^\top Z = H^\top(X\beta + U) = H^\top X\beta + H^\top U = H^\top U.$$

For the second point, noting that $U = R(\rho)^{-1}\epsilon$, we have $$e = H^\top U = H^\top R(\rho)^{-1}\epsilon.$$ Since $\epsilon \sim N(0,\sigma^2 I_n)$ and $H^\top R(\rho)^{-1}$ is deterministic conditional on $(X,W)$, it follows that $e$ is Gaussian with mean zero and covariance
\begin{align*}
\Cov(e)
&= \sigma^2 H^\top R(\rho)^{-1} I_n R(\rho)^{-\top} H \\
&= \sigma^2 H^\top R(\rho)^{-1}R(\rho)^{-\top} H
= \sigma^2 \Sigma_r(\rho),
\end{align*}
where $\Sigma_r(\rho)=H^\top R(\rho)^{-1}R(\rho)^{-\top}H$.
\end{proof}

\section{Additional Discussion of Section \ref{sec:RESAPLE-estimator}}\label{app:RESAPLE-estimator}

\subsection{Restricted Likelihood and Profiling}
\begin{lem}\label{lem:restricted-likelihood}
Under the SEM model, the restricted likelihood for $(\rho, \sigma^2)$ up to an additive constant is $$\ell_r(\rho, \sigma^2;Z) = \log|R(\rho)| - \frac{1}{2}\log|X^\top R(\rho)^\top R(\rho)X| - \frac{r}{2}\log \sigma^2 - \frac{1}{2\sigma^2}Z^\top P_{R^\top(\rho)R(\rho)}Z,$$
where
\begin{align*}
P_{R^\top(\rho)R(\rho)}
&= R(\rho)^\top R(\rho)
- R(\rho)^\top R(\rho) X\big(X^\top R(\rho)^\top R(\rho)X\big)^{-1} X^\top R(\rho)^\top R(\rho).
\end{align*}
Equivalently in terms of the residual contrasts $e=H^\top Z$, we may write
\begin{align*}
\ell_r(\rho, \sigma^2;e) &= C_r - \frac{r}{2}\log \sigma^2 - \frac{1}{2}\log\big|\Sigma_r(\rho)\big| - \frac{1}{2\sigma^2}e^\top \Sigma_r(\rho)^{-1}e,
\end{align*}
where $\Sigma_r(\rho)=H^\top R(\rho)^{-1}R(\rho)^{-\top}H$ as in Lemma \ref{lem:error-dist}, and $C_r$ is an additive constant that does not depend on $(\rho,\sigma^2)$.
\end{lem}

\begin{proof}
Let $V(\rho,\sigma^2)=\sigma^2\Sigma(\rho)=\sigma^2R(\rho)^{-1}R(\rho)^{-\top}$. For a Gaussian linear model $Z\sim N(X\beta,V)$ with $V$ positive definite and $X$ full rank, the restricted log-likelihood for the covariance parameters, up to an additive constant, is $$\ell_R(V;Z) = -\frac{1}{2}\Big(\log|V|+\log|X^\top V^{-1}X|+Z^\top P_V Z\Big),$$ where $$P_V = V^{-1} - V^{-1}X\big(X^\top V^{-1} X\big)^{-1}X^\top V^{-1}.$$ In our setting, $V^{-1}=\sigma^{-2}\Sigma(\rho)^{-1}=\sigma^{-2}R(\rho)^\top R(\rho)$, and $|V|=(\sigma^2)^n|\Sigma(\rho)|.$ Since $|\Sigma(\rho)|=|R(\rho)^{-1}R(\rho)^{-\top}|=|R(\rho)|^{-2}$, we have $$\log|V| = n\log\sigma^2 - 2\log|R(\rho)|.$$ Moreover, $$X^\top V^{-1}X = \sigma^{-2}X^\top R(\rho)^\top R(\rho)X, \qquad \log|X^\top V^{-1}X| = -p\log\sigma^2 + \log|X^\top R(\rho)^\top R(\rho)X|.$$ Finally, $P_V=\sigma^{-2}P_{R^\top(\rho)R(\rho)}$. Substituting these identities into $\ell_R(V;Z)$ and simplifying yields our asserted form for $\ell_r(\rho,\sigma^2;Z)$, up to an additive constant. For the representation in terms of $e$, note from Lemma \ref{lem:error-dist} that $e=H^\top Z$ is Gaussian with mean zero and covariance $\sigma^2\Sigma_r(\rho)$, where $\Sigma_r(\rho)=H^\top R(\rho)^{-1}R(\rho)^{-\top}H$. The restricted likelihood is invariant (up to an additive constant not depending on $(\rho,\sigma^2)$) to replacing $Z$ by any full-rank set of error contrasts that annihilate $X$; hence it may be written as the Gaussian log-likelihood of $e$, which yields the stated form for $\ell_r(\rho,\sigma^2;e)$.
\end{proof}

\begin{lem}\label{lem:profile-likelihood}
For the restricted log-likelihood in $e$ given by $$\ell_r(\rho, \sigma^2;e) = -\frac{r}{2}\log \sigma^2 - \frac{1}{2}\log|\Sigma_r(\rho)|- \frac{1}{2\sigma^2}e^\top \Sigma_r(\rho)^{-1}e,$$ the profile log-likelihood (up to additive constants) is $$\ell_r^p(\rho;e) = -\frac{1}{2}\log|\Sigma_r(\rho)| -\frac{r}{2}\log\Big(\frac{1}{r}e^\top \Sigma_r(\rho)^{-1}e\Big).$$
\end{lem}

\begin{proof}
Differentiating $\ell_r(\rho,\sigma^2;e)$ with respect to $\sigma^2$ yields $$\pdv{\sigma^2}\ell_r(\rho, \sigma^2;e) = -\frac{r}{2\sigma^2} + \frac{1}{2\sigma^{4}}e^\top\Sigma_r(\rho)^{-1}e.$$ Setting this derivative to zero and solving yields $$\hat \sigma_r^2(\rho) = \frac{1}{r}e^\top\Sigma_r(\rho)^{-1}e.$$ Substituting $\hat\sigma_r^2(\rho)$ into $\ell_r(\rho,\sigma^2;e)$ and collecting terms gives the stated profile likelihood, up to an additive constant.
\end{proof}

\subsection{Proof of Lemma \ref{lem:score}}
\begin{proof}
For the first point, at $\rho=0$ we have $R(0)=I_n$ and therefore $\Sigma(0)=I_n$. Hence $\Sigma_r(0)=H^\top I_n H=I_r$.

For the second point, note that $(I_n-\rho W)^{-1}=I_n+\rho W+\bigO(\rho^2)$ as $\rho\to 0$. Therefore,
\begin{align*}
\Sigma(\rho)&= R(\rho)^{-1}R(\rho)^{-\top} = \big(I_n+\rho W+\bigO(\rho^2)\big)\big(I_n+\rho W^\top+\bigO(\rho^2)\big) \\
&= I_n + \rho(W+W^\top) + \bigO(\rho^2),
\end{align*}
so that $\Sigma'(0)=W+W^\top$. It follows that $$\Sigma_r'(0)=H^\top\Sigma'(0)H = H^\top(W+W^\top)H = W_r + W_r^\top = 2K_r.$$

For the third point, write $Q(\rho)=e^\top\Sigma_r(\rho)^{-1}e$, so that $$\ell_r^p(\rho;e) = -\frac{1}{2}\log|\Sigma_r(\rho)| - \frac{r}{2}\log Q(\rho) + \text{const}.$$ Then $$S_r(\rho) = -\frac{1}{2}\pdv{\rho}\log|\Sigma_r(\rho)| - \frac{r}{2}\pdv{\rho}\log Q(\rho).$$ Using $\pdv{\rho}\log|\Sigma_r(\rho)|=\Tr\big(\Sigma_r(\rho)^{-1}\Sigma_r'(\rho)\big)$ gives $\pdv{\rho}\log|\Sigma_r(\rho)|\big|_{\rho=0}=\Tr(2K_r).$ Moreover, we have that $\pdv{\rho}\Sigma_r(\rho)^{-1}=-\Sigma_r(\rho)^{-1}\Sigma_r'(\rho)\Sigma_r(\rho)^{-1},$ so $$Q'(0) = -e^\top \Sigma_r'(0)e = -2e^\top K_r e, \qquad Q(0)=e^\top e,$$ and therefore $$ \pdv{\rho}\log Q(\rho)\big|_{\rho=0} = \frac{Q'(0)}{Q(0)} = -2\frac{e^\top K_r e}{e^\top e}.$$ Substituting these identities into the expression for $S_r(\rho)$ yields $$S_r(0) = -\Tr(K_r) + r\frac{e^\top K_r e}{e^\top e}.$$
\end{proof}

\subsection{Proof of Theorem \ref{thm:RESAPLE-def}}
\begin{proof}
We prove each point in turn.

For the first point, let $\tilde H$ be another choice satisfying $\tilde H\tilde H^\top=M$, $\tilde H^\top \tilde H=I_r$, and $\tilde H^\top X=0$. Then there exists an orthogonal matrix $Q\in\R^{r\times r}$ such that $\tilde H = HQ$. Define
$$\tilde e = \tilde H^\top Z = Q^\top e,\qquad \tilde W_r = \tilde H^\top W\tilde H = Q^\top W_r Q,\qquad \tilde K_r = \frac{1}{2}(\tilde W_r+\tilde W_r^\top)=Q^\top K_r Q.$$
The trace is cyclic, so it immediately follows that $\Tr(\tilde K_r)=\Tr(K_r)$ and $\Tr(\tilde K_r^2)=\Tr(K_r^2)$, so $\mu_r$ and $\omega_r$ are unchanged. Moreover, $\tilde B_r=\tilde K_r^2+\omega_r I_r=Q^\top B_r Q$. Therefore,
\begin{align*}
\tilde e^\top(\tilde K_r-\mu_r I_r)\tilde e &= e^\top Q(Q^\top K_r Q-\mu_r I_r)Q^\top e = e^\top(K_r-\mu_r I_r)e,
\end{align*}
and similarly $\tilde e^\top \tilde B_r \tilde e = e^\top B_r e.$ This proves invariance to the choice of $H$.

For the second point, under $\rho=0$ the model reduces to $Z=X\beta+\epsilon$ with $\epsilon\sim N(0,\sigma^2 I_n)$, so $e=H^\top\epsilon\sim N(0,\sigma^2 I_r)$. Note $\expect_0[e^\top A e]=\sigma^2\Tr(A)$. Taking $A=K_r-\mu_r I_r$ and using $\mu_r=\Tr(K_r)/r$ gives $\Tr(A)=0$ and hence $\expect_0[e^\top(K_r-\mu_r I_r)e]=0$. If $K_r\neq0$, then $\Tr(K_r^2)>0$, so $\omega_r>0$ and $B_r=K_r^2+\omega_r I_r\succ0$. Taking $A=B_r$ gives
$$\expect_0[e^\top B_r e] =\sigma^2\Tr(B_r)=\sigma^2\{\Tr(K_r^2)+r\omega_r\}=2\sigma^2\Tr(K_r^2).$$
By Lemma \ref{lem:restricted-info}, $2\Tr(K_r^2)=\Tr(W_r^\top W_r)+\Tr(W_r^2)$.

For the third and fourth points, consider the one-step approximation to the restricted score equation at $\rho=0$, $$S_r(0)+\mathcal I_r^A(0)\rho=0,$$ with $\mathcal I_r^A(0)$ as defined in Theorem \ref{thm:RESAPLE-def}. By Lemma \ref{lem:score}, $$S_r(0)=\frac{r}{e^\top e}e^\top(K_r-\mu_r I_r)e.$$ Hence,
\begin{align*}
-\frac{S_r(0)}{\mathcal I_r^A(0)} &= -\frac{\frac{r}{e^\top e}e^\top(K_r-\mu_r I_r)e}{-\frac{r}{e^\top e}e^\top B_r e} = \frac{e^\top(K_r-\mu_r I_r)e}{e^\top B_r e} =\hat\rho_{RESAPLE},
\end{align*}
which proves (3) and (4).
\end{proof}

\subsection{Proof of Corollary \ref{cor:design-invariance}}
\begin{proof}
Let $\tilde X=XQ$ with $Q$ invertible. Then $ \Im(\tilde X)= \Im(X)$, so the orthogonal projectors satisfy $\tilde P=P$ and $\tilde M=M$. Any admissible $\tilde H$ with $\tilde H\tilde H^\top=M$ and $\tilde H^\top\tilde H=I_r$ differs from $H$ by an orthogonal transformation. The conclusion follows from the first point of Theorem \ref{thm:RESAPLE-def}.
\end{proof}

\subsection{Motivating the Approximate Curvature and Restricted Fisher Information}
To motivate $\mathcal I_r^A(0)$, it is helpful to discuss the Fisher information for a covariance parameter in a mean zero Gaussian model.

\begin{lem}\label{lem:fisher-formula}
Let $y \sim N(0, C(\theta))$ where $C(\theta)$ is positive definite and twice continuously differentiable in $\theta \in \R$. Then the Fisher information for $\theta$ is $$\mathcal I(\theta) = \frac{1}{2}\Tr\Big(C(\theta)^{-1}C_\theta(\theta)C(\theta)^{-1}C_\theta(\theta)\Big),$$ where $C_\theta(\theta)=\pdv{\theta}C(\theta)$.
\end{lem}

\begin{proof}
The log-likelihood is $\ell(\theta;y) = -\frac{1}{2}\log|C(\theta)| - \frac{1}{2}y^\top C(\theta)^{-1}y + \text{const}$. The score is
\begin{align*}
S(\theta;y)
&= \pdv{\theta}\ell(\theta;y) = -\frac{1}{2}\Tr\big(C^{-1}C_\theta\big) + \frac{1}{2}y^\top C^{-1}C_\theta C^{-1}y,
\end{align*}
and $\expect_\theta[S(\theta;y)]=0$. For a zero mean Gaussian vector, $\Var(y^\top Ay)=2\Tr(ACA C)$ for any deterministic, symmetric $A$. Taking $A=\frac{1}{2}C^{-1}C_\theta C^{-1}$ yields
\begin{align*}
\Var\big(S(\theta;y)\big) &= \frac{1}{4}\Var\big(y^\top C^{-1}C_\theta C^{-1}y\big)
= \frac{1}{2}\Tr\big(C^{-1}C_\theta C^{-1}C_\theta\big),
\end{align*}
which equals the Fisher information.
\end{proof}

\begin{lem}\label{lem:restricted-info}
For $e=H^\top Z$ under the SEM model, let $C(\rho)=\sigma^2\Sigma_r(\rho)$. Then the Fisher information for $\rho$ at $\rho=0$ is $$\mathcal I_r(0) = 2\Tr(K_r^2)=\Tr(W_r^\top W_r) + \Tr(W_r^2).$$
\end{lem}

\begin{proof}
By Lemma \ref{lem:fisher-formula}, $$\mathcal I_r(0) = \frac{1}{2}\Tr\Big(C(0)^{-1}C_\rho(0)C(0)^{-1}C_\rho(0)\Big).$$
Using $C(0)=\sigma^2 I_r$ and $C_\rho(0)=\sigma^2\Sigma_r'(0)=\sigma^2(2K_r)$ gives
\begin{align*}
\mathcal I_r(0) &= \frac{1}{2}\Tr\big((2K_r)^2\big)= 2\Tr(K_r^2).
\end{align*}
Since $K_r=\frac{1}{2}(W_r+W_r^\top)$,
\begin{align*}
4\Tr(K_r^2) &= \Tr\big((W_r+W_r^\top)^2\big)
= \Tr(W_r^2)+\Tr(W_rW_r^\top)+\Tr(W_r^\top W_r)+\Tr((W_r^\top)^2) \\
&= 2\Tr(W_r^2)+2\Tr(W_r^\top W_r),
\end{align*}
and dividing by two yields $\mathcal I_r(0)=\Tr(W_r^2)+\Tr(W_r^\top W_r)$.
\end{proof}

The expression $\mathcal I_r(0)=\Tr(W_r^\top W_r)+\Tr(W_r^2)=2\Tr(K_r^2)$ provides an analytic characterization of the population curvature at $\rho=0$. Denote the denominator of $\hat\rho_{RESAPLE}$ by $$D(e)=e^\top B_r e,\qquad B_r=K_r^2+\omega_r I_r,\qquad \omega_r=\frac{\Tr(K_r^2)}{r}.$$ From Theorem \ref{thm:RESAPLE-def}, $\expect_0[D(e)]=\sigma^2\mathcal I_r(0)$. Since $\sigma^2$ is generally unknown, a natural estimator at $\rho=0$ is $\hat\sigma_r^2(0)=\frac{1}{r}e^\top e$, and this motivates the stabilised curvature approximation $$\mathcal I_r^A(0) = -\frac{r}{e^\top e}D(e),$$ which is approximately proportional to $-\mathcal I_r(0)$ after substituting $\hat\sigma_r^2(0)=\frac1r e^\top e$ for $\sigma^2$, since $\frac{r}{e^\top e}\approx \frac{1}{\sigma^2}$ under $\rho=0$ and $D(e)\approx \sigma^2\mathcal I_r(0)$.

Having now discussed the approximate curvature in detail, there is also the natural question of why this quantity is used at all in the formula for the RESAPLE, rather than the exact restricted information directly. There are two related reasons. First, the exact restricted Fisher information at $\rho=0$ is a scalar. Using it directly would give the one-step statistic
$$\frac{e^\top(K_r-\mu_r I_r)e}{(e^\top e/r)\mathcal I_r(0)},$$
which is locally valid but no longer has the APLE-type denominator that adapts to the realised spatially filtered residual signal. 
Second, the exact observed curvature of the restricted profile likelihood contains additional trace and quadratic-form terms from differentiating $\Sigma_r(\rho)^{-1}$ and $\log|\Sigma_r(\rho)|$. This quantity is useful for numerical REML optimisation, but it is less suitable for ESDA because it does not generally yield a simple positive Rayleigh quotient with a stable scatterplot and local contribution decomposition.

\subsection{A Law of Large Numbers for Quadratic Forms}
\begin{lem}\label{lem:quadratic-lln}
Let $e \sim N(0, \sigma^2 I_r)$ and let $A_r$ be a sequence of deterministic symmetric $r \times r$ matrices with uniformly bounded operator norm. Then, as $r \to \infty$, $$\frac{1}{r}\Big(e^\top A_r e - \sigma^2 \Tr(A_r)\Big) \pconv 0.$$
\end{lem}

\begin{proof}
We have $\expect[e^\top A_r e]=\sigma^2\Tr(A_r)$ and $\Var(e^\top A_r e)=2\sigma^4\Tr(A_r^2)$. Uniform boundedness of $\norm{A_r}_{op}$ implies $\Tr(A_r^2)\leq r\norm{A_r}_{op}^2 \leq Cr$ for some $C>0$. Therefore, $$\Var\Big(\frac{1}{r}\big(e^\top A_r e - \sigma^2 \Tr(A_r)\big)\Big) \leq \frac{2\sigma^4 C}{r}\to 0,$$
and the claim follows by Chebyshev's inequality.
\end{proof}

\begin{thm}\label{thm:curvature-consistency}
Suppose that, as $r \to \infty$, the following hold:
\begin{enumerate}
    \item $\sup_r \norm{W_r}_{op} < \infty$,
    \item $\mathcal I_r(0)/r \to k > 0$.
\end{enumerate}
Then
$$-\frac{\mathcal I_r^A(0)}{\mathcal I_r(0)} \pconv 1.$$
\end{thm}

\begin{proof}
By Lemma \ref{lem:quadratic-lln} with $A_r=I_r$, $(e^\top e)/(r\sigma^2)\pconv 1$, so $(r/e^\top e)\pconv 1/\sigma^2$.
Let $D(e)=e^\top B_r e$, where $B_r=K_r^2+\omega_r I_r$ and $\omega_r=\Tr(K_r^2)/r$. Since $\norm{K_r}_{op}\leq \norm{W_r}_{op}$, the first assumption implies $\sup_r\norm{K_r}_{op}<\infty$. The second assumption implies $\Tr(K_r^2)/r=\mathcal I_r(0)/(2r)\to k/2$, and hence $\sup_r\norm{B_r}_{op}<\infty$. By Lemma \ref{lem:quadratic-lln} with $A_r=B_r$,
$$\frac{1}{r}\Big(D(e)-\sigma^2\Tr(B_r)\Big)\pconv 0.$$
Since $\Tr(B_r)=\Tr(K_r^2)+r\omega_r=2\Tr(K_r^2)=\mathcal I_r(0)$, we have $D(e)/(\sigma^2\mathcal I_r(0))\pconv 1$, i.e. $D(e)/\mathcal I_r(0)\pconv \sigma^2$.
Therefore,
$$-\frac{\mathcal I_r^A(0)}{\mathcal I_r(0)} =\frac{r}{e^\top e}\cdot \frac{D(e)}{\mathcal I_r(0)}\pconv \frac{1}{\sigma^2}\cdot \sigma^2 = 1. $$
\end{proof}

\section{Additional Discussion of Section \ref{subsec:RESAPLE-testing}}\label{app:RESAPLE-testing}
Here we first review the exact Gaussian null distribution for $\hat\rho_{RESAPLE}$. We then show that the same null distribution holds for a spherically symmetric SEM.

Assume first the Gaussian SEM and $H_0:\rho=0$, so that $e\sim N(0,\sigma^2 I_r)$. Define $$A_r \defeq K_r-\mu_r I_r, \qquad B_r \defeq K_r^2+\omega_r I_r,$$ so that $\hat\rho_{RESAPLE} = (e^\top A_r e)/(e^\top B_r e)$. Fix $t\in\R$. If $K_r\neq0$, then $B_r\succ0$, and the event $\{\hat\rho_{RESAPLE}\geq t\}$ is equivalent to $$e^\top (A_r-tB_r)e \geq 0.$$ Let $D_t\defeq A_r-tB_r$, and write the eigendecomposition $D_t = U_t \Lambda_t U_t^\top$, where $U_t$ is orthogonal and $\Lambda_t=\diag(\lambda_{t,1},\dots,\lambda_{t,r})$. Since $U_t^\top e \eqd e$ under the Gaussian null, $$e^\top D_t e \eqd \sum_{j=1}^r \lambda_{t,j} \sigma^2 \chi_{1,j}^2,$$ where $\chi_{1,1}^2,\dots,\chi_{1,r}^2$ are independent $\chi^2_1$ random variables. Therefore, $$\prob_0\big(\hat\rho_{RESAPLE}\geq t\big)=\prob\Big(\sum_{j=1}^r \lambda_{t,j}\chi_{1,j}^2 \geq 0\Big).$$ The distribution function of such linear combinations of independent chi-square variables can be evaluated by characteristic function inversion, for instance using Imhof's method \citep{imhof1961}. This is exact under the Gaussian null.

We now extend this result to a spherically symmetric SEM. The argument follows the invariant testing logic used by \citet{kingRobustTestsSpherical1980}. It is also aligned with the invariant formulation of residual spatial autocorrelation tests in linear regression discussed by \citet{martellosio2010power}.

\begin{thm}[Exact null distribution under a spherically symmetric SEM]\label{thm:RESAPLE-spherical-null}
Let $X\in\R^{n\times p}$ have full column rank, let $r=n-p$, and let $H\in\R^{n\times r}$ satisfy $H^\top H=I_r$ and $HH^\top=M$, where $M=I_n-X(X^\top X)^{-1}X^\top$. Let $W_r=H^\top W H$, $K_r=(W_r+W_r^\top)/2$, $\mu_r=\Tr(K_r)/r$, $\omega_r=\Tr(K_r^2)/r$, $A_r=K_r-\mu_r I_r$, and $B_r=K_r^2+\omega_r I_r$. Suppose $K_r\neq0$. Consider the spherically symmetric SEM
$$Z=X\beta+U,\qquad U=\rho WU+\sigma\xi,$$
where $\beta\in\R^p$, $\sigma>0$, and $\xi$ is spherically symmetric in $\R^n$. Under $H_0:\rho=0$, assume $\prob(H^\top\xi=0)=0$. Let $G\sim N(0,I_r)$. Then
$$\frac{e^\top A_r e}{e^\top B_r e} \eqd \frac{G^\top A_r G}{G^\top B_r G}, \qquad e=H^\top Z.$$
Consequently, for every $t\in\R$, if $D_t=A_r-tB_r$ has eigenvalues $\lambda_{t,1},\dots,\lambda_{t,r}$, then
$$\prob_0\big(\hat\rho_{RESAPLE}\geq t\big)=\prob\Big(\sum_{j=1}^r \lambda_{t,j}\chi_{1,j}^2 \geq 0\Big),$$
where $\chi_{1,1}^2,\dots,\chi_{1,r}^2$ are independent $\chi^2_1$ random variables.
\end{thm}

\begin{proof}
Under $H_0:\rho=0$, the spherically symmetric SEM gives
$$Z=X\beta+\sigma\xi.$$
Since $HH^\top=M$ and $M$ is the orthogonal projector onto $\Im(X)^\perp$, we have $H^\top X=0$. Hence
$$e=H^\top Z=H^\top(X\beta+\sigma\xi)=\sigma H^\top\xi.$$
Write $Y=H^\top\xi$, so that $e=\sigma Y$.

We first show that $Y$ is spherically symmetric in $\R^r$. Let $Q\in\R^{r\times r}$ be orthogonal, and define $$\widetilde Q=H Q H^\top+P,$$ where $P=X(X^\top X)^{-1}X^\top$. Since $PH=0$, $H^\top P=0$, $H^\top H=I_r$, and $HH^\top+P=M+P=I_n$,
$$\widetilde Q^\top\widetilde Q=(H Q^\top H^\top+P)(H Q H^\top+P)=H Q^\top Q H^\top+P=HH^\top+P=I_n.$$
Thus $\widetilde Q$ is orthogonal. Since $\xi$ is spherically symmetric, $\widetilde Q\xi\eqd\xi$. Multiplying by $H^\top$ gives
$$H^\top\widetilde Q\xi \eqd H^\top\xi.$$
Also,
$$H^\top\widetilde Q\xi=H^\top(H Q H^\top+P)\xi=QH^\top\xi.$$
Therefore $QY\eqd Y$ for every orthogonal $Q\in\R^{r\times r}$, so $Y$ is spherically symmetric in $\R^r$.

Next, $B_r\succ0$. Since $K_r$ is symmetric, $K_r^2$ is positive semidefinite. Since $K_r\neq0$, $\Tr(K_r^2)>0$, so $\omega_r>0$. Thus, for any nonzero $v\in\R^r$,
$$v^\top B_rv=v^\top K_r^2v+\omega_r v^\top v=\norm{K_rv}^2+\omega_r\norm{v}^2>0.$$

Define $$T(y)=\frac{y^\top A_r y}{y^\top B_r y}$$ for $y\neq0$. Since $B_r\succ0$, this is well defined. For any scalar $c\neq0$,
$$T(cy)=\frac{c^2 y^\top A_r y}{c^2 y^\top B_r y}=T(y).$$
Hence $T$ is homogeneous of degree zero, and since $e=\sigma Y$ with $\sigma>0$, $T(e)=T(Y)$.

By assumption, $\prob(Y=0)=0$. Since $Y$ is spherically symmetric, $Y/\norm{Y}$ is uniformly distributed on the unit sphere in $\R^r$. This follows from the standard result that the law of $Y/\norm{Y}$ is invariant under all orthogonal transformations, and the unit sphere has a unique probability measure with this invariance property. Similarly, if $G\sim N(0,I_r)$, then $G/\norm{G}$ is uniformly distributed on the same unit sphere. Therefore,
$$\frac{Y}{\norm{Y}} \eqd \frac{G}{\norm{G}}.$$
Using the homogeneity of $T$,
$$T(Y)=T\Big(\frac{Y}{\norm{Y}}\Big) \eqd T\Big(\frac{G}{\norm{G}}\Big)=T(G).$$
Since $T(e)=T(Y)$,
$$\frac{e^\top A_r e}{e^\top B_r e} \eqd \frac{G^\top A_r G}{G^\top B_r G}.$$

It remains to express the Gaussian reference law. Since $B_r\succ0$, for every nonzero $g\in\R^r$,
$$T(g)\geq t \quad \iff \quad g^\top A_r g \geq t g^\top B_r g \quad \iff \quad g^\top(A_r-tB_r)g\geq0.$$
Let $D_t=A_r-tB_r$. Since $D_t$ is symmetric, write $D_t=U_t\Lambda_t U_t^\top$, where $U_t$ is orthogonal and $\Lambda_t=\diag(\lambda_{t,1},\dots,\lambda_{t,r})$. Since $G\sim N(0,I_r)$ and $U_t$ is orthogonal, $U_t^\top G\sim N(0,I_r)$. Writing $U_t^\top G=(G_1,\dots,G_r)^\top$, the variables $G_1,\dots,G_r$ are independent $N(0,1)$ variables. Hence
$$G^\top D_tG=G^\top U_t\Lambda_t U_t^\top G=\sum_{j=1}^r \lambda_{t,j}G_j^2 \eqd \sum_{j=1}^r \lambda_{t,j}\chi_{1,j}^2.$$
Combining this with $T(e)\eqd T(G)$ gives
$$\prob_0\big(\hat\rho_{RESAPLE}\geq t\big)=\prob\Big(\sum_{j=1}^r \lambda_{t,j}\chi_{1,j}^2 \geq 0\Big).$$
This completes the proof.
\end{proof}

\section{Additional Discussion of Section \ref{sec:impact-W}}\label{app:impact-W}

\subsection{Proof of Theorem \ref{thm:W-local-info}}
\begin{proof}
Recall that $K_r=H^\top K H$ with $HH^\top=M$ and $H^\top H=I_r$. Since $K$ is symmetric and $M$ is a symmetric idempotent projector, $$\Tr(K_r^2)=\Tr\big(H^\top K H H^\top K H\big)=\Tr\big(HH^\top K HH^\top K\big)=\Tr(MKMK),$$ which yields \eqref{eq:Ir-basis-free}.

We next prove \eqref{eq:Ir-upper}. Since $M$ is an orthogonal projector, the Frobenius norm is contractive under left and right multiplication by $M$. Hence $$\Tr(K_r^2)=\Tr\big((MKM)^2\big)=\|MKM\|_F^2\leq \|K\|_F^2=\Tr(K^2).$$ Multiplying by $2$ gives $\mathcal I_r(0)\leq 2\Tr(K^2)$. Finally, since $K=(W+W^\top)/2$, $$2\Tr(K^2)=\frac{1}{2}\Tr\big((W+W^\top)^2\big)=\Tr(W^\top W)+\Tr(W^2),$$ which proves \eqref{eq:Ir-upper}.

We now derive \eqref{eq:degree-identities} for $W=D^{-1}G$. Since $w_{ij}=g_{ij}/d_i$ and $g_{ij}^2=g_{ij}$, we have
$$\Tr(W^\top W)=\sum_{i=1}^n\sum_{j=1}^n w_{ij}^2 =\sum_{i=1}^n\sum_{j=1}^n \frac{g_{ij}}{d_i^2} =\sum_{i=1}^n \frac{d_i}{d_i^2} =\sum_{i=1}^n \frac{1}{d_i}.$$ Moreover, $$\Tr(W^2)=\sum_{i=1}^n (W^2)_{ii}=\sum_{i=1}^n\sum_{j=1}^n w_{ij}w_{ji} =\sum_{i=1}^n\sum_{j=1}^n \frac{g_{ij}g_{ji}}{d_i d_j}.$$ If $G=G^\top$, then each undirected edge $\{i,j\}$ contributes twice to the double sum, which yields $\Tr(W^2)=2\sum_{\{i,j\}\in E}(d_i d_j)^{-1}$. If, further, $d_i=d$ for all $i$, then $\sum_i d_i^{-1}=n/d$ and $2\sum_{\{i,j\}\in E}(d_i d_j)^{-1}=2|E|/d^2=(nd)/d^2=n/d$, since $|E|=nd/2$. This gives $\Tr(W^\top W)+\Tr(W^2)=2n/d$.
\end{proof}

\subsection{A Spectral Heuristic for the RESAPLE Sampling Distribution}
To discuss heuristically how spectral structure can affect the finite-sample distribution of RESAPLE, it is useful to recall the corresponding fact for an ordinary Rayleigh quotient involving a symmetric matrix. Although the null law of RESAPLE is, in general, a \emph{generalised} Rayleigh quotient rather than the ordinary quotient treated below, the result gives useful intuition for why eigenvalue dispersion and multiplicity can influence finite-sample behaviour. We therefore describe how the eigenvalues of a symmetric matrix control the sampling law in this simpler benchmark case.

\begin{thm}\label{thm:moran-spectrum-law}
Let $g\sim N(0,I_r)$ with $r\ge 3$, and let $S\in\mathbb R^{r\times r}$ be symmetric with eigenvalues $\lambda_1,\dots,\lambda_r$. Define the Rayleigh quotient $T(g)\defeq \frac{g^\top S g}{g^\top g}$. Then the distribution of $T(g)$ depends on $S$ only through $(\lambda_1,\dots,\lambda_r)$. Moreover, $$\mathbb E[T(g)] = \frac{\Tr(S)}{r},\qquad \Var\big(T(g)\big) = \frac{2}{r(r+2)}\sum_{j=1}^r(\lambda_j-\bar\lambda)^2, \quad \bar\lambda=\frac{1}{r}\sum_{j=1}^r\lambda_j.$$ In particular, $\Var(T(g))=0$ if and only if $S$ is a scalar multiple of the identity. 
\end{thm}

\begin{proof}
Write $S=Q\Lambda Q^\top$ with $\Lambda=\diag(\lambda_1,\dots,\lambda_r)$ and $Q$ orthogonal. Set $h\defeq Q^\top g$. Since $g\sim N(0,I_r)$ and $Q$ is orthogonal, we have $h\stackrel{d}{=}g$; hence $$T(g)=\frac{g^\top Q\Lambda Q^\top g}{g^\top g}=\frac{h^\top \Lambda h}{h^\top h}=\sum_{j=1}^r \lambda_j\frac{h_j^2}{\sum_{\ell=1}^r h_\ell^2}.$$ Thus the law of $T(g)$ depends on $S$ only through its eigenvalues.

Now let $$\pi_j \defeq \frac{h_j^2}{\sum_{\ell=1}^r h_\ell^2},\qquad j=1,\dots,r.$$ Because $h_1,\dots,h_r$ are iid standard Normal, the vector $(\pi_1,\dots,\pi_r)$ has a Dirichlet$\big(\frac{1}{2},\dots,\frac{1}{2}\big)$ distribution, which yields $$\mathbb E[\pi_i]=\frac1r,\qquad \mathbb E[\pi_i^2]=\frac{3}{r(r+2)},\qquad \mathbb E[\pi_i\pi_j]=\frac{1}{r(r+2)}\ (i\neq j).$$ Since $T(g)=\sum_{i=1}^r \lambda_i \pi_i$, we obtain $$\mathbb E[T(g)] = \sum_{i=1}^r \lambda_i\mathbb E[\pi_i]= \frac{1}{r}\sum_{i=1}^r\lambda_i= \frac{\Tr(S)}{r},$$ and 

\begin{align*}
    \mathbb E[T(g)^2] &= \sum_{i=1}^r \lambda_i^2\mathbb E[\pi_i^2] + \sum_{i\neq j}\lambda_i\lambda_j\mathbb E[\pi_i\pi_j] \\
    &= \frac{3}{r(r+2)}\sum_{i=1}^r \lambda_i^2 + \frac{1}{r(r+2)}\sum_{i\neq j}\lambda_i\lambda_j=\frac{1}{r(r+2)}\Big(2\sum_{i=1}^r\lambda_i^2 + \big(\sum_{i=1}^r\lambda_i\big)^2\Big) \\
    &=\frac{1}{r(r+2)}\Big(2\Tr(S^2)+\Tr(S)^2\Big).
\end{align*}
Subtracting $\mathbb E[T(g)]^2=\Tr(S)^2/r^2$ gives $$\Var(T(g))=\frac{2}{r(r+2)}\Big(\Tr(S^2)-\frac{\Tr(S)^2}{r}\Big),$$ and the equivalent eigenvalue form follows by expanding $\sum_j(\lambda_j-\bar\lambda)^2$. Finally, $\Var(T(g))=0$ holds if and only if $\lambda_1=\cdots=\lambda_r$, that is, $S$ is a scalar multiple of $I_r$.
\end{proof}

The theorem does not directly describe the null law of RESAPLE, but it does provide a useful spectral heuristic. In particular, when the relevant residualised operators have highly dispersed eigenvalues or large multiplicities, one should expect stronger finite-sample departures from a Gaussian approximation. This intuition is consistent with the behaviour we observe for small $r$ and irregular neighbourhood structures, and it helps motivate the use of exact or permutation-based calibration in such settings.

\section{Additional Simulation Details and Supplementary Results}\label{app:simulations}

\subsection{Complete Simulation Design}\label{app:sim-design}

This appendix gives a fuller count of the simulation used in Section~\ref{sec:simulations}. The response is generated from the spatial error model
$$Z=X\beta+U,\qquad U=(I_n-\rho W)^{-1}\epsilon.$$
The matrix $W$ is row-standardised in every experiment. The dependence parameter is varied over
$$\rho\in\{0,0.05,\ldots,0.95\}.$$
The noise scale is fixed at $\sigma=1$. Unless stated otherwise, each design point uses $1000$ Monte Carlo replicates.

The design matrix is fixed within each design point. It always includes an intercept. If $p\geq 2$, the second column is the first spatial coordinate plus $0.10$ times an independent standard normal perturbation, then standardised to mean zero and unit variance. If $p\geq 3$, the third column is constructed in the same way from the second spatial coordinate. If $p\geq 4$, the remaining columns are independent standard normal variables, each standardised column-wise. The design is retained only when $p<n$ and $X$ has full column rank. The coefficient vector is
$$\beta_1=1,\qquad \beta_j=\frac{0.6}{\sqrt{j-1}},\quad j\geq 2.$$

Three graph classes are used. The first is a regular $m\times m$ lattice, with $n=m^2$ and $n\in\{25,100,400\}$. The lattice candidate weights are
$$\mathcal W_{\mathrm{lat}}=\{\mathrm{rook},\mathrm{queen},\mathrm{knn4},\mathrm{knn6},\mathrm{knn8}\}.$$
Rook contiguity connects horizontal and vertical neighbours. Queen contiguity also includes diagonal neighbours. The $k$ nearest-neighbour graphs use Euclidean distance on the lattice coordinates, with $k\in\{4,6,8\}$. They are first formed as directed nearest-neighbour graphs, then symmetrised by taking the union of directed neighbour relations. The resulting adjacency matrix is row-standardised.

The second graph is B07. It is a fixed eight-node binary adjacency graph. The matrix is row-standardised before use. The display coordinates and the coordinate-based covariates are constructed from the eight equally spaced points on the unit circle. The binary adjacency rows are
\begin{table}[pos=!htpb]
\centering
\footnotesize
\renewcommand{\arraystretch}{1.08}
\setlength{\tabcolsep}{5pt}
\begin{tabular}{c l}
\toprule
Node & Adjacency row \\
\midrule
1 & 0 1 1 1 1 1 0 1 \\
2 & 1 0 1 1 1 1 0 0 \\
3 & 1 1 0 0 0 1 1 1 \\
4 & 1 1 0 0 1 0 0 0 \\
5 & 1 1 0 1 0 1 0 0 \\
6 & 1 1 1 0 1 0 1 1 \\
7 & 0 0 1 0 0 1 0 1 \\
8 & 1 0 1 0 0 1 1 0 \\
\bottomrule
\end{tabular}
\caption{Binary adjacency rows for the B07 graph.}
\label{tab:app-b07-adjacency}
\end{table}

The third graph is a fixed large irregular planar graph with $n=128$. It is generated once using a fixed seed. The initial graph is a triangle with coordinates $(0.05,0.05)$, $(0.95,0.05)$, and $(0.50,0.95)$. At each step, one triangular face is selected with probability proportional to the square root of its area. A new point is placed inside that face using three independent Gamma$(2.5,1)$ variables normalised to sum to one. The new vertex is connected to the three vertices of the selected face. The selected face is then replaced by the three induced subfaces. This procedure continues until $128$ vertices have been generated. The resulting adjacency matrix is row-standardised.

The simulation factors are summarised in Table~\ref{tab:app-simulation-design}. The lattice experiments use the full candidate set when selecting $W$ by restricted information. The selected lattice weight in the main text is rook contiguity. The appendix retains the remaining lattice weights where needed. B07 and the large irregular graph use their fixed adjacency matrices throughout.

\begin{table*}[pos=!htpb]
\centering
\footnotesize
\renewcommand{\arraystretch}{1.15}
\setlength{\tabcolsep}{5pt}
\begin{tabular*}{\textwidth}{@{\extracolsep{\fill}} lllll @{}}
\toprule
Graph class & $n$ & $p$ & Candidate or fixed $W$ & Main use \\
\midrule
Lattice & $\{25,100,400\}$ & $\{1,5,20\}$ & rook, queen, knn4, knn6, knn8 & selection, estimation, testing, robustness \\
B07 & $8$ & $\{1,3,5\}$ & adjacency & small irregular stress test \\
Large irregular & $128$ & $\{1,3,5\}$ & adjacency & larger irregular benchmark \\
\bottomrule
\end{tabular*}
\caption{Simulation design factors. The lattice candidate weights are compared using $\mathcal I_r(0)$. B07 and the large irregular graph use fixed adjacency matrices.}
\label{tab:app-simulation-design}
\end{table*}

The Gaussian simulations use $\epsilon\sim N(0,I_n)$. The robustness simulation also considers two non-Gaussian error laws. The first is a spherical multivariate $t$ error with $\nu=5$ degrees of freedom,
$$\epsilon=\sqrt{\frac{\nu-2}{Q}}Z,\qquad Z\sim N(0,I_n),\qquad Q\sim\chi^2_\nu.$$
The same scalar $Q$ is shared across all coordinates, so the error vector is spherical and has marginal variance one. The second is an independent skewed error,
$$\epsilon_i=\frac{Y_i-3}{\sqrt{6}},\qquad Y_i\sim\chi^2_3.$$
This error has mean zero and variance one, but is not spherical.

For testing, the null hypothesis is $H_0:\rho=0$, and the alternative is $H_1:\rho>0$. The nominal level is $\alpha=0.05$. The permutation tests use the Freedman-Lane residual procedure with $B=199$ random permutations per replicate. The permutation p-value is computed with the standard add-one correction. The RESAPLE $z$ procedure uses the statistic $\sqrt{\mathcal I_r(0)}\hat\rho_{\mathrm{RESAPLE}}$ and the upper standard normal tail. The exact RESAPLE procedure uses the Gaussian residual quadratic-form reference distribution described in Section~\ref{subsec:RESAPLE-testing}.

For estimation, we compare RESAPLE with residual Moran's $\mathcal I_M$, residual APLE, MAPLE, and REML. REML is included as a likelihood benchmark. It is computed by profiling the restricted Gaussian likelihood over $\rho\in[-0.95,0.95]$. For each estimator $\hat\rho$, the reported summaries are
$$\operatorname{Bias}(\hat\rho)=\mathbb E(\hat\rho-\rho),\qquad \operatorname{SD}(\hat\rho)=\{\operatorname{Var}(\hat\rho)\}^{1/2},\qquad \operatorname{RMSE}(\hat\rho)=\{\mathbb E(\hat\rho-\rho)^2\}^{1/2}.$$

\subsection{Full Weight-Selection Results}\label{app:W-selection}

This subsection gives the full lattice weight-selection results used in Simulation~\ref{subsec:sim-W-selection}. For each lattice design, we compute $\mathcal I_r(0)$ for $W\in\{\mathrm{rook},\mathrm{queen},\mathrm{knn4},\mathrm{knn6},\mathrm{knn8}\}$. The selected weight maximises $\mathcal I_r(0)$ within the design. Empirical power at $\rho=0.3$ is then used only to assess whether this ordering corresponds to detectability.

Table~\ref{tab:app-W-selection-selected} reports the selected weight for each $(n,p)$ design. Rook contiguity is selected in every case. This is because it has the largest restricted information for every design, despite having the lowest average degree among the lattice candidates. The result reflects the fact that denser weights do not necessarily give more restricted information after residualisation.

\begin{table}[pos=!htpb]
\centering
\footnotesize
\renewcommand{\arraystretch}{1.12}
\setlength{\tabcolsep}{5pt}
\begin{tabular*}{\columnwidth}{@{\extracolsep{\fill}} rrrr @{}}
\toprule
$n$ & $p=1$ & $p=5$ & $p=20$ \\
\midrule
25  & rook, $14.17$  & rook, $9.19$   & rook, $0.74$ \\
100 & rook, $54.88$  & rook, $49.03$  & rook, $35.05$ \\
400 & rook, $211.27$ & rook, $205.34$ & rook, $190.23$ \\
\bottomrule
\end{tabular*}
\caption{Selected lattice weights by restricted information. Each entry gives the selected $W$ and the corresponding $\mathcal I_r(0)$.}
\label{tab:app-W-selection-selected}
\end{table}

Table~\ref{tab:app-W-selection-ranking} reports the complete ranking of the candidate weights. The ordering by $\mathcal I_r(0)$ is stable across designs. Rook is always first, followed by knn4. The empirical power ranking at $\rho=0.3$ is generally consistent with the information ranking. The small exceptions occur when power is either very low, as for $n=25,p=20$, or close to one, as for $n=400$. In both cases, the empirical differences are small.

\begin{table*}[pos=!htpb]
\centering
\scriptsize
\renewcommand{\arraystretch}{1.12}
\setlength{\tabcolsep}{4pt}
\begin{tabular*}{\textwidth}{@{\extracolsep{\fill}} rrlllc @{}}
\toprule
$n$ & $p$ & Ranking by $\mathcal I_r(0)$ & Ranking by power at $\rho=0.3$ & Selected $W$ & Selected power \\
\midrule
25  & 1  & rook $>$ knn4 $>$ queen $>$ knn6 $>$ knn8 & rook $>$ queen $>$ knn4 $>$ knn6 $>$ knn8 & rook & 0.285 \\
25  & 5  & rook $>$ knn4 $>$ queen $>$ knn6 $>$ knn8 & rook $>$ queen $>$ knn4 $>$ knn6 $>$ knn8 & rook & 0.212 \\
25  & 20 & rook $>$ knn4 $>$ queen $>$ knn6 $>$ knn8 & knn4 $>$ rook $>$ knn8 $>$ queen $>$ knn6 & rook & 0.084 \\
100 & 1  & rook $>$ knn4 $>$ queen $>$ knn6 $>$ knn8 & rook $>$ knn4 $>$ queen $>$ knn6 $>$ knn8 & rook & 0.724 \\
100 & 5  & rook $>$ knn4 $>$ queen $>$ knn6 $>$ knn8 & rook $>$ knn4 $>$ queen $>$ knn6 $>$ knn8 & rook & 0.665 \\
100 & 20 & rook $>$ knn4 $>$ queen $>$ knn6 $>$ knn8 & rook $>$ knn4 $>$ knn6 $>$ queen $>$ knn8 & rook & 0.529 \\
400 & 1  & rook $>$ knn4 $>$ queen $>$ knn6 $>$ knn8 & rook $>$ knn4 $>$ queen $=$ knn6 $>$ knn8 & rook & 0.998 \\
400 & 5  & rook $>$ knn4 $>$ queen $>$ knn6 $>$ knn8 & rook $>$ knn4 $>$ queen $>$ knn6 $>$ knn8 & rook & 0.995 \\
400 & 20 & rook $>$ knn4 $>$ queen $>$ knn6 $>$ knn8 & knn4 $>$ rook $>$ queen $>$ knn6 $>$ knn8 & rook & 0.992 \\
\bottomrule
\end{tabular*}
\caption{Complete lattice weight rankings. Power is the empirical rejection rate of the RESAPLE exact test at $\rho=0.3$. The selected weight is chosen by $\mathcal I_r(0)$, not by empirical power.}
\label{tab:app-W-selection-ranking}
\end{table*}

Table~\ref{tab:app-W-selection-size} reports the empirical size of the selected exact tests. The nominal level is $\alpha=0.05$. The selected designs remain close to nominal size. The largest absolute departure is $0.017$, occurring at $n=25,p=5$. This is within the range expected for a small residual space with $r=20$ and $1000$ Monte Carlo replicates.

\begin{table}[pos=!htpb]
\centering
\footnotesize
\renewcommand{\arraystretch}{1.12}
\setlength{\tabcolsep}{5pt}
\begin{tabular*}{\columnwidth}{@{\extracolsep{\fill}} rrrrrr @{}}
\toprule
$n$ & $p$ & $r$ & $W$ & Size & $|\widehat{\alpha}-0.05|$ \\
\midrule
25  & 1  & 24  & rook & 0.051 & 0.001 \\
25  & 5  & 20  & rook & 0.067 & 0.017 \\
25  & 20 & 5   & rook & 0.058 & 0.008 \\
100 & 1  & 99  & rook & 0.047 & 0.003 \\
100 & 5  & 95  & rook & 0.046 & 0.004 \\
100 & 20 & 80  & rook & 0.043 & 0.007 \\
400 & 1  & 399 & rook & 0.063 & 0.013 \\
400 & 5  & 395 & rook & 0.055 & 0.005 \\
400 & 20 & 380 & rook & 0.058 & 0.008 \\
\bottomrule
\end{tabular*}
\caption{Empirical size of the RESAPLE exact test using the selected lattice weights.}
\label{tab:app-W-selection-size}
\end{table}

Figures~\ref{subfig:sim-W-selection-n25}, \ref{subfig:sim-W-selection-n100}, and \ref{subfig:sim-W-selection-n400} show the corresponding information-power plots. These figures show the same pattern as the tables. Larger $\mathcal I_r(0)$ generally corresponds to higher power at $\rho=0.3$, while the null rejection rates remain close to $\alpha=0.05$. The simulations therefore use rook contiguity for the main lattice results.

\subsection{Gaussian Estimation Supplement}\label{app:gaussian-estimation-supplement}

This subsection gives the full Gaussian estimation supplement for the selected lattice weight. The main text reports the central RMSE result for the $n=100$ lattice. We therefore report the remaining RMSE plots here, together with the full bias and standard deviation plots for $n\in\{25,100,400\}$ under rook contiguity.

The purpose of this supplement is to separate the two components of RMSE. RESAPLE has low RMSE on the lattice because it has low bias across most of the parameter range. It does not uniformly minimise sampling standard deviation. This is expected because Moran-type estimators can have low variance while remaining strongly biased for $\rho$. The relevant comparison is therefore RMSE, not variance alone.

\begin{figure*}[pos=!htpb]
\centering

\begin{subfigure}[t]{0.7\textwidth}
\centering
\includegraphics[width=0.86\textwidth]{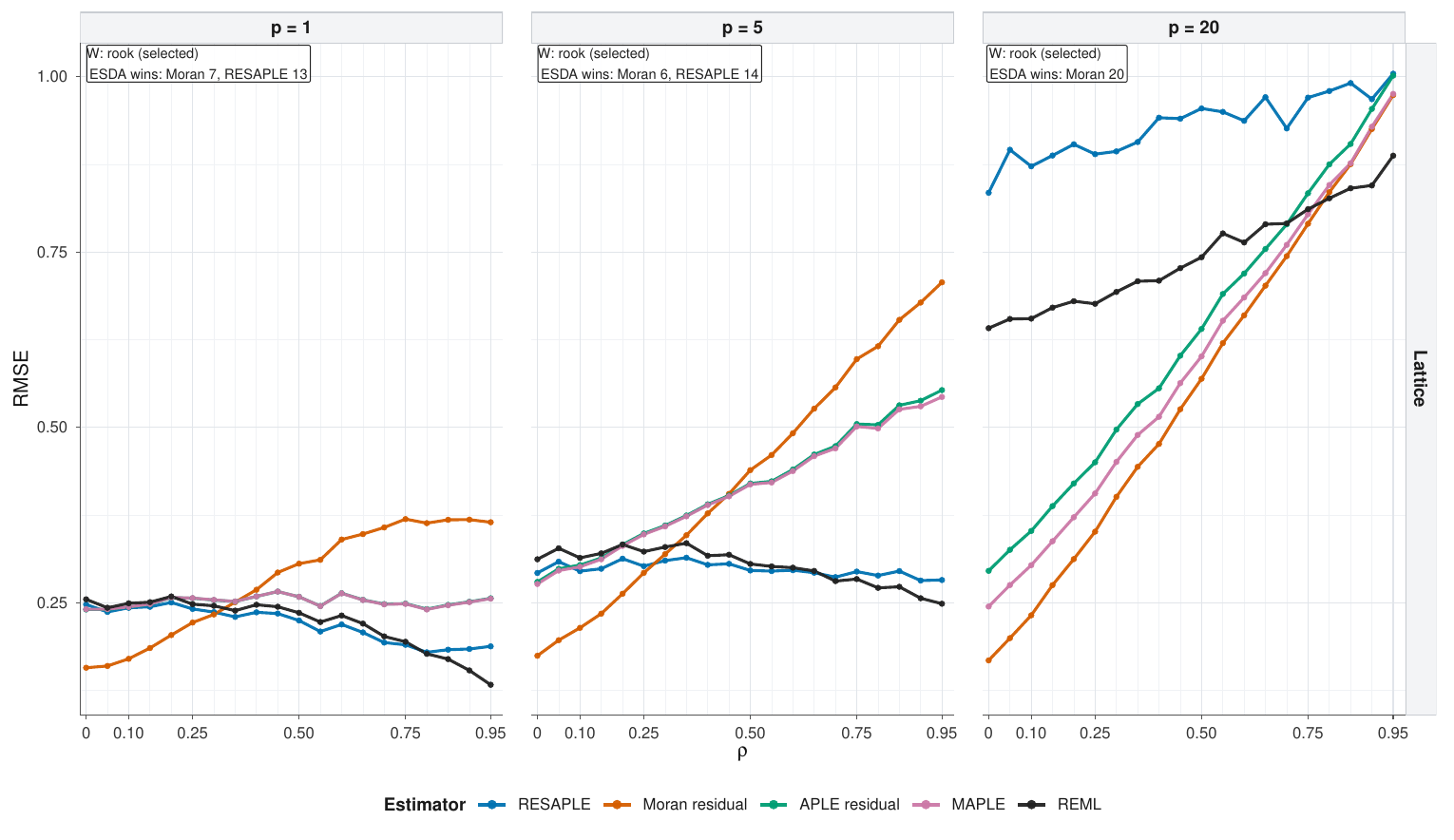}
\caption{$n=25$.}
\label{subfig:app-sim1-rmse-rook-n25}
\end{subfigure}

\vspace{0.8em}

\begin{subfigure}[t]{0.7\textwidth}
\centering
\includegraphics[width=0.86\textwidth]{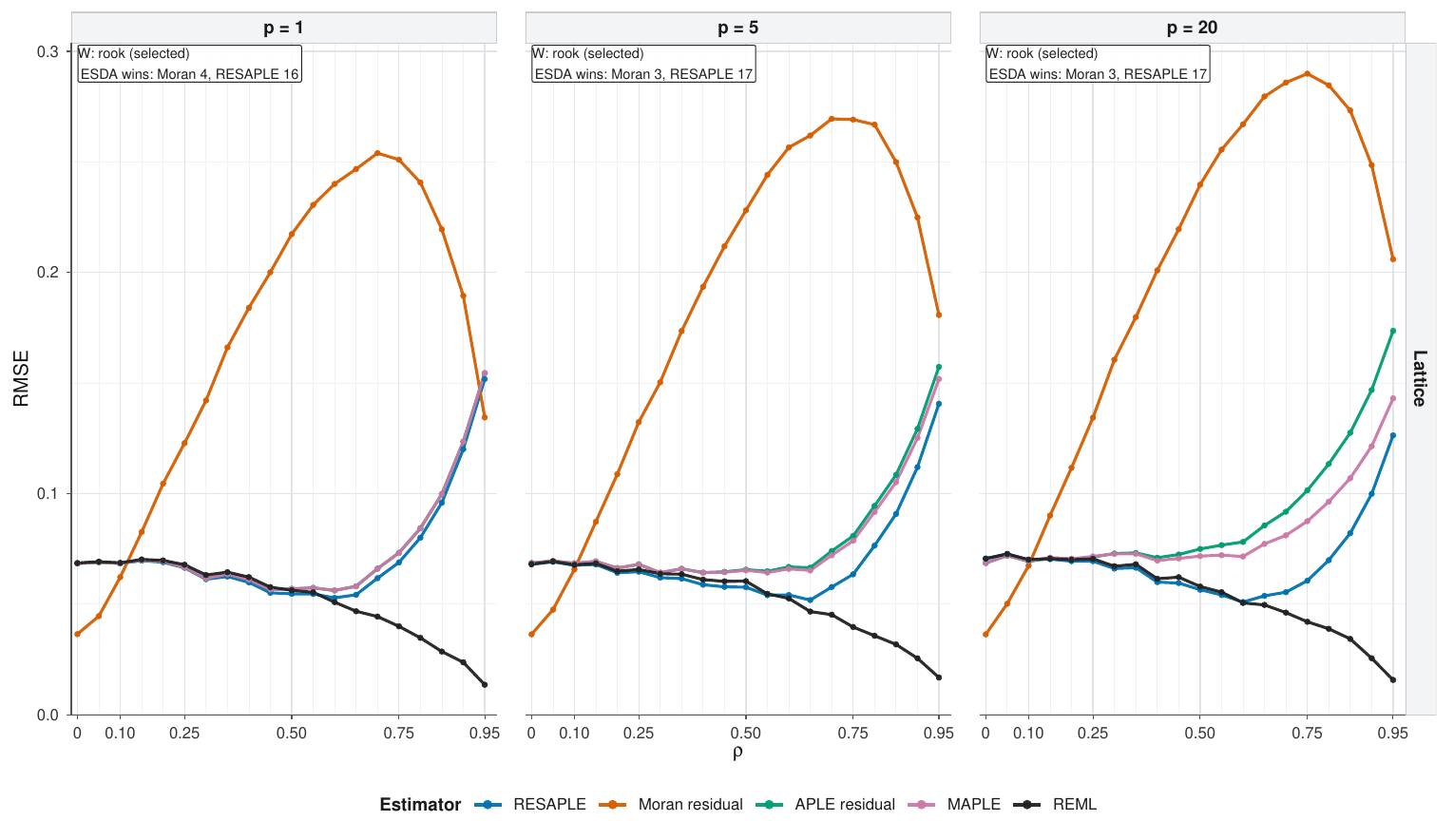}
\caption{$n=400$.}
\label{subfig:app-sim1-rmse-rook-n400}
\end{subfigure}

\caption{Gaussian estimation RMSE on the lattice using the selected weight $W=\mathrm{rook}$. The $n=100$ panel is shown in the main text and is not repeated here.}
\label{fig:app-sim1-rmse-rook}
\end{figure*}

Figure~\ref{fig:app-sim1-rmse-rook} completes the lattice RMSE results for the selected weight. For $n=25$, RESAPLE has the lowest RMSE across most values of $\rho$. Its advantage is clearest when $p=5$ or $p=20$, where residualisation makes the bias of Moran-type estimators more visible. For $n=400$, RESAPLE remains competitive across the full grid, and is strongest when $p$ is moderate or large. The remaining differences are small because the estimators become closer as the residual dimension increases.

\begin{figure*}[pos=!htpb]
\centering

\begin{subfigure}[t]{0.7\textwidth}
\centering
\includegraphics[width=0.86\textwidth]{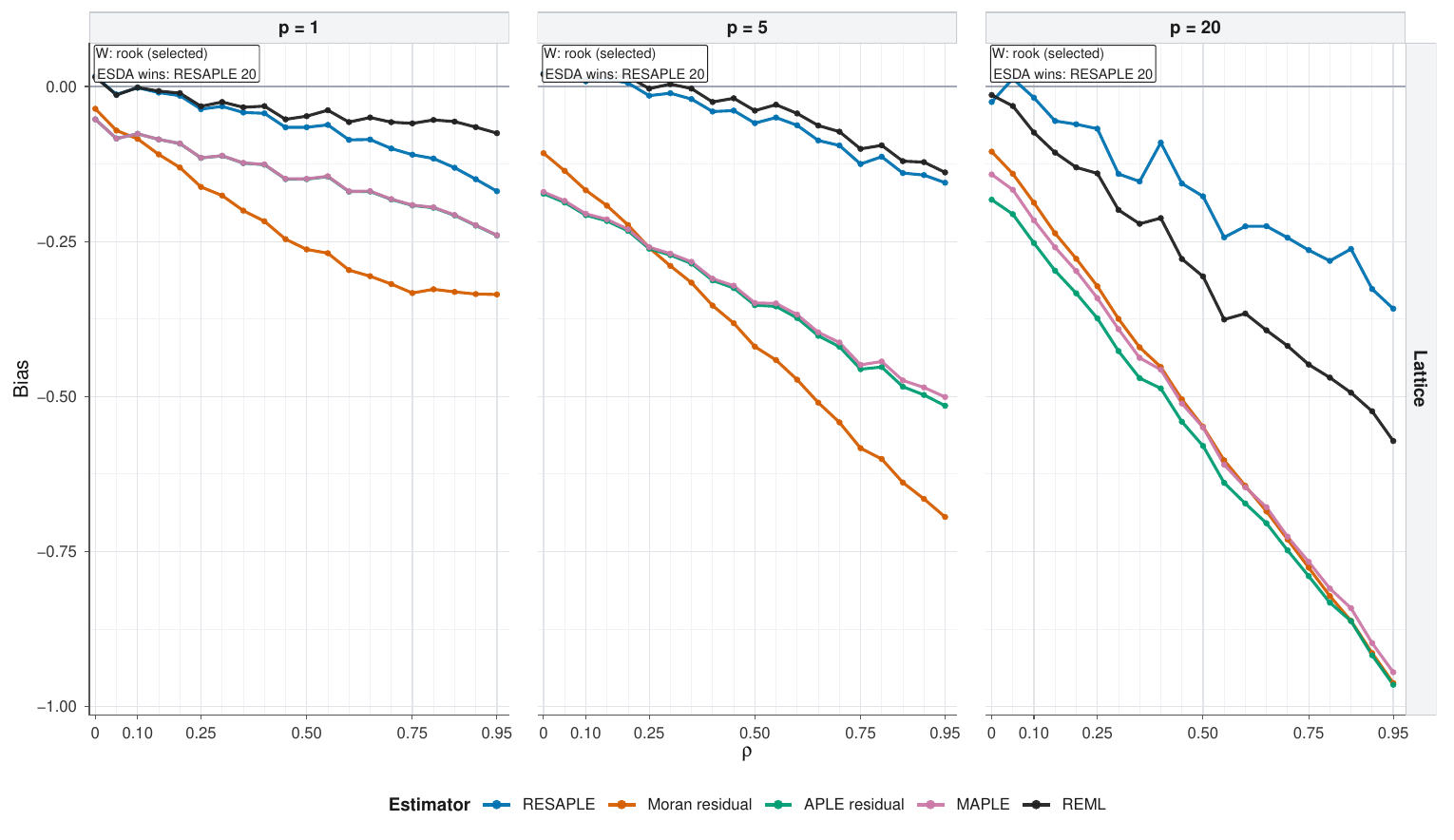}
\caption{$n=25$.}
\label{subfig:app-sim1-bias-rook-n25}
\end{subfigure}

\vspace{0.8em}

\begin{subfigure}[t]{0.7\textwidth}
\centering
\includegraphics[width=0.86\textwidth]{figs/fig_resaple_sim1_bias_gaussian_n100_Wrook.pdf}
\caption{$n=100$.}
\label{subfig:app-sim1-bias-rook-n100}
\end{subfigure}

\vspace{0.8em}

\begin{subfigure}[t]{0.7\textwidth}
\centering
\includegraphics[width=0.86\textwidth]{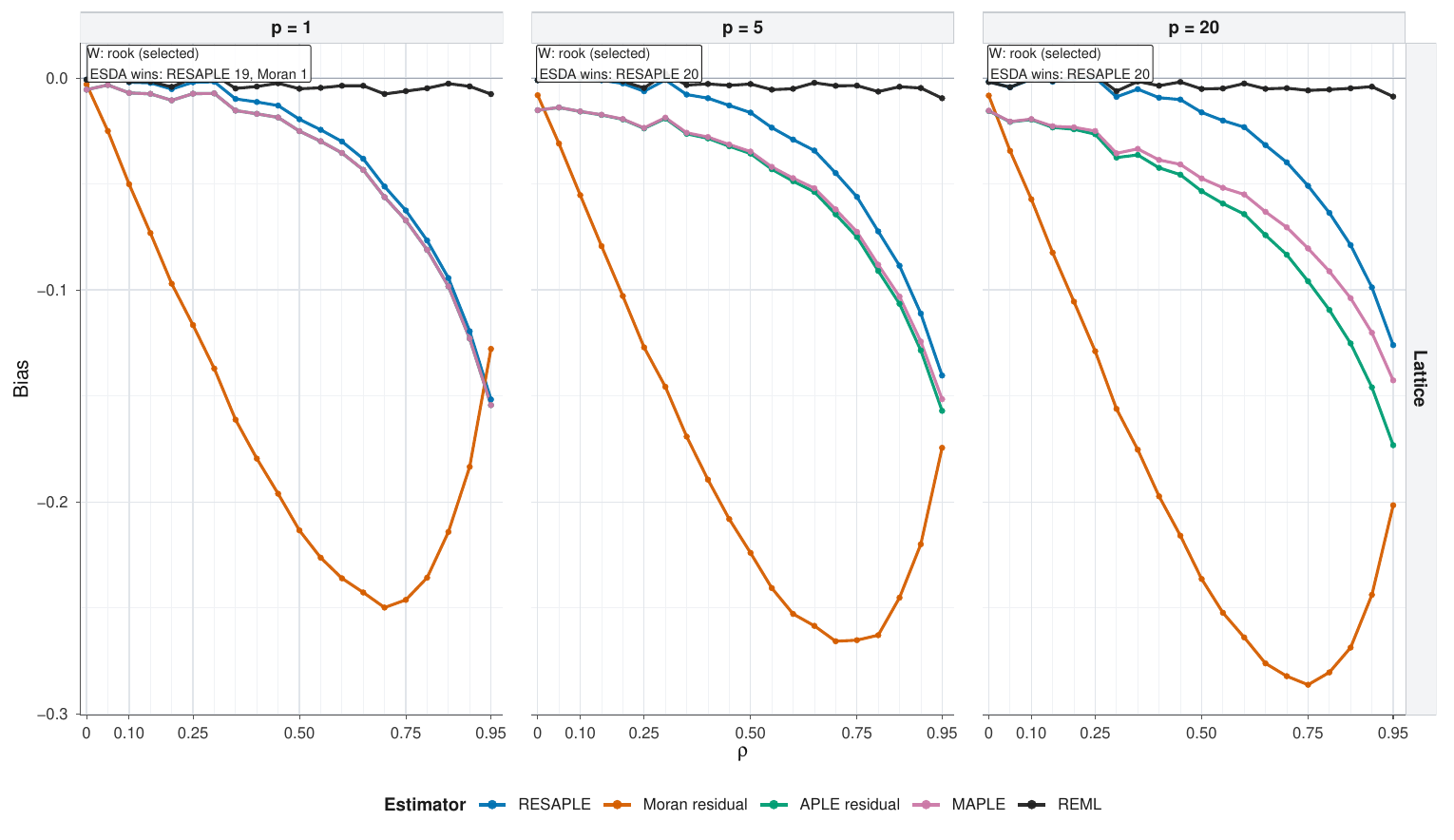}
\caption{$n=400$.}
\label{subfig:app-sim1-bias-rook-n400}
\end{subfigure}

\caption{Gaussian estimation bias on the lattice using the selected weight $W=\mathrm{rook}$.}
\label{fig:app-sim1-bias-rook}
\end{figure*}

Figure~\ref{fig:app-sim1-bias-rook} shows that RESAPLE is consistently closer to zero bias than the residual Moran estimator, APLE residual estimator, and MAPLE across the lattice designs. The difference is largest for moderate and large $\rho$. Moran residual has particularly strong negative bias as $\rho$ increases. APLE residual and MAPLE reduce this effect in some settings, but they still inherit substantial downward bias relative to RESAPLE. REML is also close to zero in some panels, as expected from a likelihood-based estimator, but it is included here as a benchmark rather than as an ESDA statistic.

\begin{figure*}[pos=!htpb]
\centering

\begin{subfigure}[t]{0.7\textwidth}
\centering
\includegraphics[width=0.86\textwidth]{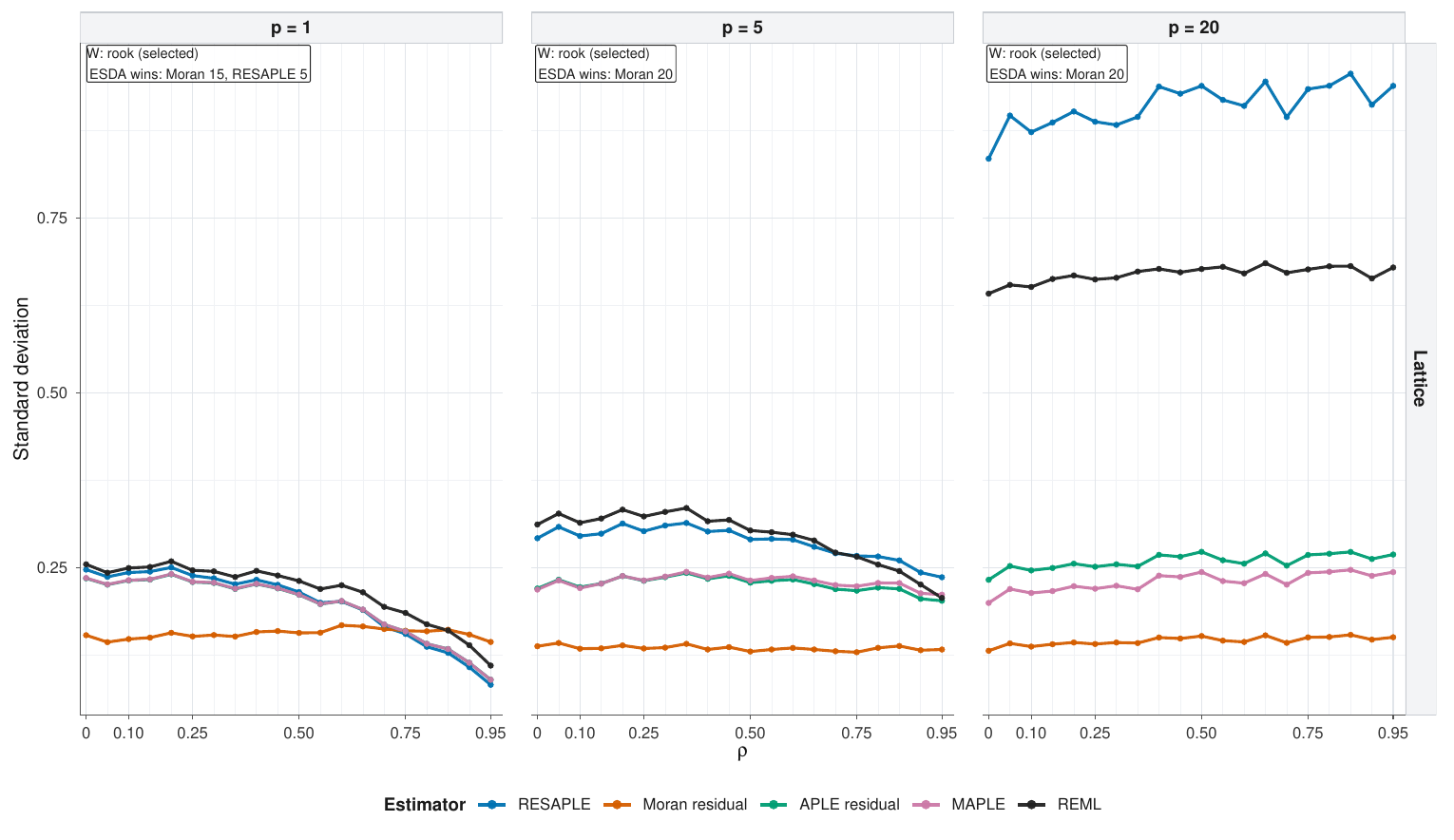}
\caption{$n=25$.}
\label{subfig:app-sim1-sd-rook-n25}
\end{subfigure}

\vspace{0.8em}

\begin{subfigure}[t]{0.7\textwidth}
\centering
\includegraphics[width=0.86\textwidth]{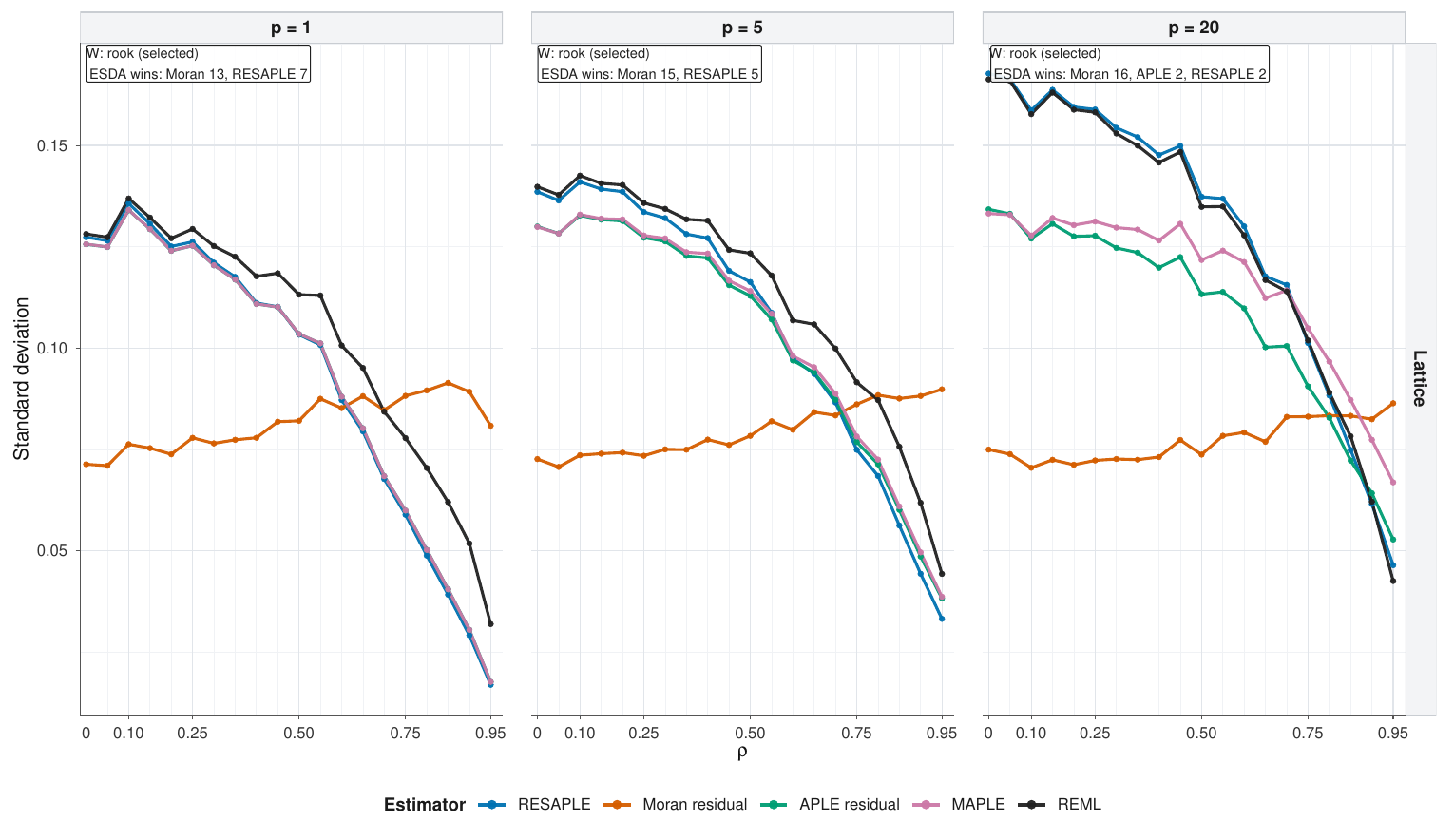}
\caption{$n=100$.}
\label{subfig:app-sim1-sd-rook-n100}
\end{subfigure}

\vspace{0.8em}

\begin{subfigure}[t]{0.7\textwidth}
\centering
\includegraphics[width=0.86\textwidth]{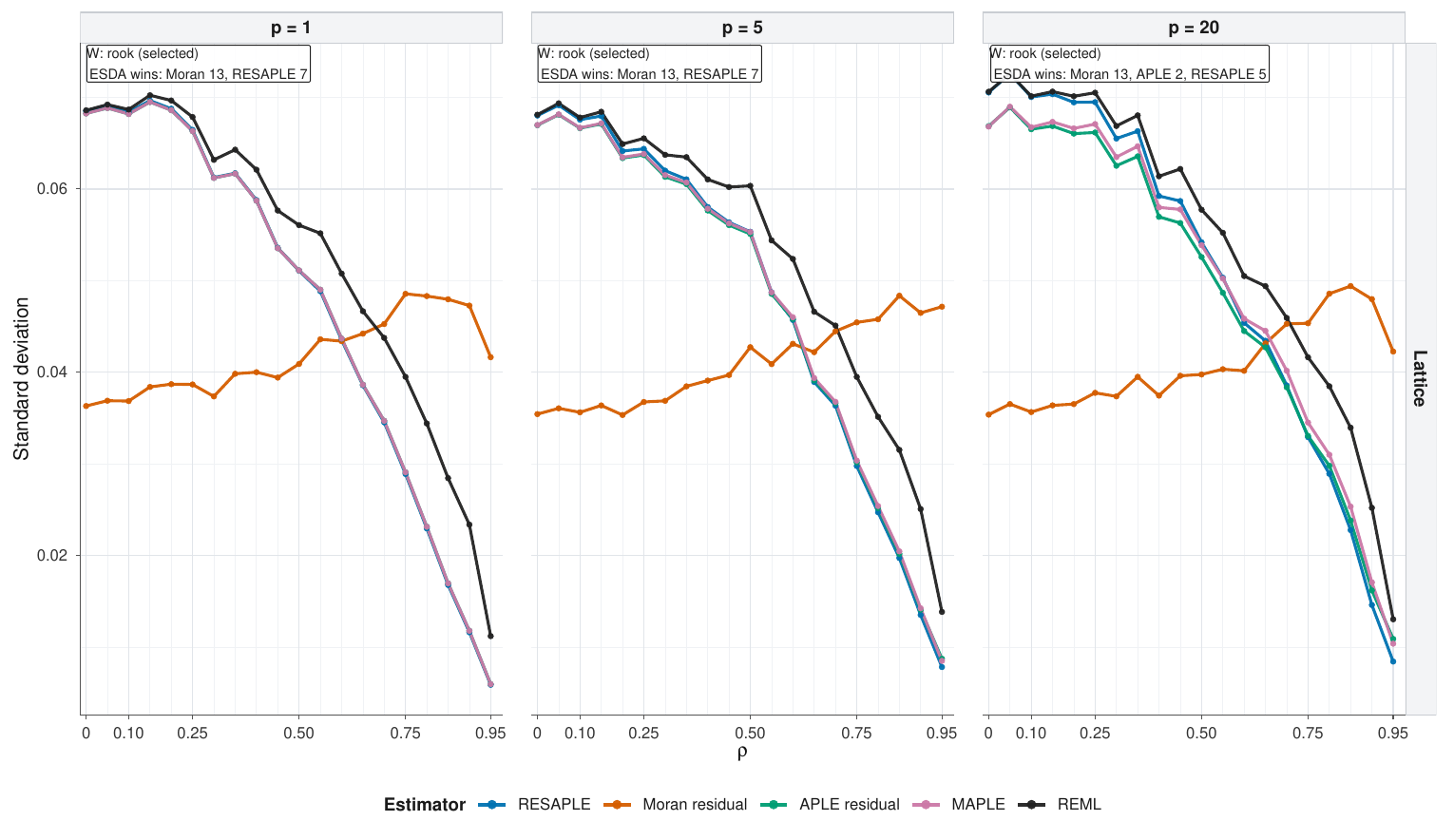}
\caption{$n=400$.}
\label{subfig:app-sim1-sd-rook-n400}
\end{subfigure}

\caption{Gaussian estimation standard deviation on the lattice using the selected weight $W=\mathrm{rook}$.}
\label{fig:app-sim1-sd-rook}
\end{figure*}

Figure~\ref{fig:app-sim1-sd-rook} shows that RESAPLE does not always have the lowest standard deviation. In several low-dimensional settings, Moran residual has lower sampling variability. This does not contradict the RMSE results, because the same estimator also has much larger bias. RESAPLE trades some sampling variability for a much smaller systematic error. This is the main reason its RMSE is usually lower.

The same qualitative pattern appears for the other lattice candidates in $\{\mathrm{queen},\mathrm{knn4},\mathrm{knn6},\mathrm{knn8}\}$. Those weights are not repeated here because Simulation~\ref{subsec:sim-W-selection} selects rook by $\mathcal I_r(0)$ for every lattice design.

\newpage

\subsection{Gaussian Testing Supplement}\label{app:gaussian-testing-supplement}

This subsection gives the Gaussian testing supplement for the selected lattice weight. The main text reports the $n=100$ lattice results, together with B07 and the large irregular graph. We therefore report the $n=25$ lattice results here. The $n=400$ lattice results are not repeated, because they show the same pattern as $n=100$ with higher power and smaller Monte Carlo variation. 

Figure~\ref{fig:app-sim2-permutation-rook-n25} compares the permutation tests at $n=25$. All four tests use the same Freedman-Lane residual permutation procedure. This isolates the effect of the statistic. For $p=1$, the four curves are almost identical. For $p=5$, the curves remain close. For $p=20$, the residual dimension is only $r=5$, and all tests have low power. This is simply a low-information setting rather than a failure of a particular statistic.

\begin{figure*}[pos=!htpb]
\centering
\includegraphics[width=0.86\textwidth]{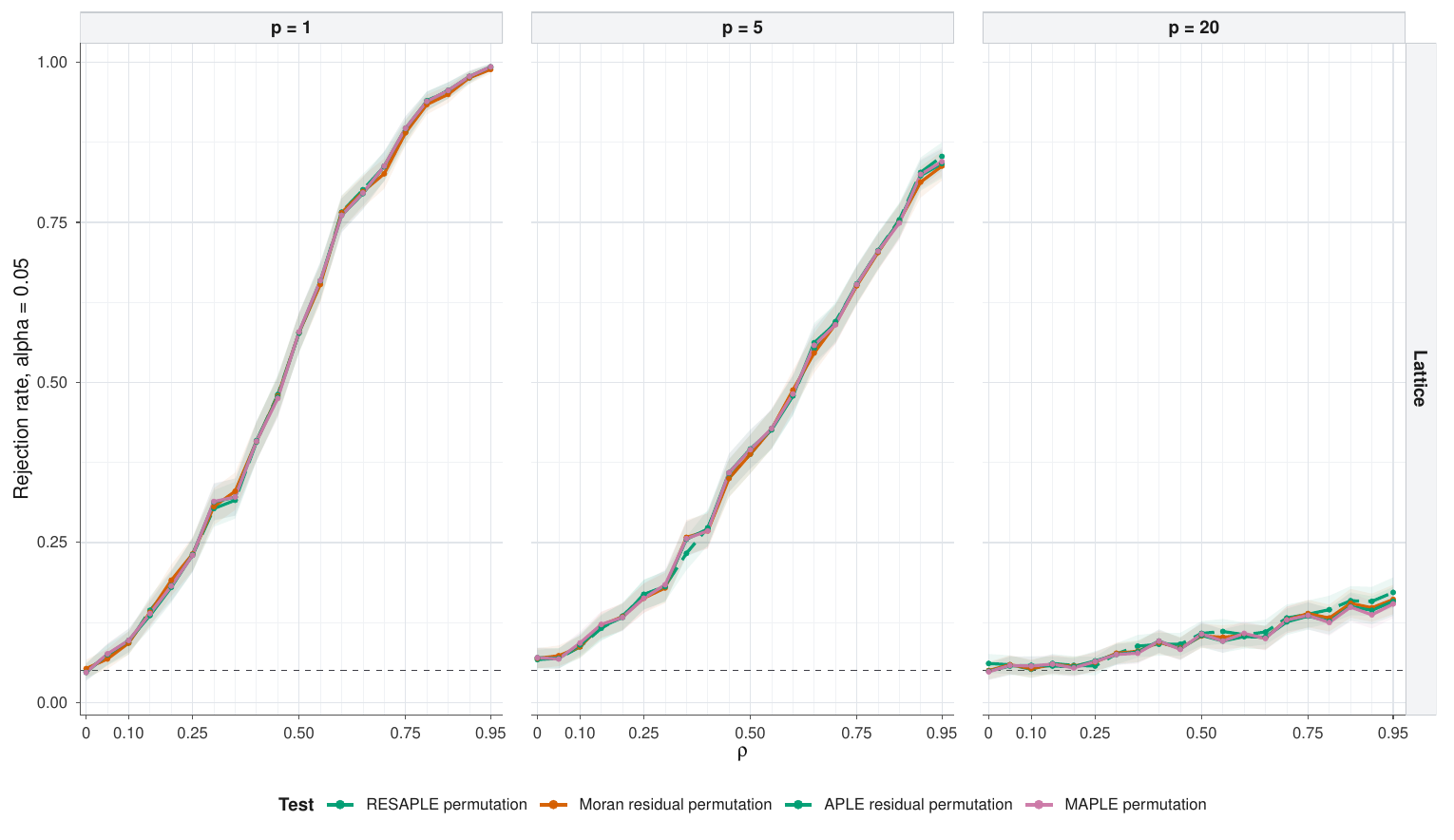}
\caption{Gaussian permutation rejection rates on the $n=25$ lattice using the selected weight $W=\mathrm{rook}$. All tests use the same Freedman-Lane residual permutation procedure. The horizontal dashed line marks $\alpha=0.05$.}
\label{fig:app-sim2-permutation-rook-n25}
\end{figure*}

Figure~\ref{fig:app-sim2-reference-rook-n25} compares the three RESAPLE reference procedures at $n=25$. The exact and permutation procedures are close for $p=1$ and $p=5$. For $p=20$, all procedures have low power because $r=5$. The $z$ procedure is the least stable in this small residual space. This agrees with the main-text results for B07, where small residual dimension and irregular spectra make the normal approximation less reliable.

\begin{figure*}[pos=!htpb]
\centering
\includegraphics[width=0.86\textwidth]{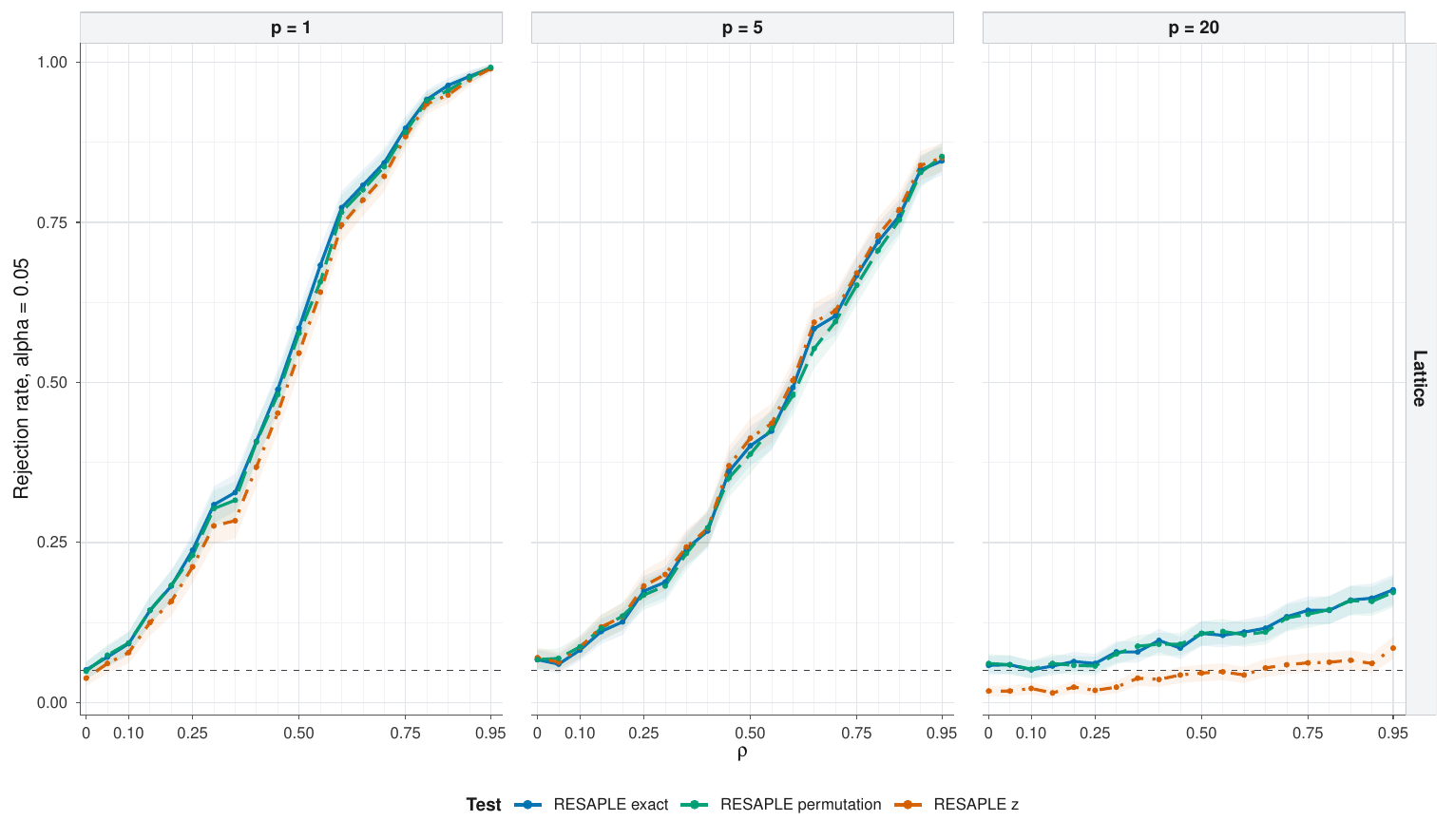}
\caption{Gaussian RESAPLE rejection rates on the $n=25$ lattice using exact, permutation, and $z$ reference procedures. The selected weight is $W=\mathrm{rook}$. The horizontal dashed line marks $\alpha=0.05$.}
\label{fig:app-sim2-reference-rook-n25}
\end{figure*}

The remaining lattice weights, $\mathrm{queen}$, $\mathrm{knn4}$, $\mathrm{knn6}$, and $\mathrm{knn8}$, give the same qualitative conclusion.

\newpage

\section{Additional Details for the Case Study}\label{app:case-study}
\subsection{Map of $Z_i$}

\begin{figure}[pos=!htp]
    \centering
    \includegraphics[width=\linewidth]{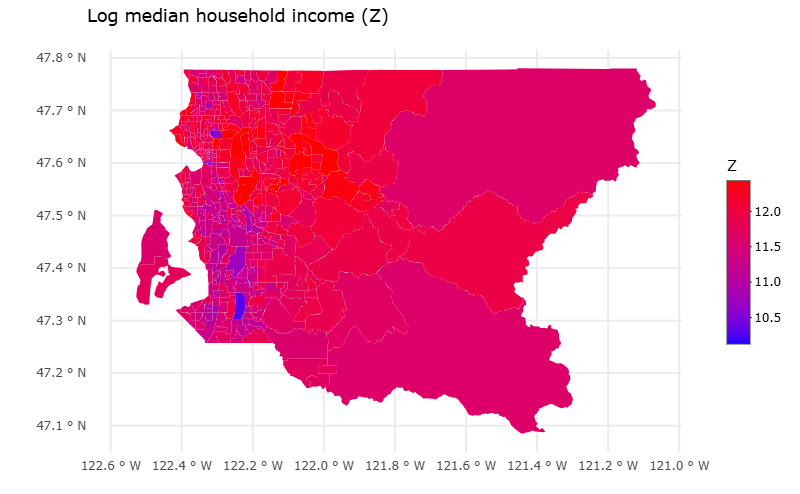}
    \caption{Map of Log Median Household Income.}
    \label{fig:zi-map}
\end{figure}

\newpage

\subsection{Restricted Information for Each $W \in \mathcal W$}

\begin{table}[pos=!htp]
\centering
\small
\setlength{\tabcolsep}{6pt}
\renewcommand{\arraystretch}{1.15}

\caption{Restricted null information across candidate weight matrices $W$, by model in the case study. Here $n=494$ tracts and $n_{\mathrm{iso}}=0$ isolated units for all candidates. For each model, the selected $W$ (maximising $I_r(0)$) is shown in \textbf{bold}. The ratio $\mathrm{info\_ratio}\defeq I_r(0)/I_n(0)$ compares restricted to unrestricted null information, where $I_r(0)=2\operatorname{Tr}(MKMK)$ and $I_n(0)=2\operatorname{Tr}(K^2)$ with $K=(W+W^\top)/2$.}
\label{tab:I0res-byW-appendix}

\vspace{0.6em}

\begin{minipage}[t]{0.325\linewidth}\vspace{0pt}
\centering
\textbf{Model M0} ($p=1$)\\[0.35em]
\begin{tabular}{@{}l r r@{}}
\toprule
$W$ &
\shortstack[c]{avg\_deg\\$I_r(0)$} &
\shortstack[c]{$I_n(0)$\\info\_ratio} \\
\midrule
\textbf{knn4} & 4.000000  & 217.5000  \\
              & 215.36336 & 0.9901764 \\
\addlinespace
rook          & 5.356275  & 194.9016  \\
              & 192.80485 & 0.9892422 \\
\addlinespace
queen         & 6.056680  & 173.4456  \\
              & 171.34649 & 0.9878978 \\
\addlinespace
knn6          & 6.000000  & 146.4444  \\
              & 144.34323 & 0.9856518 \\
\addlinespace
knn8          & 8.000000  & 110.7500  \\
              & 108.66201 & 0.9811468 \\
\bottomrule
\end{tabular}
\end{minipage}\hfill
\begin{minipage}[t]{0.325\linewidth}\vspace{0pt}
\centering
\textbf{Model M1} ($p=5$)\\[0.35em]
\begin{tabular}{@{}l r r@{}}
\toprule
$W$ &
\shortstack[c]{avg\_deg\\$I_r(0)$} &
\shortstack[c]{$I_n(0)$\\info\_ratio} \\
\midrule
\textbf{knn4} & 4.000000  & 217.5000  \\
              & 209.22927 & 0.9619737 \\
\addlinespace
rook          & 5.356275  & 194.9016  \\
              & 187.10276 & 0.9599859 \\
\addlinespace
queen         & 6.056680  & 173.4456  \\
              & 165.90602 & 0.9565308 \\
\addlinespace
knn6          & 6.000000  & 146.4444  \\
              & 139.09374 & 0.9498055 \\
\addlinespace
knn8          & 8.000000  & 110.7500  \\
              & 103.83073 & 0.9375235 \\
\bottomrule
\end{tabular}
\end{minipage}\hfill
\begin{minipage}[t]{0.325\linewidth}\vspace{0pt}
\centering
\textbf{Model M2} ($p=7$)\\[0.35em]
\begin{tabular}{@{}l r r@{}}
\toprule
$W$ &
\shortstack[c]{avg\_deg\\$I_r(0)$} &
\shortstack[c]{$I_n(0)$\\info\_ratio} \\
\midrule
\textbf{knn4} & 4.000000  & 217.5000  \\
              & 206.51133 & 0.9494774 \\
\addlinespace
rook          & 5.356275  & 194.9016  \\
              & 184.71094 & 0.9477140 \\
\addlinespace
queen         & 6.056680  & 173.4456  \\
              & 163.67882 & 0.9436899 \\
\addlinespace
knn6          & 6.000000  & 146.4444  \\
              & 137.05337 & 0.9358728 \\
\addlinespace
knn8          & 8.000000  & 110.7500  \\
              & 102.13158 & 0.9221813 \\
\bottomrule
\end{tabular}
\end{minipage}

\vspace{1.0em}

\begin{minipage}[t]{0.49\linewidth}\vspace{0pt}
\centering
\textbf{Model M3} ($p=9$)\\[0.35em]
\begin{tabular}{@{}l r r@{}}
\toprule
$W$ &
\shortstack[c]{avg\_deg\\$I_r(0)$} &
\shortstack[c]{$I_n(0)$\\info\_ratio} \\
\midrule
\textbf{knn4} & 4.000000  & 217.5000  \\
              & 202.91220 & 0.9329296 \\
\addlinespace
rook          & 5.356275  & 194.9016  \\
              & 181.19303 & 0.9296643 \\
\addlinespace
queen         & 6.056680  & 173.4456  \\
              & 160.20081 & 0.9236374 \\
\addlinespace
knn6          & 6.000000  & 146.4444  \\
              & 133.69818 & 0.9129617 \\
\addlinespace
knn8          & 8.000000  & 110.7500  \\
              & 98.92989  & 0.8932722 \\
\bottomrule
\end{tabular}
\end{minipage}\hfill
\begin{minipage}[t]{0.49\linewidth}\vspace{0pt}
\centering
\textbf{Model M4} ($p=12$)\\[0.35em]
\begin{tabular}{@{}l r r@{}}
\toprule
$W$ &
\shortstack[c]{avg\_deg\\$I_r(0)$} &
\shortstack[c]{$I_n(0)$\\info\_ratio} \\
\midrule
\textbf{knn4} & 4.000000  & 217.5000  \\
              & 198.06346 & 0.9106366 \\
\addlinespace
rook          & 5.356275  & 194.9016  \\
              & 175.98054 & 0.9029201 \\
\addlinespace
queen         & 6.056680  & 173.4456  \\
              & 155.01463 & 0.8937365 \\
\addlinespace
knn6          & 6.000000  & 146.4444  \\
              & 129.15404 & 0.8819320 \\
\addlinespace
knn8          & 8.000000  & 110.7500  \\
              & 94.53845  & 0.8536204 \\
\bottomrule
\end{tabular}
\end{minipage}

\end{table}

\newpage

\printcredits

\bibliographystyle{cas-model2-names}

\bibliography{references}



\end{document}